\documentclass[smallextended]{svjour3}       % onecolumn (second format)

\smartqed  % flush right qed marks, e.g. at end of proof

\usepackage{newtxtext,newtxmath}   % use Times fonts if available on your TeX system

\journalname{Journal}
\def\makeheadbox{} % ARXIV ONLY!
\usepackage{algorithm}
\usepackage[spaceRequire=false]{algpseudocodex}
\algnewcommand{\algorithmicassumption}{\textbf{Requirement:}}
\algnewcommand{\algorithmicsubprocedure}{\textbf{Subprocedure:}}
\makeatletter
\algnewcommand{\Assume}{%
  \algpx@endCodeCommand%
  \algpx@drawInItem{\algorithmicassumption}%
}
\algnewcommand{\Subprocedure}{%
  \algpx@endCodeCommand%
  \algpx@drawInItem{\algorithmicsubprocedure}%
}
\makeatother
\algnewcommand{\InlineIf}[2]{% single line if-then
  \algorithmicif\ #1\ \algorithmicthen\ #2}
\algnewcommand{\InlineElseIf}[2]{% single line else if-then
  \algorithmicelse\ \algorithmicif\ #1\ \algorithmicthen\ #2}
\algnewcommand{\InlineElse}[1]{% single line else
  \algorithmicelse\ #1}
\algnewcommand{\InlineIfElse}[3]{% single line if-then-else
  \algorithmicif\ #1\ \algorithmicthen\ #2\ \algorithmicelse\ #3}
\algnewcommand{\InlineFor}[2]{\algorithmicfor\ #1\ \algorithmicdo\ #2} % single line for loop
\algnewcommand{\MyComment}[1]{\Comment{\small #1}} % comment with small font
\algnewcommand{\algorithmicand}{\textbf{and}}
\algnewcommand{\algorithmicor}{\textbf{or}}
\algnewcommand{\FOR}{\algorithmicfor}
\algnewcommand{\OR}{\algorithmicor}
\algnewcommand{\AND}{\algorithmicand}
\algnewcommand{\IF}{\algorithmicif}
\algnewcommand{\THEN}{\algorithmicthen}
\algnewcommand{\ELSE}{\algorithmicelse}
\algnewcommand{\True}{\textsc{True}}
\algrenewcommand\textproc{\texttt}   %%% was textsc
\algrenewcommand\Call[2]{\nameref{#1}\ifthenelse{\equal{#2}{}}{}{(#2)}}
\newcommand{\algoName}[1]{Algorithm \nameref{#1}}
\makeatletter
  \newcommand{\StateX}[1]{%
    \setlength\@tempdima{\algorithmicindent}%
    \Statex\hskip\dimexpr#1\@tempdima\relax%
  }
  \newcommand{\algoCaptionLabel}[2]{
     \caption[\textproc{\upshape #1}]{\textproc{\upshape #1}\ifthenelse{\equal{#2}{}}{}{\((#2)\)}}%
     \NR@gettitle{\textproc{\upshape #1}}%
      \label{algo:#1}%
     }%
\makeatother
\usepackage[shortlabels]{enumitem}
\usepackage{bm}
\usepackage[skins]{tcolorbox} % 'skins' for boxed title
\tcbuselibrary{theorems}
\usetikzlibrary{shapes}
\usetikzlibrary{positioning}
\usetikzlibrary{arrows}

\usepackage{amsmath}

\usepackage[colorlinks,linkcolor=cyan,citecolor=cyan,urlcolor=cyan]{hyperref}
\makeatletter  % this is because of cleveref
\let\cl@chapter\undefined
\makeatletter
\usepackage[capitalise]{cleveref}
\newcommand{\hyptime}[1]{\cref{hyp:timepm}.\ref{#1}} % referencing an assumption of timepm
\newtheorem{assumption}{Assumption}

\newtcbtheorem[
  crefname={Problem}{Problems},
  Crefname={Problem}{Problems},
  ]
{problembox}{Problem}{%
  boxsep=0cm,%
  left=0.2cm, top=0.2cm, right=0.2cm, bottom=0.10cm,%
  enhanced,%
  attach boxed title to top left={xshift=0.2cm, yshift=-2mm},%
  colbacktitle=white, coltitle=black, fonttitle=\itshape, %
  colback=white, %
  boxed title style={ boxrule=0.04em, colframe=black, boxsep=0cm },%
}{pbm}
\newcommand{\problemName}[1]{Problem~\nameref{#1}} % reference to problem name
\newcommand{\bigO}[1]{\mathchoice{O\!\left(#1\right)}{O(#1)}{O(#1)}{O(#1)}} % big O for complexity
\newcommand{\cstmatmul}{\mu} % constant for 2 x 2 polynomial matrix multiplication
\newcommand{\cstpolmulrep}[1]{\pi_{#1}} % constant for multiplying polynomials f1, ..., fk by the same g
\newcommand{\cstgeneral}{C} % general letter for constants for relation problems
\newcommand{\timepm}[1]{\mathchoice{\mathsf{M}\!\left(#1\right)}{\mathsf{M}(#1)}{\mathsf{M}(#1)}{\mathsf{M}(#1)}} % time for polynomial multiplication
\newcommand{\timemp}[1]{\mathchoice{\mathsf{MP}\!\left(#1\right)}{\mathsf{MP}(#1)}{\mathsf{MP}(#1)}{\mathsf{MP}(#1)}} % time for polynomial middle product
\newcommand{\Cxity}[1]{\mathchoice{\mathsf{C}\!\left(#1\right)}{\mathsf{C}(#1)}{\mathsf{C}(#1)}{\mathsf{C}(#1)}} % time complexity
\makeatletter
\newcommand*{\rem}{%
  \nonscript\mskip-\medmuskip\mkern5mu%
  \mathbin{\operator@font rem}\penalty900\mkern5mu%
  \nonscript\mskip-\medmuskip
}
\makeatother
\DeclareMathOperator{\lcm}{lcm}
\newcommand{\lc}[1]{\operatorname{lc}(#1)}
\newcommand{\modulus}{\mathfrak{M}}
\newcommand{\bmodulus}{\bar{\mathfrak{M}}}
\newcommand{\nodulus}{\mathfrak{N}}
\newcommand{\valuation}[1]{\operatorname{val}(#1)} % valuation of a univariate polynomial
\newcommand{\abs}[1]{\operatorname{abs}(#1)}  % absolute value
\newcommand{\ZZp}{\mathbb{Z}_{>0}} % positive integers
\newcommand{\NN}{\mathbb{N}} % nonnegative integers
\newcommand{\ZZ}{\mathbb{Z}} % integers

\newcommand{\field}{\mathbb{K}}
\newcommand{\mKK}[2]{\field^{#1\times#2}}
\newcommand{\pR}{\field[x]}
\newcommand{\mpR}[2]{\pR^{#1\times#2}}
\newcommand{\relmod}[2]{\mathcal{R}_{#1}(#2)} % module of relations for #2 mod #1
\newcommand{\appmod}[2]{\mathcal{A}_{#1}({#2})} % module of approximants
\newcommand{\intmod}[2]{\mathcal{I}_{#1}({#2})} % module of interpolants
\newcommand{\apo}{a} % poly 1 in equation relmod a,b modulo M
\newcommand{\bpo}{b} % poly 2 in equation relmod a,b modulo M
\newcommand{\apor}{\alpha} % poly 1, residual version
\newcommand{\bpor}{\beta} % poly 2, residual version
\newcommand{\bapo}{\bar{a}} % poly 1, divided by gcd or reversed
\newcommand{\bbpo}{\bar{b}} % poly 2, divided by gcd or reversed
\newcommand{\bapor}{\bar{\alpha}} % poly 1, residual version, reversed
\newcommand{\bbpor}{\bar{\beta}} % poly 2, residual version, reversed
\newcommand{\ash}{s} % shift, entry 1
\newcommand{\ashr}{t} % shift, entry 1, residual version
\newcommand{\bshr}{\tau} % shift, entry 2, residual version
\newcommand{\aev}{a} % one evaluation of poly 1
\newcommand{\bev}{b} % one evaluation of poly 2
\newcommand{\aevs}{\bm{\aev}} % list of evaluations of poly 1
\newcommand{\bevs}{\bm{\bev}} % list of evaluations of poly
\newcommand{\arev}{\alpha} % one evaluation of poly res 1
\newcommand{\brev}{\beta} % one evaluation of poly res 2
\newcommand{\arevs}{\bm{\arev}} % list of evaluations of poly res 1
\newcommand{\brevs}{\bm{\brev}} % list of evaluations of poly res 2
\newcommand{\pg}{g} % poly gcd in xgcd
\newcommand{\pp}{p} % poly 1 in first row of relation basis
\newcommand{\pq}{q} % poly 2 in first row of relation basis
\newcommand{\num}{q} % numerator in rat recon
\newcommand{\den}{p} % denominator in rat recon
\newcommand{\dnum}{\delta_1} % degree bound on numerator in rat recon
\newcommand{\dden}{\delta_0} % degree bound on denominator in rat recon
\newcommand{\basis}{P} % basis matrix for modules (2 x 2 polynomial matrix)
\newcommand{\basisr}{Q} % basis matrix for modules (2 x 2 polynomial matrix), alternative
\newcommand{\ddegP}{\mathring{p}} % diagonal degree of basis P
\newcommand{\ddegQ}{\mathring{q}} % diagonal degree of basis P
\newcommand{\resbasis}{P}  % output basis
\newcommand{\firstbasis}{P}  % output basis of the first recursive call
\newcommand{\secondbasis}{Q}   % output basis of the second recursive call
\newcommand{\evbasis}{\bm{\basis}} % evaluations of basis
\newcommand{\evfirstbasis}{\bm{\firstbasis}} % evaluations of the first output basis
\newcommand{\evsecondbasis}{\bm{\secondbasis}} % evaluations of the second output basis
\newcommand{\lenP}{\delta} % length of basis P (deg(P) < lenP)
\newcommand{\dreca}{d_{1}}  % order for first recursive call in approx basis
\newcommand{\drecb}{d_{2}}  % order for second recursive call in approx basis
\newcommand{\dsyz}{d_{\mathrm{syz}}} % order for finding syzygy
\newcommand{\dred}{d_{\mathrm{red}}} % parameter of complexity for reduction
\newcommand{\rdeg}[2][]{\operatorname{rdeg}_{#1}(#2)}  % shifted row degree
\newcommand{\diag}[1]{\operatorname{Diag}(#1)} % diagonal matrix with diag entries #1
\newcommand{\point}{\omega} % one point
\newcommand{\points}{\bm{\point}} % list of points
\newcommand{\ratio}{r} % ratio of a geometric sequence
\newcommand{\bitrev}[1]{[#1]_N} % index in bit reversed sequence
\newcommand{\bitrevlen}[2]{[#1]_{#2}} % index in bit reversed sequence
\newcommand{\evaluate}{\textproc{evaluate}}
\newcommand{\interpolate}{\textproc{interpolate}}
\newcommand{\extrapolate}{\textproc{extrapolate}}
\newcommand{\dinp}{m} % nb of points in input
\newcommand{\dout}{n} % nb of points in output
\newcommand{\ptsinp}{\points_{\mathrm{in}}} % points in input
\newcommand{\ptsinprange}[1]{\points_{\mathrm{in},#1}} % points in input
\newcommand{\ptsout}{\points_{\mathrm{out}}} % points in output
\newcommand{\midprod}[1]{\textproc{MidProd}(#1)}  % middle product
\newcommand{\midprodBig}[1]{\textproc{MidProd}\left(#1\right)}  % middle product, big paren
\begin{document}

% NOTE [refs]
% Burgisser et al, sec 3.1: GCD, Euclidean remainder scheme, continued fraction expansion, genericity of degree pattern (n-m, 1, ..., 1, 0), half-gcd presentation
% Burgisser et al, sec 3.2: Strassen's analysis of the half-gcd (constant not given for complexity upper bound; but bound does involve degrees of individual quotients during computation through the entropy function)
%   -> Strassen. The computational complexity of continued fractions. In ACM Symposium on Symbolic and Algebraic Computation, pages 51-67, 1981
%   -> Strassen. The computational complexity of continued fractions. SIAM J. Comp. 12:1-27, 1983
% Burgisser et al, exo 3.2: Uniqueness of cofactors with strong degree bounds
% Burgisser et al, exo 3.7+3.8: direct link between XGCD and Pade

% NOTE other relevant refs:
% theorem 4.10 Gathen Gerhard
% book by Pan, Structured Matrices, page 40 et sqq: reductions PAdé / GCD (non X) / Berlekamp-Massey, link with structured matrices, Figure 2.2, ..

\title{Refined complexity bounds for rational reconstruction
and XGCD through Pad\'e approximants and Cauchy interpolants%\thanks{Grants or other notes
}
% \subtitle{Do you have a subtitle?\\ If so, write it here}

\titlerunning{Refined complexity bounds for rational reconstruction, XGCD, approximants, and interpolants}        % if too long for running head

\author{Vincent Neiger     \and
        Mohab Safey El Din \and
        Kevin Tran
}


\institute{Vincent Neiger \at
              Sorbonne Universit\'e, CNRS, LIP6,
              F-75005 Paris, 75252, France \\
              \email{vincent.neiger@lip6.fr}
           \and
           Mohab Safey El Din \at
              Sorbonne Universit\'e, CNRS, LIP6,
              F-75005 Paris, 75252, France \\
              \email{mohab.safey@lip6.fr}
           \and
           Kevin Tran \at
              Sorbonne Universit\'e, CNRS, LIP6,
              F-75005 Paris, 75252, France \\
              \email{kevin.tran@lip6.fr}
}

\date{17 September 2026}

\maketitle

\begin{abstract}
  When computing with univariate polynomials, two fundamental and related
  problems are the extended greatest common divisor (XGCD) and rational
  reconstruction, classically solved in quasi-linear complexity using the
  half-gcd algorithm. These problems have various applications in algebraic
  computations and bear strong connections to linearly recurrent sequences,
  structured matrices, and continued fractions.

  This article first gives a collection of algorithmic reductions, showing that
  rational reconstruction and XGCD can be solved via the computation of bases
  of relations modulo a freely-chosen polynomial \(\modulus(x)\). In
  particular, one recovers the folklore idea that bases of Pad\'e approximants
  (i.e., \(\modulus(x) = x^d\)) can be used to perform quasi-linear rational
  reconstruction or XGCD, extending to fast algorithms the well-known link
  between the Berlekamp-Massey algorithm and the extended Euclidean algorithm.
  One highlight of these reductions is that, instead of approximants, one may
  rather rely on Cauchy interpolants (i.e., \(\modulus(x)\) vanishes at chosen
  points).

  In a second part, this article describes divide-and-conquer algorithms for
  approximants and interpolants along with complexity analyses showing an
  explicit leading constant in front of the dominant term. For interpolants,
  the best leading constant is obtained through a variant that stores
  polynomials represented by evaluations, and exploits fast extrapolation in
  order to avoid repeated conversions to the monomial basis; this requires
  special points, in geometric or arithmetic progression, or FFT
  points when the base field allows them.

  Combining the analyses with the reductions leads to the best complexity
  bounds we are aware of for rational reconstruction and XGCD. Perhaps
  surprisingly, even Pad\'e approximants or Berlekamp-Massey-like computations,
  which intrinsically involve \(\modulus(x) = x^d\), are accelerated by
  reducing them to Cauchy interpolation at well-chosen points.
%Include keywords, PACS and mathematical
%subject classification numbers as needed.
%\subclass{MSC code1 \and MSC code2 \and more}
\end{abstract}

\section{Introduction}
\label{sec:intro}

Amongst the most fundamental computational problems on univariate polynomials,
one finds rational reconstruction \cite[\S5.7]{GathenGerhard2013}, with
important special cases such as Pad\'e approximation and Cauchy interpolation
\cite[\S5.8 and \S5.9]{GathenGerhard2013}.

\begin{problembox}{RatRecon}{ratrecon}
  \emph{Input:}
     nonzero modulus \(\modulus \in \pR\) of degree \(d \in \NN\),
     polynomial \(\apo \in \pR\) of degree \(< d\),
     integers \(\delta_0\) and \(\delta_1\) both in \(\{0,\ldots,d\}\).

  \emph{Output:}
  nonzero pair \((\den, \num) \in \pR^2\) such that
  \(\den \apo = \num \bmod \modulus\),
  \(\deg(\den) \le \dden\) and \(\deg(\num) < \dnum\);
  or flag \textsc{None} if no such pair exists.
\end{problembox}

Here and hereafter, \(\field\) is a field and \(\pR\) is the ring of univariate
polynomials over \(\field\). The typical instance has \(\dnum = d - \dden\),
which ensures the existence of a solution. Pad\'e approximation corresponds to
\(\modulus = x^d\), while Cauchy interpolation is with \(\modulus = \prod_{0
\le i < d}(x - \point_i)\) for known pairwise distinct elements \(\points =
(\point_i)_{0 \le i < d}\) from \(\field\).

There are strong links between rational reconstruction and other core problems,
notably finding minimal generators for linearly recurrent sequences, computing
greatest common divisors, and solving structured linear systems with small
displacement rank \cite{SuKaHiNa75,BrentGustavsonYun1980,Dornstetter87}.
Well-known algorithms of quadratic complexity for these tasks include the
Berlekamp-Massey algorithm \cite{Berlekamp68,Massey69} and the extended
Euclidean Algorithm \cite[Sec.\,3]{GathenGerhard2013}. The latter computes the
greatest common divisor along with cofactors, thereby solving
\problemName{pbm:xgcd}. This was the first of these problems to be solved in
quasi-linear complexity, by the so-called half-gcd algorithm
\cite{Knu70,Sch71,Moe73}.
Quasi-linear algorithms have been described for rational reconstruction
problems as well
\cite{GustavsonYun1979,BrentGustavsonYun1980,BeckermannLabahn1994,BeckermannLabahn97}.

\begin{problembox}{XGCD}{xgcd}
  \emph{Input:}
  nonzero polynomials \(\apo, \bpo \in \pR \setminus \{0\}\) with \(\deg(\apo) < \deg(\bpo)\).

  \emph{Optional input:}
  integer bound \(\ell \in \NN\) such that \(\ell \le \deg(\gcd(\apo,\bpo))\),
  defaults to \(\ell=0\).

  \emph{Output:} the unique polynomials \(u,v,\pg\) in \(\pR\) with \(\pg\)
  monic such that
  {
    \setlength{\abovedisplayskip}{2pt} %% locally reduce vertical space before displayed equation
    \[
      u\apo + v\bpo = \pg = \gcd(\apo,\bpo),
      ~ \deg(u) < \deg(\bpo/\pg),
      ~\text{and } \deg(v) < \deg(\apo/\pg).
    \]
  }
\end{problembox}

%% NOTE concerning degree assumption. Imagine the only constraint was a and b nonzero.
% Since this problem is symmetric in \(\apo\) and
% \(\bpo\), we may assume \(\deg(\apo) \le \deg(\bpo)\). Up to performing a
% single step of the extended Euclidean algorithm if these degrees are equal
% (i.e., replacing \(\apo\) by \(\apo - \lc{\bpo}^{-1}\lc{\apo}\bpo\)), which has
% linear cost in that case, we further assume \(\deg(\apo) < \deg(\bpo)\).

The complexity of the half-gcd algorithm has attracted attention
\cite{Strassen1981,Strassen1983},
\cite[Sec.\,3.2]{BurgisserClausenShokrollahi2010}, \cite[Sec.\,11.1,
Pbm.\,11.11]{GathenGerhard2013}. For cryptographic purposes, a constant-time
version of the algorithm has been designed  \cite{BernsteinYang2019}. Recently,
refinements of the half-gcd algorithm have been described, along with a
thorough complexity analysis leading to a good estimate and understanding of
the leading constant in the asymptotic complexity bound \cite{vdHoeven2025}.
Following a similar path concerning Pad\'e approximation and Cauchy
interpolation on special points, the first set of contributions of this paper
presents optimized algorithms for these problems together with complexity
analyses that exhibit the leading constants (\cref{sec:approx,sec:interp},
outlined in \cref{sec:intro:approx_interp}). The second set of contributions
provides efficient algorithmic reductions implying that the latter leading
constants are also valid for the general case of \problemName{pbm:ratrecon} as
well as for \problemName{pbm:xgcd} (\cref{sec:rel_xgcd}, outlined in
\cref{sec:intro:rel_xgcd,sec:intro:approx_interp}).

\subsection{Bases of relations with minimal degrees}
\label{sec:intro:bases}

The above-mentioned algorithms for \cref{pbm:ratrecon,pbm:xgcd} manipulate,
implicitly or explicitly, polynomial matrices of size \(2 \times 2\). The
half-gcd algorithm computes a collection of matrices of determinant \(-1\),
more precisely of the form \([\begin{smallmatrix} 0 & 1 \\ 1 & -f
\end{smallmatrix}]\) where \(f \in \pR\) is a quotient arising at one step of
the Euclidean algorithm. Multiplying these matrices yields a
unimodular matrix which transforms the input polynomials into \(\pg =
\gcd(\apo,\bpo)\) and \(0\), such as
\(
[\begin{smallmatrix}
  u & v \\
  \bpo/\pg & -\apo/\pg
\end{smallmatrix}]
[\begin{smallmatrix}
  \apo \\ \bpo
\end{smallmatrix}]
=
[\begin{smallmatrix}
  \pg \\ 0
\end{smallmatrix}]
\).
On the other hand, algorithms for Pad\'e approximation compute a \(2\times 2\)
matrix, called approximant basis, whose determinant divides \(\modulus = x^d\)
and whose rows generate all solutions \((\den,\num)\) to the equation \(\den
\apo = \num \bmod x^d\) \cite{BrentGustavsonYun1980,BeckermannLabahn1994}. One
computes such a matrix with some degree minimality property so that at least
one of its rows satisfies the prescribed degree constraints \(\deg(\den) \le
\dden\) and \(\deg(\num) < \dnum\), unless no such solution exists.

Similarly, our algorithms will compute \(2 \times 2\) matrices which are bases
of some set of relations. For polynomials \(\modulus,\apo,\bpo \in \pR\) with
\(\modulus\) nonzero, we define the \(\pR\)-module
\begin{equation}
  \label{eqn:relmod}%
  \relmod{\modulus}{\apo,\bpo} = \{(p,q) \in \pR^2 \mid \apo \pp + \bpo \pq = 0 \bmod \modulus\}.
\end{equation}
Note that the relations of \problemName{pbm:ratrecon} correspond to \(\bpo =
-1\). We introduce some notation for two particular cases that play an
important role in this paper: the module of approximants \(\appmod{d}{\apo,\bpo}
= \relmod{\modulus}{\apo,\bpo}\), for a given approximation order \(d \in \NN\) and
where \(\modulus = x^d\); and the module of interpolants
\(\intmod{\points}{\apo,\bpo} = \relmod{\modulus}{\apo,\bpo}\), for points
\(\points  = (\point_0, \ldots, \point_{d-1})\in \field^{d}\) and where
\(\modulus = \prod_{i=0}^{d-1} (x-\point_i)\). We will focus on pairwise
distinct points \(\points\).

The module \(\relmod{\modulus}{\apo,\bpo}\) has rank \(2\), so that its bases
can be represented as matrices \(P \in \mpR{2}{2}\) whose rows generate all
relations of \(\relmod{\modulus}{\apo,\bpo}\) via \(\pR\)-linear combinations.
We call such a matrix a relation basis. To account for the degree constraints
in \(\deg(\den) \le \dden\) and \(\deg(\num) < \dnum\), we use the notion of
shifted weak Popov form \cite{BeckermannLabahnVillard1999}, recalled in
\cref{sec:prelim:polmat}. Roughly, given some integer \(\ash \in \ZZ\), a basis with
this form has rows of minimal degree after some shifting by \(\ash\). Targeting
the above degree constraints amounts to taking \(\ash = \dnum-\dden\).

This leads us to reformulate \problemName{pbm:ratrecon} into the following
problem, which is the one that our algorithms will actually solve, for some of
them with \(\bpo=-1\).

\begin{problembox}{RelBasis\textsubscript{2}}{rel}
  \emph{Input:}
  monic polynomial \(\modulus \in \pR\) of degree \(d \in \NN\),
  polynomials \((\apo,\bpo) \in \pR^2\) of degree less than \(d\),
  shift \(\ash \in \ZZ\).

  \emph{Output:} a basis of \(\relmod{\modulus}{\apo,\bpo}\) in \(\ash\)-weak Popov form.
\end{problembox}

\subsection{Complexity framework}
\label{sec:intro:complexity}

This paper studies the complexity of algebraic algorithms operating on
univariate polynomials. For a given algorithm, we are interested in finding an
asymptotic upper bound on the number of basic operations in the base field
\(\field\) it performs. We count at unit cost each of these operations, namely
addition, subtraction, multiplication, inversion, and testing equality.
Following \cite[Sec.\,8.3]{GathenGerhard2013}, we denote by \(\timepm{\cdot}\)
a time function for polynomial multiplication: two polynomials in \(\pR\) of
degree \(\le d\) can be multiplied using at most
\(\timepm{d}\) operations in \(\field\). One can always take \(\timepm{d} \in
\bigO{d \log(d) \log(\log(d))}\) \cite{CantorKaltofen1991}.

Most nontrivial algorithms we will rely on run
%% either in \(\bigO{d}\), such as polynomial addition or multiplication by a constant, or
in \(\bigO{\timepm{d}}\), such as inversion modulo \(x^d\) and polynomial
division with remainder. The half-gcd algorithm uses \(\bigO{\timepm{d}
\log(d)}\) field operations, and the variants of \problemName{pbm:ratrecon}
studied in this paper are also already solved within this complexity bound by
the above-referenced results. We aim to get refined bounds of the form \(\cstgeneral\,
\timepm{d} \log(d) + \bigO{\timepm{d}}\), with a leading constant \(\cstgeneral > 0\) which is as
low as possible, through careful complexity analyses combined with algorithmic
optimizations and reductions.

We will determine such leading constants for three classical contexts. The
first one is general: we make no assumption on \(\field\). More details on our
complexity framework in this case are described in
\cref{sec:approx:polmul_midprod}. The second context is when \(\field\)
possesses an element \(\ratio\) of order at least \(d\), which allows us to
exploit fast evaluation and interpolation at the points from the geometric
progression \(\points = (\point_0 \ratio^i)_{0 \le i < d}\); points in arithmetic
progression \(\points = (\point_0 + i\ratio)_{0 \le i < d}\) are also
considered in our algorithms, although to a lesser extent. The third context
is when the radix-2 fast Fourier transform (FFT) is feasible, that is,
\(\field\) possesses a principal \(d\)-th root of unity \(\ratio\) where \(d\)
is a power of \(2\). We will simply say that \(\field\) supports geometric
progressions or supports FFT. More details on these latter two contexts are
given in \cref{sec:interp:eval_interp_extrap}. While the second
context remains fairly general, since it roughly covers all ``sufficiently
large'' fields \(\field\), the third one is quite more restrictive. It is
nevertheless interesting, notably for multimodular methods where one works in
fields \(\ZZ/p\ZZ\) with \(p\) that can be chosen to be an FFT prime.

For each of these contexts concerning \(\field\), we give two different
complexity leading constants. One is for arbitrary input, using at every stage
of the analysis the worst-case degree bounds available for the manipulated
polynomials. The other constant is lower, by exploiting better bounds that are
valid for input that satisfies some genericity property discussed in
\cref{app:approx:genericity}, akin to the so-called normal case for
\nameref{pbm:xgcd} about inputs for which the sequence of degrees of the
remainders produced by the Euclidean algorithm decreases by exactly one at each
iteration after the first one \cite[Sec.\,3.1]{BurgisserClausenShokrollahi2010}
\cite{GathenGerhard2013,vdHoeven2025}. Note that our algorithms handle
arbitrary input, and it is only in the analysis of their complexity that we
distinguish these two cases. For input that does not meet the genericity
requirements but where most of the computations in the algorithm behave as in
the generic case, we expect the actual leading constant to sit somewhere
between the two extremes we report.

We let \(\cstmatmul > 0\) be a constant such that two matrices in
\(\mpR{2}{2}\) of degree at most \(d\) can be multiplied using
\(\cstmatmul\timepm{d} + \bigO{d}\) operations in \(\field\). The naive
algorithm yields \(\cstmatmul = 8\), and one can take \(\cstmatmul = 7\) using
Strassen's algorithm or \(\cstmatmul = 4\) when FFT techniques are used for
multiplication. Note that the latter does not require the specific radix-2 FFT
to be feasible in \(\field\). Similarly, for a fixed \(k \ge 1\), we let
\(\cstpolmulrep{k} > 0\) be a constant such that one can perform the \(k\) multiplications \(p_1 q,
\ldots, p_k q\) in \(\cstpolmulrep{k}\timepm{d} + \bigO{d}\) operations in
\(\field\), where the \(p_i\)'s and \(q\) are polynomials in \(\pR\) of degree
\(\le d\). One obviously has \(\cstpolmulrep{k} \le k\); furthermore, one can
take \(\cstpolmulrep{k} = (2k+1)/3\) when FFT techniques are available.

\subsection{Algorithmic reductions for relation reconstruction and XGCD}
\label{sec:intro:rel_xgcd}

A common strategy of the quasi-linear half-gcd and Pad\'e approximation
algorithms, at the core of their improvement over quadratic algorithms, is to
leverage fast polynomial multiplication through a divide and conquer approach.
For input polynomials of degree \(\le d\), this is done by splitting each of
them in two polynomials of degree \(\le d/2\), and solving two subproblems with
half the input size. These subproblems are not independent: the input of the
second one, often called residual, is built from the output of the first
subproblem. One has two choices: handle the high degrees first, as usually done
in the half-gcd algorithm, or handle the low degrees first, as usually done for
Pad\'e approximation. Roughly, going one direction or the other can be
simulated by reversing the order of input or output coefficients, which is one
main intuition behind the links between the different problems and algorithms
\cite[Sec.\,4]{BrentGustavsonYun1980} \cite{Dornstetter87}.

In this paper, we describe reductions that provide a broad understanding of
the computational links between \cref{pbm:ratrecon,pbm:xgcd,,pbm:rel}. They
allow us to solve both \problemName{pbm:xgcd} and the general case of
\problemName{pbm:ratrecon} with arbitrary modulus \(\modulus\) by solving an
instance of \problemName{pbm:rel} with another modulus \(\nodulus\) of our
choice. The only restriction is that \(\deg(\nodulus)\) must be sufficiently
large: the required lower bound, denoted by \(\dred\) below, is roughly between
\(0\) and \(2d\) depending on the original instance.
One can thus pick \(\nodulus\) that leads to the most efficient solutions for
\problemName{pbm:rel}. For example, one can solve \Call{pbm:xgcd}{} or find
relations modulo any \(\modulus\) by relying on approximant bases (i.e.,
\(\nodulus=x^{\dred}\)), which concurs with the known links mentioned in the
previous paragraph.

What is a new observation, to the best of our knowledge, is that one can also
rely on interpolant bases, i.e., use \(\nodulus = \prod_{i=0}^{\dred-1}
(x-\point_i)\) for points \((\point_i)_i\) of our choice. As we will see,
combined with our optimized algorithms for interpolant bases, this approach
yields the best leading constant we are aware of, including for
\problemName{pbm:xgcd}. Furthermore, this sheds a light on the relative
difficulty of instances with different moduli, as highlighted by the following
summary of some interesting consequences of our reductions:

\begin{itemize}[noitemsep,topsep=2pt]
  \item The leading constant for \problemName{pbm:ratrecon} is not larger than
    the one for \problemName{pbm:rel} with a modulus of our choice of degree
    \(\ge d + \dden - \dnum\).

  \item The leading constant for \problemName{pbm:xgcd} is not larger than
    that for \problemName{pbm:rel} with a modulus of our choice of degree
    \(\ge 2(\deg(\apo) - 1 - \ell)\).

  \item To obtain a rational reconstruction or an XGCD algorithm of complexity
    \(\bigO{\timepm{d}}\), it is sufficient to design an algorithm of
    complexity \(\bigO{\timepm{d}}\) for \problemName{pbm:rel} with modulus
    \(\prod_{i=0}^{d-1} (x-\point_i)\), where the points \((\point_i)_i\) can
    be chosen freely. This includes choosing points in arithmetic or geometric
    progression, or FFT points, if the field \(\field\) allows them.
\end{itemize}
The third item answers a question that arises naturally from the knowledge
that the problems of multipoint evaluation and interpolation are solved in
\(\bigO{\timepm{d}\log(d)}\) for general points, but in \(\bigO{\timepm{d}}\)
for points in geometric progression. If one manages to remove the \(\log(d)\)
factor for Cauchy interpolation at points in geometric progression, this would
directly imply that general points can be handled in \(\bigO{\timepm{d}}\); in
addition, \problemName{pbm:xgcd} would then be solved in \(\bigO{\timepm{d}}\)
too.
The first two items are detailed in \cref{thm:rel,cor:xgcd} and proved in
\cref{sec:rel_xgcd}.

\begin{theorem}
  \label{thm:rel}%
  Let \(\modulus \in \pR\) be monic of degree \(d \in \NN\), let \(\apo \in
  \pR_{<d}\), let \(\ash \in \ZZ\), and define \(\dred = 2\deg(\apo)
  - d - \ash - 1 + \min(0,\ash+d)\), which satisfies \(\dred < \deg(\apo) - \ash\).
  \begin{itemize}[noitemsep,topsep=2pt]
    \item If \(\dred \le 0\), one can solve \(\Call{pbm:rel}{\modulus,\apo,-1,\ash}\)
      using \(\bigO{\timepm{d}}\) operations in \(\field\).
    \item Assume \(\dred > 0\). Let \(\nodulus \in \pR\) have degree \(\delta
      \ge \dred\) and be either \(\nodulus = x^{\delta}\) or such that
      \(\valuation{\nodulus} = 0\). There are polynomials \(\bapo,
      \bbpo \in \pR_{<\delta}\) which can be computed in
      \(\bigO{\timepm{d}}\) operations in \(\field\) and are such that from the
      output of \(\Call{pbm:rel}{\nodulus,\bapo,\bbpo,\delta-\dred}\)
      one can solve \(\Call{pbm:rel}{\modulus,\apo,-1,\ash}\)
      using \(\bigO{\timepm{d}}\) operations in \(\field\).
  \end{itemize}
\end{theorem}

\begin{corollary}
  \label{cor:xgcd}%
  Let \(\apo, \bpo \in \pR \setminus \{0\}\) with \(\deg(\apo) < \deg(\bpo)\)
  and let \(\ell \in \NN\) with \(\ell \le \deg(\gcd(\apo,\bpo))\).
  Let \(\dred %% = 2\deg(\apo) - \deg(\bpo) - 1 - \ell + \deg(\bpo) - 1 - \ell
  = 2(\deg(\apo) - 1 - \ell) \in \NN \cup \{-2\}\).
  If \(\dred < 0\), then the solution to \cref{pbm:xgcd} is \((u,v,\pg) =
  (\lc{\apo}^{-1}, 0, \lc{\apo}^{-1}\apo)\).
  Assume \(\dred \ge 0\). Let \(\nodulus \in \pR\) have degree \(\delta \ge
  \dred\) and be either \(\nodulus = x^{\delta}\) or such that
  \(\valuation{\nodulus} = 0\). There are polynomials \(\bar{\apo}, \bar{\bpo}
  \in \pR_{<\delta}\) which can be computed in \(\bigO{\timepm{\deg(\bpo)}}\)
  operations in \(\field\) and are such that
  from the output of \(\Call{pbm:rel}{\nodulus,\bapo,\bbpo,\delta-\dred}\)
  one can solve \cref{pbm:xgcd} using \(\bigO{\timepm{\deg(\bpo)}}\)
  operations in \(\field\).
\end{corollary}

When applying \cref{thm:rel} to solve \problemName{pbm:ratrecon}, by taking
\(\ash = \dnum-\dden\) as mentioned above, it is sufficient to choose
\(\nodulus\) such that \(\deg(\nodulus) \ge \deg(\apo) + \dden - \dnum\), where
typically \(\deg(\apo) \approx d\). For the frequent case where one seeks a
solution with uniform degree constraints \(\dden = \dnum\), one can take
\(\nodulus\) of degree \(\deg(\apo)\). In the extreme case with \(\dden = d-1\)
and \(\dnum = 1\), the modulus degree is increased from \(\deg(\modulus) = d\)
to \(\deg(\nodulus) \approx \deg(\apo) + d\), which means a growth factor
around \(2\) when \(\deg(\apo) \approx d\). This is consistent with the factor
\(2\) in \cref{cor:xgcd}: for these degree constraints, \nameref{pbm:ratrecon}
corresponds to modular inversion \(\den \apo = 1 \bmod \modulus\), typically
done via \nameref{pbm:xgcd}.
It is however not satisfactory to have this factor \(2\) increase
when \(\modulus = x^d\), since it is known how to perform
modular inversion in \(\bigO{\timepm{d}}\) by a Newton iteration. Following
this observation, in \cref{sec:rel_xgcd:app_to_modN} we prove more generally
that for \(\modulus = x^d\) and arbitrary constraints
\((\dden,\dnum)\), one can take \(\nodulus\) of degree \(d - \abs{\dnum -
\dden} \le d\), where \(\abs{\cdot}\) is the absolute value.

To summarize, these reductions show that both Problems \nameref{pbm:xgcd} and
\nameref{pbm:ratrecon} can be solved through the computation of relation bases
with a modulus of our choice. Consequently, complexity bounds for these two
problems are readily deduced from those established in
\cref{sec:approx,sec:interp} for approximant and interpolant bases, and
discussed below in \cref{sec:intro:approx_interp}. The reductions and the
resulting leading constants are summarized in
\cref{fig:constants_via_reductions}. Prior to this work, the best leading
constants found in the literature for quasi-linear \nameref{pbm:xgcd} and
\nameref{pbm:ratrecon} were obtained through variants of the half-gcd
algorithm, and are summarized in \cref{fig:tabhgcd}.

\begin{figure}[htb]
  \centering
  \begin{tikzpicture}
    \node[draw] (relbasisgen) {
        \nameref{pbm:rel}, \(\relmod{\nodulus}{\apo,\bpo}\)
        for any \(\nodulus\) with \(\deg(\nodulus) \ge \dred\)
      };

    \node[below=1.5cm of relbasisgen.south west, draw, text width=0.29\textwidth] (relbasis) {
        \nameref{pbm:rel}, \(\relmod{\modulus}{\apo,-1}\) \\
        \(\dred \le \deg(\apo) - \ash\)
      };
    \draw[-stealth] (relbasis) --
      node[left] {\small \cref{thm:rel}, \cref{algo:RelationBasis-via-modN}}
        (relbasisgen);

    \node[below=2.0cm of relbasis.south west, xshift=1.0cm, draw, text width=0.25\textwidth] (ratrecon) {
        \nameref{pbm:ratrecon}, general \(\modulus\) \strut \\
        \(\dred \le \deg(\apo) + \dden - \dnum\)
      };
    \draw[-stealth] (ratrecon) --
      node[left] {\small \(\ash = \dnum-\dden\)}
        (relbasis);

    \node[below=2.0cm of relbasis.south east, xshift=1.0cm, draw, text width=0.29\textwidth] (xgcd) {
        \nameref{pbm:xgcd}\((\apo,\bpo)\) \\
        \(\dred \le 2(\deg(\apo) - 1 - \ell)\)
      };
    \draw[-stealth] (xgcd) --
      node[right, xshift=0.5cm, yshift=-0.3cm, text width=0.35\textwidth] {
        \small
        \cref{cor:xgcd}, \cref{sec:rel_xgcd:xgcd_to_rel} \\
        \(\modulus = \bpo, \ash = 2\ell-d+1\) \\
      }
      (relbasis);

    \node[right=4.5cm of relbasis, draw, text width=0.22\textwidth] (appbasis) {
        \nameref{pbm:app}, \(\appmod{d}{\apo,\bpo}\) \\
        \(\dred \le d - \abs{\ash}\)
      };
    \draw[-stealth] (appbasis) --
      node[above, text width=0.33\textwidth] {
        \small 
        \cref{sec:rel_xgcd:app_to_modN} \\
        \(\modulus = x^{d-v}\), \(v \in \{\valuation{\apo}, \valuation{\bpo}\}\) \\
      }
      node[below, text width=0.33\textwidth] {
        \small
        \(\apo \gets (x^{-v} \bpo)^{-1} \apo \rem x^{d-v}\) \\
        ~~ or \((x^{-v} \apo)^{-1} \bpo \rem x^{d-v}\) \\
      }
      (relbasis);

    \node[below=2.0cm of appbasis, draw, text width=0.24\textwidth] (ratreconxd) {
       \nameref{pbm:ratrecon}, \(\modulus=x^d\) \\
       \(\dred \le d - \abs{\dnum-\dden}\)
      };
    \draw[-stealth] (ratreconxd) --
      node[left,text width=0.11\textwidth] {\small \(\ash = \dnum-\dden\) \\ \(\bpo=-1\) \\ }
        (appbasis);
  \end{tikzpicture}

  \bigskip

  {\small
  \setlength{\tabcolsep}{4pt}
  \begin{tabular}{|c||c|c||c|c|c||c|c|c|}
    \hline
    field & \multicolumn{2}{c||}{\(\field\) is arbitrary}
                           & \multicolumn{3}{c||}{\(\field\) supports geometric}
                           & \multicolumn{3}{c|}{\(\field\) supports FFT}
    \\ \hline \hline
    case & worst & generic & worst & generic & persp. & worst & generic & persp. \\ \hline
    \rule{0pt}{1.3\normalbaselineskip} % top strut
    \rule[-10pt]{0pt}{0pt} % bottom strut
    constant & \(\dfrac{7}{2}\) & \(\min\left(\dfrac{\cstmatmul+6}{4}, \dfrac{\cstmatmul}{2}\right)\) & \(3\) & \(\dfrac{5\cstpolmulrep{4}}{8}\) & \(\dfrac{3\cstpolmulrep{3} + 2\cstpolmulrep{2}}{8}\) & \(1\) & \(\dfrac{5}{6}\) & \(\dfrac{13}{24}\) \\ 
    \rule[-10pt]{0pt}{0pt} % bottom strut
             &                  & \(\in [2, \frac{7}{2}] \) &       & \(\in [\frac{15}{8}, \frac{5}{2}]\) & \(\in [\frac{31}{24}, \frac{13}{8}]\) &       &                  &                    \\ 
    \hline
    approach & \multicolumn{2}{c||}{approximants}
                           & \multicolumn{3}{c||}{interpolants}
                           & \multicolumn{3}{c|}{interpolants}
    \\ \hline
    section & \multicolumn{2}{c||}{\ref{sec:approx:pmbasis}}
            & \multicolumn{2}{c|}{\ref{sec:interp:eval_intbasis:arith_geom}}
            & \ref{sec:interp:eval_intbasis_opti} 
            & \multicolumn{2}{c|}{\ref{sec:interp:eval_intbasis:fft}}
            & \ref{sec:interp:eval_intbasis_opti} \\ \hline
  \end{tabular}
  }

  \caption{The diagram shows reductions between different problems and
    instances, with for each of them an upper bound on the degree \(\dred\) of the
    modulus \(\nodulus\) in the target instance of relation bases. The
    reduction itself costs \(\bigO{\timepm{d}}\) operations in \(\field\),
    where \(d = \deg(\modulus)\) for all problems except \(d = \deg(\bpo)\) for
    XGCD. The bottom table shows resulting values of the leading constant \(\cstgeneral\)
    in the overall complexity bound \(\cstgeneral\, \timepm{\lceil \dred/2 \rceil}
    \log_2(\lceil \dred/2 \rceil) + \bigO{\timepm{d}}\), for different
  choices of \(\nodulus\). For the leading constants in the ``perspective''
  columns, we sketch the algorithmic approach but do not provide a full
study, in particular concerning genericity aspects.}
  \label{fig:constants_via_reductions}
\end{figure}
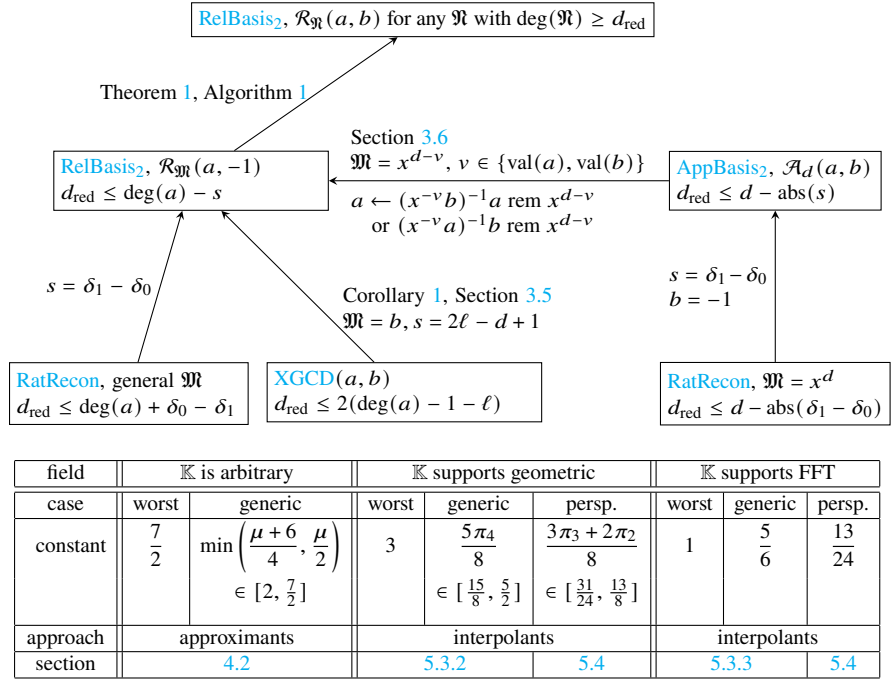

\begin{figure}[htb]
  \centering
  \begin{tabular}{|c||c|c||c|c||c|c|}
    \hline
    field & \multicolumn{2}{c||}{\(\field\) is arbitrary} &
    \multicolumn{2}{c||}{\(\field\) supports \(\cstmatmul=4\)}
                           & \multicolumn{2}{c|}{\(\field\) supports FFT}
    \\ \hline \hline 
    case & worst & generic & worst & generic & worst & generic \\ \hline
    \cite[Chap.\,11]{GathenGerhard2013} & \(22\) & \(10\) & -- & -- & -- & -- \\ \hline
    \cite[Chap.\,8]{Mateer2008} & \(19/2\) & \(5\) & -- & -- & \(15/4\) & \(2\) \\ \hline
    \cite{vdHoeven2025} & \(17/4\) & \(7/2\) & \(11/4\) & \(2\) & \(19/12\) & \(4/3\) \\ \hline
  \end{tabular}
  \caption{Leading constants in complexity bounds found in the literature about
  variants and analyses of the half-gcd algorithm.}%
  \label{fig:tabhgcd}
\end{figure}

% NOTE [refs]
% Mateer's work (see page 242 of Master's thesis)
% -> constant  19/2 (general) or 5 (generic), cf equations 8.48 to 8.51 p225
% -> constant 7.5/2 (FFT, general) or 2 (FFT, generic), cf equations 8.58 to 8.61 p230
% here, we gather M_M(2n) + A_M(2n) as being the number of operations for one degree-n multiplication
% and we exploit the 1/2 factor that arises from assumption M(n) / nlog(n) increases

\subsection{Algorithms for approximant and interpolant bases, and leading constants}
\label{sec:intro:approx_interp}

Thanks to the above reductions, we may restrict our attention to solving instances of
\cref{pbm:rel} with ``favorable'' moduli, that is, ones that help in the design
of fast algorithms. We focus on such instances where the modulus \(\modulus\)
splits into known linear factors: this property allows one to use Beckermann
and Labahn's recursive approach
\cite{BeckermannLabahn1994,BeckermannLabahn97}, which leads to quasi-linear
complexity when it is carried out in a divide and conquer fashion. More
precisely we study the cases of approximant bases and interpolant bases,
detailed in \cref{pbm:app,pbm:int}, and we design and analyze variants of this
approach specialized to our setting with two polynomials in input.

\begin{problembox}{AppBasis\textsubscript{2}}{app}
  \emph{Input:}
  integer \(d \in \NN\),
  polynomials \((\apo,\bpo) \in \pR_{<d}^2\),
  shift \(\ash \in \ZZ\).

  \emph{Output:} a basis of \(\appmod{d}{\apo,\bpo}\) in \(\ash\)-weak Popov form.
\end{problembox}

\begin{problembox}{IntBasis\textsubscript{2}}{int}
  \emph{Input:}
  pairwise distinct points \(\points \in \field^d\) with \(d\in\NN\),
  polynomials \((\apo,\bpo) \in \pR_{<d}^2\),
  shift \(\ash \in \ZZ\).

  \emph{Output:} a basis of \(\intmod{\points}{\apo,\bpo}\) in \(\ash\)-weak Popov form.
\end{problembox}

The main idea behind Beckermann and Labahn's recursion is summarized in
\cref{lem:dnc_relbas}, and its dominant computational steps are as follows:
\begin{itemize}[noitemsep,topsep=2pt]
  \item \emph{compute the residual:} use the basis obtained from the first
    recursive call to update the input polynomials, yielding the input for the
    second recursive call;
  \item \emph{multiply bases:} compute the product of the two bases obtained
    from the recursive calls.
\end{itemize}
These operations take several forms according to the tackled problem
(\cref{pbm:app} or \cref{pbm:int}) and to the designed algorithm, as outlined
in the next paragraphs. Their complexity depends on the degrees of the
polynomial entries in the computed bases, which are \(2\times 2\) matrices over
\(\pR\). All our analyses are based on the same upper bounds for these degrees.
Instead of using the general bound \(d\) for the degree of each of the four
entries, we use four individual degree bounds and observe and exploit the fact
that the sum of these bounds is always at most \(2d\), which is lower than the
naive bound by a factor of \(2\). Although considering heterogeneous degree bounds in
the bases prevents one from applying some multiplication techniques with
\(\cstmatmul<8\), this still proved beneficial for the worst-case analysis
compared to using the same naive bound \(d\) for all entries. For generic
input, the four entries all have similar degree close to \(d/2\), which
makes both steps a bit faster and leads to better leading constants.

For the original version of fast approximant bases
\cite{BeckermannLabahn1994,GiorgiJeannerodVillard2003}, the residual is
computed via middle products and the second step uses standard \(2\times 2\)
matrix multiplication with leading constant either \(\cstmatmul=8\) in general
due to heterogeneous degree bounds, or \(\cstmatmul \le 7\) for generic input.
These are the same operations that are used in the half-gcd algorithm, since
one can also exploit middle products there \cite{vdHoeven2025}. Our analysis
(\cref{sec:approx:pmbasis}) yields a leading constant of
\(\min(\frac{\cstmatmul+6}{4}, \frac{\mu}{2})\), which is \(7/2\) in general,
and \(2\) for generic input in the most favorable
case \(\cstmatmul=4\). Depending on the cases, these constants for
\cref{pbm:app} either match or are slightly lower than the best known ones
from \cite{vdHoeven2025} for \cref{pbm:xgcd}.
In \cref{sec:approx:appbasis}, we augment this classical approximant basis
algorithm with optimizations to handle possible valuations of the input
polynomials and to handle large shifts. As summarized in the next theorem, this
does not impact the above leading constant, but ensures that the efficiency is
sensitive to the shift \(\ash\), to the degrees of the input \(\apo\) and
\(\bpo\), and also to the degree their GCD (still without requiring the latter
degree to be known a priori).

\begin{theorem}
  \label{thm:appbasis:cx}%
  Let \(d \in \mathbb{N}\), \(\ash \in \mathbb{Z}\), and \((\apo,\bpo) \in
  \pR_{<d}^2\). If \(\apo=0\) or \(\bpo=0\), \cref{algo:AppBasis2} returns the
  \(\ash\)-Popov basis of \(\appmod{d}{\apo,\bpo}\) using \(\bigO{d}\)
  operations in \(\field\). Now suppose \(\apo\) and \(\bpo\) are both nonzero.
  Let \(\dreca = \min(d,\abs{\ash})\) be the minimum of the
  order and of the amplitude of the shift. Let \(n = \deg(\apo)\) and \(m =
  \deg(\bpo)\), and consider the positive integer
  \[
    \dsyz = \min(d, \;\; \max(m+n, 2n - \ash, 2m + \ash) + 1 - 2\deg(\gcd(\apo,\bpo))).
  \]
  Then, for \(\bar{d} = \max(0,\dsyz-\dreca)\), \cref{algo:AppBasis2}
  computes an \(\ash\)-weak Popov basis of \(\appmod{d}{\apo,\bpo}\) using
  \[
    {\textstyle\min(\frac{\cstmatmul+6}{4}, \frac{\mu}{2})} \, \timepm{\lceil \bar{d}/2 \rceil}\log_2(\lceil \bar{d}/2 \rceil)
    + \bigO{\timepm{d}}
  \]
  operations in \(\field\), where \(\cstmatmul = 8\) in the general case, and
  \(\cstmatmul \le 7\) is a constant for \(2\times 2\) polynomial matrix
  multiplication in the case where \((\apo \rem x^{\dsyz},\bpo \rem
  x^{\dsyz})\) satisfies the genericity condition of
  \cref{cor:approx_generic_mindeg}, that is,
  \(\bar\Delta_{\dsyz,\ash}\) does not vanish at the first \(\dsyz\)
  coefficients of \(\apo\) and \(\bpo\).
\end{theorem}

\begin{remark}
  For clarity, we only highlight the constant in front of the term which will
  most often be the dominant one; the other terms are constant multiples of
  \(\timepm{\bar{d}}\), \(\timepm{\dreca}\), and \(\timepm{d}\), without extra
  log factor, and are gathered in the \(\bigO{\timepm{d}}\). If seeking more
  information on the constant factors hidden by this \(\bigO{\timepm{d}}\), the
  reader may refer to
  \cref{prop:appbasis:gcdsensitive,prop:appbasis:shiftreduction}.
  \qed
\end{remark}

In particular, as explained in \cref{sec:rel_xgcd:gcd_syz}, this provides a
direct route to solve \problemName{pbm:xgcd} without the reductions mentioned
above. For this, one takes the shift \(\ash=n-m\) and the order \(d = m+n+1\).
Suppose \(n \ge m\); the other case is similar. Then \(\dreca = n-m\), \(\dsyz
= m+n+1-2\ell\) where \(\ell = \deg(\gcd(\apo,\bpo))\), and \(\bar{d} =
2m+1-2\ell\). Thus the above bound becomes
\({\textstyle\min(\frac{\cstmatmul+6}{4}, \frac{\mu}{2})}
\timepm{m-\ell+1}\log_2(m-\ell+1) + \bigO{\timepm{d}}\), which coincides with
the one reported in \cref{cor:xgcd,fig:constants_via_reductions} based on
reductions. In fact, it improves upon the latter: here, one does
not need any a priori knowledge about \(\ell = \deg(\gcd(\apo,\bpo))\), and the
cost really is sensitive to this degree, whereas in the bound via reductions
the quantity \(\ell\) is only a known lower bound on this degree. This refined
sensitivity is made possible by the above-mentioned optimizations described in
this article (without them, the main term would involve \(\lceil (m+n+1) / 2
\rceil\) instead of \(m-\ell+1\)).

For interpolants, using Beckermann and Labahn's approach on an arbitrary set of
points leads to a cost of \(\bigO{\timepm{d}\log(d)^2}\) operations in
\(\field\); one should then rather rely on the above reductions to avoid the
squared logarithmic factor. Thus, for \cref{pbm:int}, the focus of this paper
is on algorithms for the special sets of points introduced in
\cref{sec:intro:complexity}: points in geometric or arithmetic progression, and
radix-2 FFT points. One may also use other types of points; the efficiency will
then depend on how fast one can evaluate, interpolate, or extrapolate (see
below) at such points. We observe that the classical divide and conquer
algorithm yields leading constants that are very similar to the ones of
approximants (\cref{sec:interp:pmintbasis}). The main novelty here, leading to
improved complexity bounds, is the introduction of a variant of this algorithm,
which we call \nameref{algo:Eval-IntBasis2}, and which operates mainly on
evaluated representations: all relevant polynomials are known through their
evaluations at sufficiently many points among the input \(\points \in
\field^d\). The algorithm and its complexity analysis are presented in
\cref{sec:interp:eval_intbasis}. One of its core tools is extrapolation, which
consists in converting an evaluated representation of a polynomial at some
points to one at other points. Using fast extrapolation at the special points,
\nameref{algo:Eval-IntBasis2} avoids repeated conversions between
representations on monomial bases and on interpolation bases throughout the
recursion. In addition, working on evaluated representations helps in storing
and re-using the same evaluations for both main steps (the residual computation
and the basis multiplication) and, to some extent, across different recursive
levels; this may be seen as a systematic use of the folklore strategy of
``caching evaluations'', often carried out with FFT points
\cite[Sec\,2.9]{Bernstein2008-mult} \cite{vdHoeven2025}.

Detailed complexity bounds are provided in \cref{sec:interp:eval_intbasis}. The
resulting leading constants are summarized in the table of
\cref{fig:constants_via_reductions}, and are better than those for approximant
bases and than those previously known when solving the problem through
\nameref{pbm:xgcd}. As highlighted in the columns ``persp.'' of that table,
we also sketch some possible optimizations
(\cref{sec:interp:eval_intbasis_opti}), which bring a speedup factor around 1.5
by exploiting specific properties of the computed bases to save some of the
polynomial extrapolations.

\subsection{Perspectives}
\label{sec:intro:perspectives}

The most obvious perspective is to carry out the details of the latter
optimizations for solving Cauchy interpolation at special points
(\cref{sec:interp:eval_intbasis_opti}). In particular, one could prove the
genericity of the properties that are exploited to design the optimizations;
until now, we have observed this experimentally, checking that the properties
hold with very high probability on random input instances.

A second perspective is to develop, in the case of interpolants, optimizations
that are similar to those we describe for approximants in
\cref{sec:approx:appbasis}. We expect this to be feasible without major
obstacles. This would make the complexity more sensitive to the GCD degree
\(\deg(\gcd(\apo,\bpo))\), as discussed above for the approximant case.

Last, but not least, an important perspective is the design of software
implementations. We have done preliminary experiments, based on the FLINT
library \cite{flint}, and for a start focusing on prime finite fields of
cardinality \(\operatorname{card}(\field) < 2^{64}\). The observed running
times are promising and already improve upon FLINT's native XGCD and
Berlekamp-Massey implementations for such fields, even though we have only
integrated a subset of the optimizations developed in this paper. We think that
an in-depth study of implementations of the algorithms in this paper
should be made, to better assess their impact on the efficiency of practical
computations.

\section{Preliminaries}
\label{sec:preliminaries}

\subsection{Notation}
\label{sec:prelim:notation}

For a given integer \(d\in\NN\), we denote by \(\pR_{<d}\) (resp.\ \(\pR_{\le
d}\)) the set of univariate polynomials of degree less than \(d\) (resp.\ less
than or equal to \(d\)). For \(\apo \in \pR\), its valuation is written
\(\valuation{\apo}\) (with \(\valuation{0} = +\infty\)), and if \(\apo \neq 0\)
we write \(\lc{\apo}\) for its leading coefficient. For \(\apo, \modulus \in
\pR\) with \(\modulus\) nonzero, we write \(\apo \rem \modulus\) for the
remainder in the polynomial Euclidean division of \(\apo\) by \(\modulus\); we
give it higher precedence than addition but lower than multiplication (for
example, \(a + bc \rem \modulus = a + ((bc) \rem \modulus)\)).

Matrix spaces over \(\pR\) are denoted by \(\mpR{m}{n}\): for example
\(\mpR{2}{2}\) for square \(2 \times 2\) matrices and \(\mpR{1}{2}\) for row
vectors of length \(2\). The diagonal matrix with diagonal entries
\((\apo,\bpo) \in \pR^2\) is denoted by \(\diag{\apo,\bpo}\).
When considering properties of \(2\times 2\) polynomial matrices, we will
manipulate integers \(\ash \in \ZZ\) that we call shifts. The amplitude of such
a shift is defined as its absolute value \(\abs{\ash}\), and the shift
\(\ash\) is said to be uniform if \(\ash = 0\).

\begin{remark}
  This terminology comes from the polynomial matrix literature (see for example
  \cite{m_pade,BeckermannLabahnVillard1999}): There, for \(2 \times 2\)
  matrices, a shift would rather be a pair \((\ash_0,\ash_1) \in \ZZ^2\). Yet, in
  basically all occurrences of shifts, using the pair \((\ash_0,\ash_1)\) is
  equivalent to using \((\ash_0+u,\ash_1+u)\) for any \(u\in \ZZ\). This means
  that \((\ash_0,\ash_1)\) can be replaced by \((\ash_0-\ash_1,0)\), and for
  simplicity we choose to represent it by the single integer \(\ash =
  \ash_0-\ash_1\). The word ``uniform'' when \(\ash = 0\) refers to the fact that
  \(\ash_0=\ash_1\). Note that uniformity also occurs when the shift \(\ash = \dnum - \dden\) used
  for solving \cref{pbm:ratrecon} is zero (see \cref{sec:intro}),
  which corresponds to situations where the input degree constraints \(\dden\)
  and \(\dnum\) are uniform (i.e., \(\dden = \dnum\)).
  \qed
\end{remark}

\subsection{Modules and weak Popov forms}
\label{sec:prelim:polmat}

We recall basic facts about submodules of \(\pR^2\) and their bases; for more
details, the reader may refer to \cite{abstract_algebra}. The set of relations
we consider, in its most general form \(\relmod{\modulus}{\apo,\bpo}\) as in
\cref{eqn:relmod}, is a \(\pR\)-submodule of \(\pR^2\). It has rank \(2\),
since it contains \((\modulus,0)\) and \((0,\modulus)\). Any such rank \(2\)
submodule \(\mathcal{R} \subseteq \pR^2\) has bases, and they consist of two
\(\pR\)-linearly independent vectors \((p_{00},p_{01})\) and
\((p_{10},p_{11})\) both in \(\mathcal{R}\), which are conveniently represented
as a nonsingular \(2 \times 2\) polynomial matrix
\[
  \basis =
  \begin{bmatrix}
    p_{00} & p_{01} \\
    p_{10} & p_{11}
  \end{bmatrix}
  \in \mpR{2}{2}.
\]
Other bases of \(\mathcal{R}\) all have the form \(U P\) for some \(U \in
\mpR{2}{2}\) which is unimodular, that is, \(\det(U) \in
\field\setminus\{0\}\). In particular, all bases of \(\mathcal{R}\) have the
same determinant up to multiplication by a nonzero constant.

Fix a matrix \(P = [p_{ij}] \in \mpR{2}{2}\) with no zero row, and a shift
\(\ash\in\ZZ\). The \(\ash\)-row degree of \(P\) is the integer pair
\[
  \rdeg[\ash]{P} = (\rho_0,\rho_1) \;\;\;\text{where } \rho_i = \max(\deg(p_{i0}) + \ash, \deg(p_{i1}))
  ,
\]
and the \(\ash\)-pivot of the \(i\)th row of \(P\) is the rightmost
entry of that row for which the maximum \(\rho_i\) is reached. Following
\cite{Kailath1980,BeckermannLabahnVillard1999}, \(P\) is said to be
\begin{itemize}[noitemsep,topsep=2pt]
  \item in \(\ash\)-reduced form if \(\rho_0 + \rho_1 = \deg(\det(P)) + \ash\);
  \item in \(\ash\)-weak Popov form if the \(\ash\)-pivots of
    \(P\) are on its diagonal, which simply means that one has both
    \(\deg(p_{01}) < \deg(p_{00}) + \ash = \rho_0\) and
    \(\deg(p_{10}) + \ash \le \deg(p_{11}) = \rho_1\);
  \item in \(\ash\)-Popov form if the \(\ash\)-pivots of \(P\) are monic, on
    the diagonal, and furthermore one has both \(\deg(p_{10}) < \deg(p_{00})\)
    and \(\deg(p_{01}) < \deg(p_{11})\).
\end{itemize}
We have ordered these items by increasing strength: \(\ash\)-Popov implies
\(\ash\)-weak Popov, which implies \(\ash\)-reduced. Note that
the inequality \(\deg(\det(P)) + \ash \le \rho_0 + \rho_1\) always
holds, so \(\ash\)-reduced forms are those matrices for which it is
reached.

If the shift has large amplitude --- for example, beyond \(\deg(\det(P))\) is
always sufficient --- then being in \(\ash\)-Popov form implies being
triangular; lower or upper, depending on the sign of \(\ash\). These forms
are classically called the (lower or upper) Hermite normal forms. (See
\cref{cor:shift_amplitude} for a similar remark in the context of relation
bases.)

For a given rank \(2\) submodule \(\mathcal{R} \subseteq \pR^2\) and a given
shift \(\ash\), there exists a unique basis \(\basis \in \mpR{2}{2}\) of
\(\mathcal{R}\) which is in \(\ash\)-Popov form, and all other
\(\ash\)-weak Popov bases of \(\mathcal{R}\) have the same diagonal
degrees as \(\basis\) \cite{BeckermannLabahnVillard1999}. This invariant pair \((\deg(p_{00}),
\deg(p_{11}))\) is called the \(\ash\)-minimal degree of
\(\mathcal{R}\). It has a minimality property: for any
\(r = (r_0, r_1) \in \mathcal{R}\setminus\{0\}\), if the \(\ash\)-pivot
of \(r\) is \(r_0\) (resp.\ \(r_1\)) then one has \(\deg(r_0) \ge
\deg(p_{00})\) (resp.\ \(\deg(r_1) \ge \deg(p_{11})\)).

We present a short, useful lemma about degrees in \(\ash\)-weak Popov
forms. It will be instrumental in deriving sharp degree bounds for our
complexity analyses.

\begin{lemma}
  \label{lem:degree_bounds_2x2}%
  Let \(\ash\in\ZZ\) and let \(P = [p_{ij}] \in \mpR{2}{2}\) be a
  matrix in \(\ash\)-weak Popov form. Then, one has \(\deg(p_{01}) +
  \deg(p_{10}) < \deg(p_{00}) + \deg(p_{11}) = \deg(\det(P))\).
\end{lemma}
\begin{proof}
  For the equality, combine \(\rho_0 + \rho_1 = \deg(\det(P)) + \ash\)
  (from the fact that \(P\) is \(\ash\)-reduced) with \(\rho_0 =
  \deg(p_{00}) + \ash\) and \(\rho_1 = \deg(p_{11})\) (from the fact
  \(P\) is \(\ash\)-weak Popov). The inequality follows from
  \(
    \deg(p_{01}) + (\deg(p_{10}) + \ash)
    <
    (\deg(p_{00}) + \ash) + \deg(p_{11})
  \)
  (again using the fact \(P\) is \(\ash\)-weak Popov).
  \qed
\end{proof}

In this paper, we focus on computing weak Popov forms, as it is the forms
naturally returned by the customary divide and conquer approach presented in
\cref{lem:dnc_relbas}. In case one is interested in rather computing the
corresponding canonical Popov form, the next lemma highlights that it can be
deduced at a low computational cost. For completeness, we also consider the
transformation from reduced forms to weak Popov forms.

\begin{lemma}
  \label{lem:weak_to_popov_2x2}%
  Let \(\ash\in\ZZ\) and let \(P = [p_{ij}] \in \mpR{2}{2}_{<d}\).
  If \(P\) is in \(\ash\)-reduced form, a matrix
  which is left-unimodularly equivalent to \(P\) and in \(\ash\)-weak
  Popov form can be computed using \(\bigO{d}\) operations in \(\field\).
  If \(P\) is in \(\ash\)-weak Popov form, its
  \(\ash\)-Popov form can be computed using \(\bigO{\timepm{d}}\)
  operations in \(\field\).
\end{lemma}
\begin{proof}
  Finding an \(\ash\)-weak Popov form of an
  \(\ash\)-reduced matrix \(P\) amounts to finding a constant
  transformation \(U\) in \(\mKK{2}{2}\) via Gaussian elimination and
  left-multiplying \(P\) by a simple matrix built from \(U\) and the degrees in
  \(P\), and costs \(\bigO{d}\) field operations
  \cite[Sec.\,3]{SarkarStorjohann2011}  \cite[Sec.\,4]{NeigerPernet2021}.

  Now suppose \(P\) is in \(\ash\)-weak Popov form. Its
  \(\ash\)-pivots are thus already on the diagonal, and from now on
  we assume that they are monic, which can be ensured using
  \(\bigO{d}\) field operations. Denote by \((\rho_0,\rho_1)\) the
  \(\ash\)-row degree of \(P\). To begin with, assume that \(\rho_0 \le
  \rho_1\). Then \(\deg(p_{01}) < \deg(p_{11})\) is already satisfied,
  since we have \(\deg(p_{01}) < \rho_0 \le \rho_1 = \deg(p_{11})\).
  Now compute the quotient and remainder \((q,r)\) such that \(p_{10} =
  q p_{00} + r\) with \(\deg(r) < \deg(p_{00})\), and consider the unimodular
  transformation
  \[
    \begin{bmatrix}
      1 & 0 \\
      -q  & 1
    \end{bmatrix}
    P
    =
    Q,
    \quad\text{where we define }
    Q =
    \begin{bmatrix}
      p_{00} & p_{01} \\
      r & p_{11} - q p_{01}
    \end{bmatrix}
    \in \mpR{2}{2}.
  \]
  The key observation here is that \(\deg(q p_{01}) \le \deg(p_{10}) -
  \deg(p_{00}) + \deg(p_{01}) < \deg(p_{11})\)
  where the latter strict inequality follows from \cref{lem:degree_bounds_2x2}.
  This property ensures that \(p_{11} - q p_{01}\) is monic, that the second
  row of \(Q\) has its \(\ash\)-pivot on the diagonal, and that the
  constraint \(\deg(p_{01}) < \deg(p_{11}) = \deg(p_{11} - q p_{01})\) holds.
  Thus \(Q\) is the \(\ash\)-Popov form of \(P\). In terms of
  complexity, the division with remainder and the computation of  \(p_{11} - q
  p_{01}\) both cost \(\bigO{\timepm{d}}\) field operations.

  The other case \(\rho_1 < \rho_0\) is dealt with similarly, using
  \(\deg(p_{10}) + \ash \le \rho_1 < \rho_0 = \deg(p_{00}) + \ash\) to deduce
  \(\deg(p_{10}) < \deg(p_{00})\), and then performing the unimodular transformation
  \[
      \begin{bmatrix}
        1 & -q \\
        0  & 1
      \end{bmatrix}
      P
      =
      \begin{bmatrix}
        p_{00} - q p_{10} & r \\
        p_{10} & p_{11}
      \end{bmatrix}
      \in \mpR{2}{2},
    \]
    where \((q,r)\) are such that \(p_{01} = q p_{11} + r\) with \(\deg(r) <
    \deg(p_{11})\).
    \qed
\end{proof}

\subsection{Some basic results on relation bases}
\label{sec:prelim:relbas}

In the next lemma, we give a description of the basis in Hermite normal form of
\(\relmod{\modulus}{\apo,\bpo}\), from which we will deduce the determinant of
all bases of \(\relmod{\modulus}{\apo,\bpo}\) (\cref{cor:determinant_relmod}).
We also show that this Hermite basis equals the basis in \(\ash\)-Popov
form as soon as the amplitude of the shift is large enough. Thanks to this, we
highlight that one can always modify the input shift \(\ash\) into a
shift with amplitude at most \(d = \deg(\modulus)\) without loss of
generality (\cref{cor:shift_amplitude}).

\begin{lemma}
  \label{lem:relbas_hermite_bases}%
  Let \(\apo, \bpo \in \pR\) and let \(\modulus \in \pR\setminus\{0\}\) be
  monic. Let \(\alpha = \gcd(\apo,\modulus)\), \(\beta = \gcd(\bpo,\modulus)\),
  and \(\gamma = \gcd(\apo, \bpo,\modulus)\).
  \begin{itemize}[noitemsep,topsep=2pt]
    \item If \(\apo = 0 \bmod \modulus\), then the basis of
      \(\relmod{\modulus}{0,\bpo}\) in (both lower and upper) Hermite form is
      \(\diag{1,\modulus/\beta}\). It is also the basis in
      \(\ash\)-Popov form for any \(\ash \in \ZZ\).
    \item If \(\apo \neq 0 \bmod \modulus\), then the basis of
      \(\relmod{\modulus}{\apo,\bpo}\) in lower Hermite form is
      \[
        H =
        \begin{bmatrix}
          \modulus/\alpha & 0 \\
          \tilde{\bpo} & \alpha/\gamma
        \end{bmatrix}
        \in \mpR{2}{2},
      \]
      where \(\tilde{\bpo} = (-\tilde{\apo} \bpo / \gamma) \rem \modulus / \alpha\)
      and \(\tilde{\apo} = (\apo/\alpha)^{-1} \rem \modulus/\alpha\). It is
      also in \(\ash\)-Popov form for any shift \(\ash \le
      -\deg(\modulus/\alpha) + \deg(\alpha/\gamma) + 1\).
  \end{itemize}
\end{lemma}

\begin{proof}
  Let \(D = \diag{1,\modulus/\beta}\), which is obviously
  both triangular and in \(\ash\)-Popov form, for any shift \(\ash\). By definition of \(\beta\),
  the rows of \(D\) are in \(\relmod{\modulus}{0,\bpo}\) (note that
  \(\beta = \modulus\) if \(\bpo=0 \bmod \modulus\)). Furthermore, for any
  \((p,q) \in \relmod{\modulus}{0,\bpo}\), we have \(q \bpo = 0 \bmod
  \modulus\), and therefore \(\modulus/\beta\) divides \(q\) since
  \(\bpo/\beta\) and \(\modulus/\beta\) are coprime; thus \((p,q)\) is
  generated by \(D\).

  By construction, the rows of \(H\) are in \(\relmod{\modulus}{\apo, \bpo}\), and
  \(H\) is lower triangular with monic diagonal entries and
  \(\deg(\tilde{\bpo}) < \deg(\modulus/\alpha)\), so it is in Hermite form. The
  fact that it is in \(\ash\)-Popov form follows from \(\deg(\tilde{b}) + \ash
  \le \deg(\modulus/\alpha) - 1 + \ash \le \deg(\alpha/\gamma)\), thanks to our
  assumption on \(\ash\). It remains to show that any relation \((p_0, p_1) \in
  \relmod{\modulus}{\apo,\bpo}\) can be written \([p_0 \;\; p_1] =[q_0 \;\; q_1]
  H\) for some \(q_0,q_1 \in \pR\). By definition, \(\alpha/\gamma\) divides
  both \(\apo/\gamma\) and \(\modulus/\gamma\) but is coprime with
  \(\bpo/\gamma\). Thus, from the relation \(p_0\apo/\gamma + p_1\bpo/\gamma =
  0 \bmod \modulus/\gamma\) we can deduce that \(p_1 = q_1 \alpha/\gamma\) for
  some \(q_1 \in \pR\). Then, injecting this in \(p_0\apo + p_1\bpo = 0 \bmod
  \modulus\), we get \(p_0\apo + q_1\bpo \alpha/\gamma = 0 \bmod \modulus\) and
  therefore \(p_0\apo/\alpha + q_1\bpo/\gamma = 0 \bmod \modulus/\alpha\). This
  leads to \(p_0 = q_1\tilde{\bpo} \bmod \modulus/\alpha\), hence \(p_0 =
  q_0\frac{\modulus}{\alpha} + q_1\tilde{\bpo}\) for some \(q_0 \in \pR\).
  \qed
\end{proof}

The assumption that \(\modulus\) is monic is only here to avoid introducing
\(\lc{\modulus}\) in the above formulas; note that this is not restrictive
since \(\relmod{\modulus}{\apo,\bpo} =
\relmod{\modulus/\lc{\modulus}}{\apo,\bpo}\).

Symmetrically, the \(\ash\)-Popov basis of
\(\relmod{\modulus}{\apo,0}\) is \(\diag{\modulus/\alpha,1}\), and if \(\bpo
\neq 0 \bmod \modulus\) then the basis of \(\relmod{\modulus}{\apo,\bpo}\) in
upper Hermite form is
\[
  H =
  \begin{bmatrix}
    \beta/\gamma & \tilde{\apo} \\
    0            & \modulus/\beta
  \end{bmatrix}
  \in \mpR{2}{2},
\]
where \(\tilde{\apo} = (-\tilde{\bpo} \apo / \gamma) \rem \modulus / \beta\)
and \(\tilde{\bpo} = (\bpo/\beta)^{-1} \rem \modulus/\beta\). The latter matrix
\(H\) is also in \(\ash\)-Popov form for any \(\ash \ge
\deg(\modulus/\beta) - \deg(\beta/\gamma)\).

\begin{corollary}
  \label{cor:determinant_relmod}%
  Let \(\apo, \bpo, \modulus \in \pR\) with \(\modulus\) nonzero, and let
  \(\basis \in \mpR{2}{2}\) whose rows are in \(\relmod{\modulus}{\apo,\bpo}\).
  The matrix \(\basis\) is a basis of \(\relmod{\modulus}{\apo,\bpo}\) if and
  only if its determinant has the form \(\det(\basis) =
  \lambda\modulus/\gamma\) for some nonzero constant \(\lambda
  \in \field\), where \(\gamma = \gcd(\apo,\bpo,\modulus)\).
\end{corollary}

\begin{proof}
  In both cases \(\apo = 0 \bmod \modulus\) and \(\apo \neq 0 \bmod \modulus\),
  \cref{lem:relbas_hermite_bases} describes a basis \(H\) of
  \(\relmod{\modulus}{\apo,\bpo}\) whose determinant is \(\modulus/\gamma\).
  \qed
\end{proof}

\begin{corollary}
  \label{cor:shift_amplitude}%
  Let \(\apo, \bpo, \modulus \in \pR\) with \(\modulus\) nonzero of degree \(d\),
  and let \(\ash\) be a shift.
  If \(\ash \le -d\), any \(-d\)-weak Popov basis of
  \(\relmod{\modulus}{\apo,\bpo}\) is lower triangular and is in \(\ash\)-weak Popov form.
  %%
  % Assume \(\ash \le -\deg(\modulus/\alpha) +
  % 1\), and take any \(\tau \ge\deg(\modulus/\alpha) -
  % 1\) (in particular, this covers the case \(ash \le -d + 1\) and all
  % choices of \(\tau \ge d - 1\)). Any \(-\tau\)-weak Popov basis of
  % $\relmod{\modulus}{\apo,\bpo}$ is lower triangular and is also in \(\ash\)-weak Popov form.
\end{corollary}

\begin{proof}
  Let \(P = [p_{ij}] \in \mpR{2}{2}\) be a \(-d\)-weak Popov basis of
  \(\relmod{\modulus}{\apo,\bpo}\). Recall from \cref{cor:determinant_relmod}
  that \(\deg(\det(P)) \le d\), and from \cref{sec:prelim:polmat} that
  \(\deg(p_{00}) + \deg(p_{11}) \le \deg(\det(P))\). Thus one has
  \(\deg(p_{01}) < \deg(p_{00}) - d \le 0\), which yields \(p_{01} = 0\).
  Furthermore, from \(\deg(p_{10}) - d \le \deg(p_{11})\), we deduce
  \(\deg(p_{10}) + \ash \le \deg(p_{11})\).
  Hence \(P\) is lower triangular and in \(\ash\)-weak Popov form.
  \qed
\end{proof}

Similarly, if \(\ash \ge d+1\), any \((d+1)\)-weak Popov basis of
\(\relmod{\modulus}{\apo,\bpo}\) is upper triangular and is also in \(\ash\)-weak
Popov form. This shows that when the shift \(\ash\) has amplitude greater than
\(d\), one can always replace it by \(\ash = d+1\) or \(\ash = -d\), thus with
amplitude at most \(d+1\).

The following lemma will help us give refined complexity analyses for the
computation of approximants of \(\apo,\bpo\), in particular when these
polynomials have some (a priori unknown) nontrivial common factor
(\cref{sec:approx:appbasis:gcdsensitive}). It also leads to an XGCD algorithm
(\cref{sec:rel_xgcd:gcd_syz}).

\begin{lemma}
  \label{lem:kernel_approx}
  Let \((\apo,\bpo) \in \pR^2\), not both zero. The module of syzygies
  \(\{(\pp,\pq) \in \pR^2 \mid \pp \apo + \pq \bpo = 0\}\) has rank \(1\) and
  its bases are all nonzero constant multiples of \([-\bpo/\pg \;\;
  \apo/\pg]\), where \(\pg = \gcd(\apo,\bpo)\). Let \(\ash \in \ZZ\)
  and let \(\modulus\in \pR\setminus\{0\}\) be coprime with \(\pg\) and with
  degree greater than or equal to
  \begin{equation}
    \label{eqn:order_yielding_syz}%
    \dsyz = \max(m+n, 2n - \ash, 2m + \ash) + 1 -2\ell,
  \end{equation}
  where \(n = \deg(\apo)\), \(m = \deg(\bpo)\), \(\ell = \deg(\pg)\).
  Any \(\ash\)-weak Popov basis of \(\relmod{\modulus}{\apo,\bpo}\)
  contains one of these bases of syzygies, that is, it has a row of the form
  \([-\lambda\bpo/\pg \;\; \lambda\apo/\pg]\) for some
  \(\lambda\in\field\setminus\{0\}\). The same conclusion holds without the
  coprimeness assumption if \(\deg(\modulus) \ge \dsyz + \ell\).
\end{lemma}
\begin{proof}
  The first claim is a classical fact: \(S = [-\bpo/\pg \;\; \apo/\pg] \in
  \mpR{1}{2}\) is obviously in the module of syzygies, and proving that the
  latter is generated by \(S\) boils down to the fact that \(\apo/\pg\) and
  \(\bpo/\pg\) are coprime. Observe that this basis \(S\) is in
  \(\ash\)-weak Popov form, for any shift \(\ash\). For the rest
  of the proof, we consider the case where the \(\ash\)-pivot of \(S\)
  is its right-hand entry; the other case, when it is the left-hand entry, is
  proved by similar arguments.

  Let \(\basis\) be an \(\ash\)-weak Popov basis of
  \(\relmod{\modulus}{\apo,\bpo}\). Let \([\pp \;\; \pq]\) be the second row of
  \(\basis\), whose \(\ash\)-pivot is \(\pq\). As mentioned above, by
  minimality, the fact that \(S\) is in
  \(\relmod{\modulus}{\apo,\bpo}\) implies that \([\pp \;\; \pq]\) has
  \(\ash\)-pivot degree at most that of \(S\), that is, \(\deg(\pq) \le
  \deg(\apo/\pg) = n-\ell\). By definition of \(\ash\)-pivot, it follows
  that \(\deg(\pp) \le \deg(\pq) - \ash \le n-\ell - \ash\). As a
  result,
  \[
    \deg(\pp \apo/\pg + \pq \bpo/\pg)
    \le \max(2n-2\ell-\ash, m+n-2\ell)
    \le \dsyz
    < \deg(\modulus).
  \]
  On the other hand, we have \(\pp \apo/\pg + \pq \bpo/\pg = 0 \bmod \modulus\)
  since \([\pp \;\; \pq]\) is in \(\relmod{\modulus}{\apo,\bpo}\) and
  \(\modulus\) is coprime with \(\pg\). This implies that \(\pp \apo/\pg + \pq
  \bpo/\pg = 0\), hence \([\pp \;\; \pq] = [-\lambda\bpo/\pg \;\;
  \lambda\apo/\pg]\) for some \(\lambda\in\pR\setminus\{0\}\). It remains to
  observe that \(\deg(\lambda) = 0\), since we have seen above that \(\deg(\pq)
  \le \deg(\apo/\pg)\).

  In the case \(\deg(\modulus) \ge \dsyz+\ell\), we actually have \(\deg(\pp
  \apo + \pq \bpo) < \deg(\modulus)\), hence \(\pp \apo + \pq
  \bpo = 0\). This means that \(\pp \apo/\pg + \pq
  \bpo/\pg = 0\), and one can conclude as in the previous paragraph
  without requiring the assumption of coprimeness.
  \qed
\end{proof}

To conclude this section, we recall a recursion property for relation bases
\cite[Sec.\,5]{BeckermannLabahn97}
\cite[Sec.\,2.2]{JeannerodNeigerVillard2020}, in its general form which
requires the knowledge of a factorization of \(\modulus\). This is not a
restriction when \(\modulus\) is a power of \(x\) or when it is built from
known roots. As such, this property is at the core of most approximant and
interpolant basis algorithms in the literature, and also of those presented
in \cref{sec:approx,sec:interp}.

\begin{lemma}
  \label{lem:dnc_relbas}%
  Let \(\modulus,\modulus_1,\modulus_2 \in \pR \setminus \{0\}\) be nonzero polynomials such that
  \(\modulus = \modulus_1 \modulus_2\). For polynomials \((\apo,\bpo) \in
  \pR^2\) and a shift \(\ash \in \ZZ\), take the following steps:
  \begin{algorithmic}[1]
    \State
      let \(
      \basis =
      [\begin{smallmatrix}
        p_{00} & p_{01} \\
        p_{10} & p_{11}
      \end{smallmatrix}]
      \in \mpR{2}{2}
      \)
      be an \(\ash\)-weak Popov basis of
      \(\relmod{\modulus_1}{\apo,\bpo}\) \\
      \MyComment{first basis, modulo \(\modulus_1\)}
    \State
      let \(\apor = \frac{p_{00} \apo + p_{01} \bpo}{\modulus_1} \rem \modulus_2\)
      and \(\bpor = \frac{p_{10} \apo + p_{11} \bpo}{\modulus_1} \rem \modulus_2\)
      \MyComment{residual polynomials}
    \State
      let \(\ashr = \ash+\deg(p_{00}) - \deg(p_{11})\)
      \MyComment{updated shift}
    \State
      let \(
      \basisr
      \in \mpR{2}{2}
      \)
      be a \(\ashr\)-weak Popov basis of
      \(\relmod{\modulus_2}{\apor,\bpor}\)
      \MyComment{second basis, modulo \(\modulus_2\)}
  \end{algorithmic}
  Then the product \(\basisr \basis\) is an \(\ash\)-weak Popov basis of
  \(\relmod{\modulus}{\apo,\bpo}\). Furthermore, the \(\ash\)-minimal
  degree of \(\relmod{\modulus}{\apo,\bpo}\) is the sum
  \((\mathring{p}_{00}+\mathring{q}_{00},\mathring{p}_{11}+\mathring{q}_{11})\),
  where \((\mathring{p}_{00},\mathring{p}_{11})\) is the
  \(\ash\)-minimal degree of \(\relmod{\modulus_1}{\apo,\bpo}\) and
  \((\mathring{q}_{00},\mathring{q}_{11})\) is the \(\ashr\)-minimal
  degree of \(\relmod{\modulus_2}{\apor,\bpor}\).
\end{lemma}

In this framework, the main tasks beyond recursive calls is to efficiently
compute the residual polynomials and the product of bases. We will use
the next remark for interpolant bases.

\begin{remark}
  \label{rmk:modified_relbas}%
  For the residual polynomials, one can also pick any divisor \(\gamma\) of
  \(\modulus_1\) such that \(\modulus_1/\gamma\) is coprime with
  \(\modulus_2\), and take \(\apor = \frac{p_{00} \apo + p_{01} \bpo}{\gamma}
  \rem \modulus_2\) and \(\bpor = \frac{p_{10} \apo + p_{11} \bpo}{\gamma}
  \rem \modulus_2\). Indeed, since
  \(\modulus_1/\gamma\) is invertible modulo \(\modulus_2\), it is easily
  verified that \(\relmod{\modulus_2}{u, v} =
  \relmod{\modulus_2}{u \modulus_1/\gamma, v  \modulus_1/\gamma}\) for
  any \(u, v\in \pR\). In particular, if \(\modulus_1\) and
  \(\modulus_2\) are coprime, we can take \(\gamma=1\).
  \qed
\end{remark}

\section{Minimal relation bases: rational reconstruction and XGCD}
\label{sec:rel_xgcd}

In this section, we describe a reduction (\cref{sec:rel_xgcd:main_proof}) which
allows one to solve Problems~\Call{pbm:ratrecon}{} and~\Call{pbm:xgcd}{} via an
instance of \problemName{pbm:rel} with some freely chosen modulus. One can thus
reduce to instances which provide the best efficiency, such as approximant
bases or interpolant bases as discussed in \cref{sec:approx,sec:interp}.

We decompose this reduction into three building blocks, presented in
\cref{sec:rel_xgcd:rel_to_app,sec:rel_xgcd:rel_to_rel,sec:rel_xgcd:handle_dega}.
The first one reduces the computation of an \(\ash\)-weak
Popov basis of \(\relmod{\modulus}{\apo,-1}\) to that of a \(0\)-weak Popov
basis of approximants at order \(d-1-\ash\), where \(d = \deg(\modulus)\).
The second one shows how to reduce the computation of an \(\ash\)-weak
Popov basis of \(\relmod{\modulus}{\apo,-1}\) to that of an
\(\ash\)-weak Popov basis of relations modulo any \(\nodulus\) also of
degree \(d\), provided that \(\nodulus\) is coprime with \(\modulus\) and
\(\ash \ge 0\). The third one allows one to take into account the possible
discrepancy between \(\deg(\apo)\) and \(d\), leading to a complexity bound
whose main term depends on \(\deg(\apo)\) rather than \(d\). Finally,
we discuss the particular case of \problemName{pbm:xgcd} in
\cref{sec:rel_xgcd:xgcd_to_rel,sec:rel_xgcd:gcd_syz}.

\subsection{Reducing from relation bases to approximant bases}
%% , when \texorpdfstring{\(\bpo = -1\)}{b == -1}}
\label{sec:rel_xgcd:rel_to_app}

In this subsection, we show that an \(\ash\)-weak Popov basis of
\(\relmod{\modulus}{\apo,-1}\) can be deduced efficiently from some
\(0\)-weak Popov basis of approximants at order \(d-1-\ash\), where
\(d = \deg(\modulus)\). The complexity of computing such a relation basis thus
directly depends on the discrepancy between \(d\) and \(\ash\), which
agrees with intuition. Indeed, in this case with \(\bpo=-1\), shifts that
are positive and large should be easier to handle than the uniform
shift and than shifts with small amplitude. This is easily observed in the
extreme case \(\ash \ge d\): then the \(\ash\)-Popov basis of
\(\relmod{\modulus}{\apo,-1}\) is
\begin{equation}
  \label{eqn:relbas_bm1_hermite}
  H =
  \begin{bmatrix}
    1 & \apo \\
    0 & \modulus
  \end{bmatrix}
  \in \mpR{2}{2},
\end{equation}
which is known without any computation. More generally, when \(\ash\) is
close to \(d\), a small amount of work should suffice to transform the above
basis into an \(\ash\)-weak Popov one, since this can be done through
only a few steps of the extended Euclidean algorithm.

At the other side of the range of parameters, when \(\ash\) is large and negative,
the approximation order \(d-1-\ash\) can still be thought of as being less
than \(2d\). Indeed one may assume \(-\ash \le d\), since when \(-\ash > d\)
one may as well work with the shift \(-d\), as showed in
\cref{cor:shift_amplitude}.

\begin{proposition}
  \label{prop:ratrecon_reduction_general}%
  Let \(\modulus \in \pR\) be monic of degree \(d \in \NN\), let \(\apo \in
  \pR_{<d}\), and let \(\ash \in \ZZ\). The upper Hermite basis of
  \(\relmod{\modulus}{\apo,-1}\) is the matrix \(H\) in
  \cref{eqn:relbas_bm1_hermite}. If \(\ash \ge \deg(\apo) + 1\), this
  matrix \(H\) is in \(\ash\)-Popov form. If \(\ash =
  \deg(\apo)\), then \(H\) is in \(\ash\)-reduced form, and an
  \(\ash\)-weak Popov form of it is
  \([\begin{smallmatrix} q & -r \\ 1 & \apo \end{smallmatrix}] \in
  \mpR{2}{2}\) where \(q = x^{d - \deg(\apo)}/\lc{\apo}\) and \(r =
  \modulus - q\apo\).

  %% now the real stuff:
  If \(\ash < \deg(\apo)\) (which implies \(d - 1 - \ash > 0\)),
  consider the following construction:
  \begin{itemize}[noitemsep,topsep=2pt]
    \item define the reversals \(\bapo = x^{d-1} \apo(1/x)\) and \(\bmodulus = x^d \modulus(1/x)\);
    \item let \(\bar{P} = [\bar{p}_{ij}] \in \mpR{2}{2}\) be a \(0\)-weak Popov basis of \(\appmod{d-1-\ash}{\bapo, \bmodulus}\);
    \item define the reversals \(p_{00} = x^{\deg(\bar{p}_{00})} \bar{p}_{00}(1/x)\)
          and \(p_{10} = x^{\deg(\bar{p}_{11})+1} \bar{p}_{10}(1/x)\);
    \item let \(\apor = p_{00} \apo \rem \modulus\) and \(\bpor = p_{10} \apo \rem \modulus\)
      and finally construct
      \begin{equation}
        \label{eqn:ratrecon_reduction_mat}%
        Q =
        \begin{bmatrix}
          p_{00} & \apor \\
          p_{10} & \bpor
        \end{bmatrix}
        \in \mpR{2}{2};
      \end{equation}
  \end{itemize}
  then \(Q\) is an \(\ash\)-reduced basis of
  \(\relmod{\modulus}{\apo,-1}\).
\end{proposition}

\begin{proof}
  The first statement about the Hermite basis is a particular case of
  \cref{lem:relbas_hermite_bases} (more precisely, of the comment that follows
  this lemma). The facts that \(H\) is in \(\ash\)-Popov form when
  \(\ash \ge \deg(\apo)+1\) and in \(\ash\)-reduced form when
  \(\ash = \deg(\apo)+1\) are easily verified; in the latter case, it is
  also easily checked that the matrix \([\begin{smallmatrix} q & -r \\ 1 &
  \apo \end{smallmatrix}] \in \mpR{2}{2}\) is in \(\ash\)-weak Popov
  form and left-unimodularly equivalent to \(H\). In the rest of the proof we
  assume \(\ash < \deg(\apo)\).

  In particular, \(\apo \neq 0\), and the valuation of \(\bapo = x^{d-1}
  \apo(1/x)\) satisfies \(\valuation{\bapo} = d - 1 - \deg(\apo) < \bar{d}\),
  where \(\bar{d} = d - 1 - \ash \in \ZZp\).
  Let \(\bar{P} \in \mpR{2}{2}\) be a \(0\)-weak Popov basis of
  \(\appmod{\bar{d}}{\bapo, \bmodulus}\) and let \(\bapor, \bbpor \in \pR\) be
  the associated residuals:
  \begin{equation}
    \label{eqn:right_col_via_midprod}
    \bar{P} =
    \begin{bmatrix}
      \bar{p}_{00} & \bar{p}_{01} \\
      \bar{p}_{10} & \bar{p}_{11}
    \end{bmatrix}
    ,\quad
    \bapor = x^{-\bar{d}} (\bar{p}_{00} \bapo + \bar{p}_{01} \bmodulus)
    ,\quad
    \bbpor = x^{-\bar{d}} (\bar{p}_{10} \bapo + \bar{p}_{11} \bmodulus)
    .
  \end{equation}
  Let \(\bar{\rho}_0 = \deg(\bar{p}_{00})\)
  and \(\bar{\rho}_1 = \deg(\bar{p}_{11})\). The definition and properties
  recalled in \cref{sec:prelim:polmat,cor:determinant_relmod} ensure that
  \(\bar{\rho}_0 + \bar{\rho}_1 \le \bar{d}\), \(\bar{\rho}_0 >
  \deg(\bar{p}_{01})\), and \(\bar{\rho}_1 \ge \deg(\bar{p}_{10})\).
  Furthermore, the case \((\bar{\rho}_0, \bar{\rho}_1) = (0,\bar{d})\) is not
  possible here, as it would imply \(\bar{p}_{01} = 0\) and
  \(\deg(\bar{p}_{00})=0\) and thus \(\bapo = 0 \bmod x^{\bar{d}}\), which
  would contradict the above remark on the valuation of \(\bapo\).
  Hence \(\bar{\rho}_1 < \bar{d}\).

  Since \(\deg(\bapo) < d\) and \(\deg(\bmodulus) \le
  d\), we obtain the degree bounds \(\deg(\bapor) \le \bar{\rho}_0 + d - 1 -
  \bar{d} = \bar{\rho}_0 + \ash < d\) and
  \(\deg(\bbpor) \le \bar{\rho}_1 + d - \bar{d} = \bar{\rho}_1 + 1 + \ash
  < d\). We then define the reversals
  \[
    \apor = x^{\bar{\rho}_0 + \ash} \bapor(1/x)
    \quad\text{and}\quad
    \bpor = x^{\bar{\rho}_1 + 1 + \ash} \bbpor(1/x).
  \]
  We also define \(p_{ij}\) as the reversal of \(\bar{p}_{ij}\) according to
  the above stated degree bounds:
  \begin{align*}
    & p_{00} = x^{\bar{\rho}_0} \bar{p}_{00}(1/x)
    , \;\;
    p_{01} = x^{\bar{\rho}_0-1} \bar{p}_{01}(1/x)
    , \\
    & p_{10} = x^{\bar{\rho}_1+1} \bar{p}_{10}(1/x)
    , \;\;
    p_{11} = x^{\bar{\rho}_1} \bar{p}_{11}(1/x)
    .
  \end{align*}
  (Note that, for reversing \(\bar{p}_{10}\), we use \(\bar{\rho}_1+1\) rather
  than \(\bar{\rho}_1\).) It is then easily verified that \(\apor = p_{00} \apo
  + p_{01} \modulus\) and \(\bpor = p_{10} \apo + p_{11} \modulus\). Therefore
  both rows of the matrix \(Q\) defined above are in
  \(\relmod{\modulus}{\apo,-1}\). Note that these \(\apor\) and \(\bpor\) are
  indeed those from the statement, since  \(\deg(\apor) < d\)
  and  \(\deg(\bpor) < d\) holds by construction. To complete the proof, we are
  going to show that \(Q\) generates this \(\pR\)-module
  \(\relmod{\modulus}{\apo,-1}\), and that \(Q\) is in \(\ash\)-reduced
  form.

  Interpreting the above equations for \(\apor\) and \(\bpor\) as the \(2 \times
  2\) linear system \([\begin{smallmatrix} \apor \\ \bpor \end{smallmatrix}] =
  [\begin{smallmatrix} p_{00} & p_{01} \\ p_{10} & p_{11} \end{smallmatrix}]
[\begin{smallmatrix} \apo \\ \modulus \end{smallmatrix}]\), Cramer's rule shows that \(\det(Q) = (p_{00} p_{11} -
  p_{01} p_{10}) \modulus\). Since \(\bmodulus\) has nonzero constant
  coefficient, according to \cref{cor:determinant_relmod} we have
  \(\bar{p}_{00} \bar{p}_{11} - \bar{p}_{01}
  \bar{p}_{10} = \det(\bar{P}) = \lambda
  x^{\bar{d}}\) for some \(\lambda \in \field\setminus\{0\}\), and since
  \(\bar{P}\) is in weak Popov form we also have \(\bar{d} =
  \deg(\det(\bar{P})) = \bar{\rho}_0 + \bar{\rho}_1\). We deduce that
  \begin{align*}
    p_{00} p_{11} - p_{01} p_{10} & =
    x^{\bar{\rho}_0} \bar{p}_{00}(1/x) x^{\bar{\rho}_1} \bar{p}_{11}(1/x) -
    x^{\bar{\rho}_0-1} \bar{p}_{01}(1/x) x^{\bar{\rho}_1+1} \bar{p}_{10}(1/x) \\
                                  & = x^{\bar{d}} \det(\bar{P}(1/x))
          % x^{\bar{d}} (\bar{p}_{00}(1/x) \bar{p}_{11}(1/x) - \bar{p}_{01}(1/x) \bar{p}_{10}(1/x))
          = \lambda.
  \end{align*}
  Hence \(\det(Q) = \lambda \modulus\), which implies that \(Q\) is a basis of
  \(\relmod{\modulus}{\apo,-1}\) according to \cref{cor:determinant_relmod}.

  Finally, to prove that \(Q\) is in \(\ash\)-reduced form, we write
  \((\rho_0,\rho_1)\) for its \(\ash\)-row degree and we rely on
  the definition (see \cref{sec:prelim:polmat}): since \(\deg(\det(Q)) = d\),
  it is sufficient to prove that \(\rho_0 + \rho_1 \le d + \ash\).
  From the degree bounds on \(p_{00}\) and \(\apor\), we deduce that
  \(\deg(p_{00}) + \ash\) and \(\deg(\apor)\) are both less than or
  equal to \(\bar{\rho}_0 + \ash\), whence \(\rho_0 \le \bar{\rho}_0 + \ash\).
  Similarly, we deduce that \(\deg(p_{10}) + \ash\) and \(\deg(\bpor)\)
  are both less than or equal to \(\bar{\rho}_1 + 1 + \ash\), whence \(\rho_1
  \le \bar{\rho}_1 + 1 + \ash\). This yields \(\rho_0 + \rho_1 \le
  \bar{\rho}_0 + \bar{\rho}_1 + 1 + 2\ash = \bar{d} + 1 + 2\ash = d + \ash\).
  \qed
\end{proof}

Apart from computing \(\bar{\basis}\), the only task which involves algebraic
operations is to find \(\apor\) and \(\bpor\). The statement and the proof
offer two different ways to compute them, which both cost
\(\bigO{\timepm{d}}\). For an arbitrary \(\modulus\), we expect that going
through the computation of \(\bapor\) and \(\bbpor\) as in the proof would lead
to a slightly better leading constant via the use of middle products.

%% NOTE (kept in case, but this is a bit redundant)
% Observe that the approximant instance satisfies the requirements of the
% subsequent reduction, to be presented in the next subsection. Indeed, it is
% with an input shift \(0\) which is nonnegative and with an
% input second polynomial \(\bpo=-1\); while the latter is not directly visible,
% the fact that \(\bmodulus\) has valuation \(0\) ensures that
% \(\appmod{\bar{d}}{\bapo, \bmodulus} = \appmod{\bar{d}}{\hat{\apo}, -1}\),
% where \(\hat{\apo} = -\bapo\bmodulus^{-1} \rem x^{\bar{d}}\).

The obtained approximant basis instance benefits from \(\apo\) having small
degree: \(\deg(\apo) < d-1\) means that \(\bapo\) has positive valuation, and
then one can reduce the order from \(d-1-\ash\) to about \(2\deg(\apo) - d
- \ash\) using the ideas in \cref{sec:approx:appbasis:optimizations}. We
will not detail this further, as in
\cref{sec:rel_xgcd:handle_dega,sec:rel_xgcd:main_proof} we describe how to
handle such a degree discrepancy in a similar way but before applying the above
reduction.
%% NOTE slightly more detailed:
%% Indeed, when \(\bapo\) has valuation \(v_{\bapo} = d-1-\deg(\apo) > 0\),
%% and computing the a \(0\)-weak Popov basis of \(\appmod{\bar{d}}{\bapo,
%% \bmodulus}\) is the same as computing a \(-v_{\bapo}\)-weak Popov basis of
%% \(\appmod{\deg(\apo)-\ash}{x^{-v_{\bapo}}\bapo, \bmodulus}\) (see
%% \cref{sec:approx:appbasis:optimizations}).
%% Then one would apply some Newton inversion to compute an approx at order v_\bapo+1
%% and bring the shift back to (1,0).
%% The remaining order is
%%   ~ d - 1 -\ash - 2v_\bapo ~ 2\deg(\apo) - d - \ash
%% similar to the other method in sec:rel_xgcd:handle_dega

\subsection{Reducing from modulus \texorpdfstring{\(\modulus\)}{M}
            to modulus \texorpdfstring{\(\nodulus\)}{N},
            when the shift is nonnegative}
\label{sec:rel_xgcd:rel_to_rel}

We now present the second main piece of our reduction algorithm
(\cref{algo:RelationBasis-via-modN}): we have a target modulus \(\modulus\) and
an auxiliary modulus \(\nodulus\), and we show how to find an
\(\ash\)-weak Popov basis of \(\relmod{\modulus}{\apo,-1}\) from one of
\(\relmod{\nodulus}{\bapo, \bbpo}\), for some polynomials \(\bapo\) and
\(\bbpo\) described below.
Our method has several requirements, which prevent from switching between
moduli in full generality; specifically, the shift must be nonnegative,
and the moduli \(\modulus\) and \(\nodulus\) should have the same
degree and be coprime. As we will observe in
\cref{algo:RelationBasis-via-modN}, these requirements are not an issue when we
combine this with the first piece presented in \cref{sec:rel_xgcd:rel_to_app}.

\begin{proposition}
  \label{prop:ratrecon_reduction_niceshift}%
  Let \(\modulus \in \pR\) be monic of degree \(d \in \ZZp\), let \(\apo \in
  \pR_{<d}\), and let \(\ash \in \ZZ\) with \(\ash \ge 0\).
  Let \(\nodulus\) be any monic polynomial in \(\pR\) of degree \(d\) which is
  coprime with \(\modulus\). Define the polynomials \(\bapo = (\apo \nodulus)
  \rem \modulus\) and \(\bbpo = \modulus - \nodulus\), both of degree \(< d\),
  and let \(\basis \in \mpR{2}{2}\) be an \(\ash\)-weak Popov basis of
  \(\relmod{\nodulus}{\bapo, \bbpo}\). Define also
  \[
    \basisr =
    \begin{bmatrix}
      p_{00} & p_{01} + \apor \\
      p_{10} & p_{11} + \bpor
    \end{bmatrix}
    \text{ where }
    \basis =
    \begin{bmatrix}
      p_{00} & p_{01} \\
      p_{10} & p_{11}
    \end{bmatrix},\;
    \apor = \frac{p_{00} \bapo + p_{01} \bbpo}{\nodulus},\;
    \bpor = \frac{p_{10} \bapo + p_{11} \bbpo}{\nodulus}.
  \]
  Then \(\basisr\) is an \(\ash\)-weak Popov basis of
  \(\relmod{\modulus}{\apo,-1}\).
\end{proposition}

\begin{proof}
  By construction of \(\apor\), we have \(\nodulus \apor = p_{00}
  \bapo + p_{01} \bbpo\). Thus, by definition of \(\bapo\) and \(\bbpo\), we
  have \(\apor \nodulus = (p_{00} \apo \nodulus - p_{01} \nodulus) \bmod
  \modulus\), and since \(\modulus\) and \(\nodulus\) are coprime we can
  multiply by \(\nodulus^{-1} \bmod \modulus\) to obtain \(p_{00} \apo -
  (p_{01} + \apor) = 0 \bmod \modulus\). This means that the first row of
  \(\basisr\) is in \(\relmod{\modulus}{\apo,-1}\). Considering the definition
  of \(\bpor\) shows similarly that the second row of \(\basisr\) is in
  \(\relmod{\modulus}{\apo,-1}\).

  We now prove that \(\basisr\) is a basis of \(\relmod{\modulus}{\apo,-1}\).
  From \cref{cor:determinant_relmod} and the fact that \(\bbpo = \modulus -
  \nodulus\) is coprime with \(\nodulus\), we obtain that \(\det(\basis) =
  \lambda \nodulus\) for some nonzero constant \(\lambda \in
  \field\setminus\{0\}\). Using Cramer's rule on \(\basis [\begin{smallmatrix} \bapo \\ \bbpo
  \end{smallmatrix}] = [\begin{smallmatrix} \apor \nodulus \\ \bpor \nodulus \end{smallmatrix}]\),
  we get
  \[
    \modulus - \nodulus = \bbpo =
    \det\left(
      \begin{bmatrix}
        p_{00} & \apor \nodulus \\
        p_{10} & \bpor \nodulus
      \end{bmatrix}
    \right)
    / \det(\basis)
    =
    \det\left(
      \begin{bmatrix}
        p_{00} & \apor \\
        p_{10} & \bpor
      \end{bmatrix}
    \right)
    / \lambda.
  \]
  We deduce
  \[
    \det(\basisr) = \det(\basis) + \det\left(
      \begin{bmatrix}
        p_{00} & \apor \\
        p_{10} & \bpor
      \end{bmatrix}
    \right)
    =
    \lambda \nodulus + \lambda (\modulus - \nodulus)
    = \lambda \modulus.
  \]
  By \cref{cor:determinant_relmod}, this ensures that \(\basisr\) is
  a basis of \(\relmod{\modulus}{\apo,-1}\).

  Recall that \(\deg(\bapo) < d\), \(\deg(\bbpo) < d\), and \(d =
  \deg(\nodulus)\). Thus by definition of \(\apor\) and \(\bpor\), we get
  \[
    \deg(\apor) < \max(\deg(p_{00}), \deg(p_{01}))
    \quad\text{and}\quad
    \deg(\bpor) < \max(\deg(p_{10}), \deg(p_{11})).
  \]
  Since \(\basis\) is in \(\ash\)-weak Popov form with \(\ash \ge
  0\), we have \(\deg(p_{11}) \ge \deg(p_{10}) + \ash \ge
  \deg(p_{10})\). As a result, \(\deg(p_{11} + \bpor) = \deg(p_{11})\), and the
  \(\ash\)-pivot of the second row of \(\basisr\) is its right-hand
  entry. For the first row, since \(\basis\) is in \(\ash\)-weak Popov
  form we have \(\deg(p_{01}) < \deg(p_{00}) + \ash\), and the above bound
  on \(\deg(\apor)\) yields \(\deg(\apor) < \deg(p_{00}) + \ash\). Together, these
  inequalities ensure that
  \[
    \deg(p_{01} + \apor) \le \max(\deg(p_{01}), \deg(\apor))
    < \deg(p_{00}) + \ash,
  \]
  which concludes the proof.
  \qed
\end{proof}

Apart from the computation of \(\basis\), the algebraic operations in this
reduction consist in computing \((\bapo,\bbpo)\) as well as the second column
of \(\basisr\); both can be done in \(\bigO{\timepm{d}}\). Similarly to the
previous subsection, one can find the second column of \(\basisr\) either by
computing \(\apor,\bpor\) with the formulas in the statement, or as the product
of its first column by \(\apo\) modulo \(\modulus\). The involved leading
constants and how they compare is dependent on the properties of the moduli
\(\modulus\) and \(\nodulus\).

\subsection{Handling degree discrepancies}
\label{sec:rel_xgcd:handle_dega}

When considering a relation reconstruction instance, there may be a gap between
\(\deg(\apo)\) and \(d = \deg(\modulus)\), even when the instance
\((\apo,\modulus)\) is generic in the sense that all degrees are the expected
ones during the computation. In this case, it is interesting to have a
complexity bound where the main term \(\cstgeneral\, \timepm{d}\log(d)\) is
replaced by \(\cstgeneral\, \timepm{\deg(\apo)}\log(\deg(\apo)) +
\bigO{\timepm{d}}\). For this reason, similarly to performing a first step of
the extended Euclidean algorithm, we provide a way to compute a basis of
\(\relmod{\modulus}{\apo,-1}\) from one of \(\relmod{\apo}{r,-1}\), where \(r =
\modulus \rem \apo\). This is only intended as some input preprocessing that
is not repeated further: one typically has \(\deg(r) = \deg(\apo) - 1\), so
that computing recursively \(\relmod{\apo}{r,-1}\) from \(\relmod{r}{\apo \rem
r, -1}\) would lead to a quadratic-time process akin to the classical extended
Euclidean algorithm.

\begin{proposition}
  \label{prop:handle_dega}%
  Let \(\modulus \in \pR\) be monic of degree \(d \in \ZZp\), let \(\apo \in
  \pR_{<d}\), and let \(\ash \in \ZZ\). Assume \(\ash \le
  \deg(\apo)\); in particular, \(\apo\neq 0\). Let \((q,r)\) be the quotient
  and remainder in the division \(\modulus = q \apo + r\), with \(\deg(r) <
  \deg(\apo)\). Let \(\ashr = \ash + d - \deg(\apo)\).

  The matrix \(B = [\begin{smallmatrix} q & -r \\ 1 & \apo
  \end{smallmatrix}] \in \mpR{2}{2}\) is a basis of
  \(\relmod{\modulus}{\apo,-1}\). If \(\ashr > \deg(r)\), then \(B\) is
  in \(\ash\)-weak Popov form. If \(\ashr = \deg(r)\), then \(B\) is
  in \(\ash\)-reduced form, and the matrix \([\begin{smallmatrix} 1 + u
    q & \apo - u r \\ q & -r
  \end{smallmatrix}]\) is an \(\ash\)-weak Popov form of it, where \(u
  = x^{\deg(\apo) - \deg(r)} \lc{\apo} / \lc{r}\).

  Let \(P = [p_{ij}] \in \mpR{2}{2}\) be a \(\ashr\)-weak
  Popov basis of \(\relmod{\apo}{r,-1}\). Then, the matrix
  \[
    Q =
    \begin{bmatrix}
      p_{00} q + \apor & -p_{01} \\
      p_{10} q + \bpor & -p_{11}
    \end{bmatrix}
    \in \mpR{2}{2}
    \quad \text{where }
    \apor = \frac{p_{00} r - p_{01}}{\apo}
    \text{ and }
    \bpor = \frac{p_{10} r - p_{11}}{\apo}.
  \]
  is an \(\ash\)-weak Popov basis of \(\relmod{\modulus}{\apo,-1}\).
\end{proposition}

\begin{proof}
  Both rows of the matrix \(B\) are in \(\relmod{\modulus}{\apo,-1}\). Since
  \(\det(B) = q\apo + r = \modulus\), \cref{cor:determinant_relmod} ensures
  that \(B\) is a basis of \(\relmod{\modulus}{\apo,-1}\). Recall \(\ash
  \le \deg(\apo)\) and assume that \(\ashr \ge \deg(r)\). These
  inequalities imply \(\deg(1) + \ash \le \deg(\apo)\) and \(\deg(q) +
  \ash = \ashr \ge \deg(r)\). In particular, when \(\ashr >
  \deg(r)\), the latter inequality is strict and \(B\) is in
  \(\ash\)-weak Popov form. More generally, the \(\ash\)-row
  degree of \(B\) is therefore \((\deg(q) + \ash, \deg(\apo))\), which
  sums to \(\deg(q\apo) + \ash = \deg(\det(B)) + \ash\), ensuring
  that \(B\) is \(\ash\)-reduced. Suppose now \(\ashr = \deg(r)\).
  The matrix \(\bar{B} = [\begin{smallmatrix} 1 + u q & \apo - u r \\ q & -r
    \end{smallmatrix}] = [\begin{smallmatrix} 1 & u \\ 0 & 1
  \end{smallmatrix}]B\) is unimodularly equivalent to \(B\). Its second row has
  its \(\ash\)-pivot on the diagonal, as showed above. By construction,
  \(\deg(uq) = d - \deg(r) > 0\) and \(\deg(\apo - u r) < \deg(\apo)\). Hence
  \(\deg(1+uq) + \ash = d - \deg(r) + \ash = \deg(\apo) > \deg(\apo - u r)\),
  so the first row of \(\bar{B}\) has its \(\ash\)-pivot on
  the diagonal as well.

  We now forget the assumption \(\ashr \ge \deg(r)\) and turn to the last
  claim of the statement. The first row of \(Q\) satisfies
  \(
  (p_{00} q + \apor) \apo - (-p_{01})
  = p_{00} q \apo + p_{00} r - p_{01} + p_{01}
  % = p_{00} (q \apo + r)
  = p_{00} \modulus
  \),
  and thus it is in \(\relmod{\modulus}{\apo,-1}\). Similarly, the second row
  of \(Q\) is also in \(\relmod{\modulus}{\apo,-1}\).
  To prove that \(Q\) is a basis of \(\relmod{\modulus}{\apo,-1}\), we rely on
  \cref{cor:determinant_relmod} and show that \(\det(Q)\) is \(\modulus\) up to
  a constant multiple. The same corollary ensures \(\det(P) =
  \lambda\apo\) for some \(\lambda\in\field\setminus\{0\}\). Besides,
  Cramer's rule on the \(2\times 2\) system \(P [\begin{smallmatrix} r \\ -1
    \end{smallmatrix}] = [\begin{smallmatrix} \apor \apo \\ \bpor \apo
  \end{smallmatrix}]\) states that \(\det(P) r = \lambda \apo r =
  \det([\begin{smallmatrix} \apor\apo & p_{01} \\ \bpor\apo & p_{11}
  \end{smallmatrix}])\), hence
  \[
    \det(Q) =
    \det\left(
      \begin{bmatrix}
        p_{00} q + \apor & -p_{01} \\
        p_{10} q + \bpor & -p_{11}
      \end{bmatrix}
    \right)
    =
    - q \det(P) -
    \det\left(
      \begin{bmatrix}
        \apor & p_{01} \\
        \bpor & p_{11}
      \end{bmatrix}
    \right)
    =
    -\lambda(q\apo + r)
    = -\lambda\modulus.
  \]

  It remains to prove that \(Q\) is in \(\ash\)-weak Popov form.
  Concerning the first row, since \(P\) is in \(\ashr\)-weak Popov form
  we have \(\deg(p_{00}) + \ashr > \deg(p_{01})\), hence \(\deg(p_{01})
  < \deg(p_{00}) + \ashr \le \deg(p_{00}) + d\) thanks to our assumption
  \(\ash \le \deg(\apo)\). Thus \(\deg(p_{01}) < \deg(p_{00}\modulus)\),
  which ensures \(\deg(p_{00} \modulus - p_{01}) = \deg(p_{00}) + d\). Since
  \(p_{00} \modulus - p_{01} = p_{00} q \apo + \apor \apo\), we deduce
  \[
    \deg(p_{00} q + \apor)
    = \deg(p_{00}) + d - \deg(\apo)
    = \deg(p_{00}) + \ashr - \ash
    > \deg(p_{01}) - \ash,
  \]
  which shows that the \(\ash\)-pivot of the first row of \(Q\) is its
  leftmost entry.

  For the second row, assume first \(p_{10} \neq 0\). Then similarly,
  from \(p_{10} q \apo + \bpor \apo = p_{10} \modulus - p_{11}\)
  and \(\deg(p_{11}) \le \deg(\apo) < d\) we get
  \(\deg(p_{10} q + \bpor) = \deg(p_{10}) + d - \deg(\apo)\);
  then the fact that \(P\) is in \(\ashr\)-weak Popov form leads to
  \(\deg(p_{10} q + \bpor) + \ash = \deg(p_{10}) + \ashr \le \deg(p_{11})\),
  meaning that the \(\ash\)-pivot of the second row of \(Q\) is
  its rightmost entry. We now suppose \(p_{10} = 0\) and seek the same
  conclusion. From \(p_{10} r - p_{11} = 0
  \bmod \apo\) and \(\deg(p_{11}) \le \deg(\apo)\), it follows that \(p_{11} =
  \lambda \apo\) for some \(\lambda\in\field\setminus\{0\}\), and \(\bpor =
  -\lambda\). Thus \(\deg(p_{10}q + \bpor) = \deg(\bpor) = 0\)
  and our assumption \(\ash \le \deg(\apo)\) implies
  \(\deg(p_{10}q + \bpor) + \ash \le \deg(\apo) = \deg(p_{11})\).
  \qed
\end{proof}

\subsection{Reduction algorithm and proof of Theorem~\ref{thm:rel}}
\label{sec:rel_xgcd:main_proof}

We gather the results of
\crefrange{sec:rel_xgcd:rel_to_app}{sec:rel_xgcd:handle_dega} into
\cref{algo:RelationBasis-via-modN} which proves \cref{thm:rel}.

\begin{algorithm}[t]
  \algoCaptionLabel{RelationBasis-via-modN}{\modulus, \apo, \ash, \nodulus}
  \begin{algorithmic}[1]

    \Require \(\modulus \in \pR\) monic of degree \(d \in \NN\),
             polynomial \(\apo \in \pR_{<d}\),
             shift \(\ash \in \mathbb{Z}\),
             \(\nodulus \in \pR\) monic of degree \(\delta \ge \dred
             = 2\deg(\apo) - d - \ash - 1 + \min(0,\ash+d)\)
             with \(\valuation{\nodulus} = 0\) or \(\nodulus = x^{\delta}\)

    \Ensure a basis \(P \in \mpR{2}{2}_{\le d}\) of \(\relmod{\modulus}{\apo,-1}\) in \(\ash\)-weak Popov form

    \Assume some algorithm \(\textproc{RelationBasis-modN}(\nodulus, \apo, \bpo, \ash)\)
    which solves \cref{pbm:rel}

    \LComment{\cref{prop:handle_dega}: handle degree discrepancy}

    \State \(\ashr \gets \ash + d - \deg(\apo)\); \(\bshr \gets \min(0, \ash+d)\)
    \MyComment{we seek an \((\ash-\bshr)\)-weak Popov basis (\cref{cor:shift_amplitude})}
    \label{step:rel:shift}

    \State\InlineIf{\(\ash - \bshr > \deg(\apo)\)}{\Return \([\begin{smallmatrix} 1 & \apo \\ 0 & \modulus \end{smallmatrix}]\)}
    \label{step:rel:large_amp_1}

    \State \((q,r) \gets\) division with remainder \(\modulus = q \apo + r\) and \(\deg(r) < \deg(\apo)\)
    \label{step:rel:divrem}

    \State\InlineIf{\(\ashr - \bshr > \deg(r)\)}{\Return \([\begin{smallmatrix} q & -r \\ 1 & a \end{smallmatrix}]\)}
    \MyComment{after this, \(r \neq 0\)}
    \label{step:rel:large_amp_2}

    \State\InlineIf{\(\ashr - \bshr = \deg(r)\)}
    {
      \(u \gets x^{\deg(\apo) - \deg(r)} \lc{\apo} / \lc{r}\);
      \Return \([\begin{smallmatrix} 1 + u q & \apo - u r \\ q & -r \end{smallmatrix}]\)
    }
    \label{step:rel:large_amp_3}

    \State \(\modulus_1, d_1, \apo_1, \ash_1 \gets \apo, \deg(\apo), r, \ashr - \bshr\)
    \label{step:rel:reds_start}
    \label{step:rel:to_red1}
    \MyComment{\(-d_1 \le \ash_1 < \deg(\apo_1)\)}

    \LComment{\cref{prop:ratrecon_reduction_general}: reduce to approximants via reversals}

    \State \(\bapo \gets x^{d_1-1} \apo_1(1/x)\);
    \(\bmodulus \gets x^{d_1} \modulus_1(1/x)\)

    \State \(d_2 \gets d_1 - 1 - \ash_1\);
    \MyComment{\(0 < d_2 = \dred \le \min(2d_1-1, \delta)\)}

    \State \(\modulus_2, \apo_2, \ash_2, \gets x^{d_2}, -\bapo \bmodulus^{-1} \rem x^{d_2}, 0\)
    \label{step:rel:to_red2}

    \LComment{\cref{prop:ratrecon_reduction_niceshift}: reduce to modulus \(\nodulus\)}

    \State \(\modulus_3, \apo_3, \ash_3 \gets x^{\delta}, x^{\delta-d_2} \apo_2, \delta-d_2\)
    \MyComment{increase modulus degree from \(d_2\) to \(\delta\) (\cref{lem:approx_increase_order})}
    \label{step:rel:to_red3}

    \State \(\apo_4, \bpo_4, \ash_4 \gets\)
    \InlineIfElse{\(\nodulus = x^{\delta}\)}
    {\(\apo_3, -1, \ash_3\)}
    {\(\apo_3 \nodulus \rem x^{\delta}, x^{\delta} - \nodulus, \ash_3\)}
    \label{step:rel:to_red4}

    \LComment{call relations modulo \(\nodulus\) and deduce the output basis}

    \State \(\basis_4 = [g_{ij}] \gets \textproc{RelationBasis-modN}(\nodulus, \apo_4, \bpo_4, \ash_4)\)
    \MyComment{\(\ash_4\)-weak Popov basis of \(\relmod{\nodulus}{\apo_4,\bpo_4}\)}
    \label{step:rel:from_red4}

    \State \(\basis_2 = [f_{ij}] \gets
           \begin{bmatrix}
             g_{00} & g_{00}\apo_2 \rem x^{d_2} \\
             g_{10} & g_{10}\apo_2 \rem x^{d_2}
           \end{bmatrix}
    \)
    \label{step:rel:basis2}
    \MyComment{\(0\)-weak Popov basis of \(\appmod{d_2}{\apo_2,-1}\)}
    \label{step:rel:from_red2}

    \State \(q_{00} \gets x^{\deg(f_{00})} f_{00}(1/x)\);
    \(q_{10} \gets x^{\deg(f_{11})+1} f_{10}(1/x)\);
    \Statex \(\basisr \gets
           \begin{bmatrix}
             q_{00} & q_{00} \apo_1 \rem \modulus_1 \\
             q_{10} & q_{10} \apo_1 \rem \modulus_1
           \end{bmatrix}
    \)
    \MyComment{\(\ash_1\)-reduced basis of \(\relmod{\modulus_1}{\apo_1,-1}\)}
    \Statex \(\basis_1 = [p_{ij}] \gets\) an \(\ash_1\)-weak Popov form of \(\basisr\)
    \MyComment{costs \(\bigO{d}\), see \cref{lem:weak_to_popov_2x2}}
    \label{step:rel:from_red1}

    \State
    \(\apor \gets (p_{00} r - p_{01})/\apo\);
    \(\bpor \gets (p_{10} r - p_{11})/\apo\)
    \Statex
    \(\basis \gets
    \begin{bmatrix}
      p_{00} q + \apor & -p_{01} \\
      p_{10} q + \bpor & -p_{11}
    \end{bmatrix}\)
    \MyComment{\((\ash-\bshr)\)-weak Popov basis of \(\relmod{\modulus}{\apo,-1}\)}
    \label{step:rel:from_red}

    \State \Return \(\basis\)
    \label{step:rel:reds_end}
  \end{algorithmic}
\end{algorithm}

Define the integers  \(\ashr = \ash + d - \deg(\apo)\) and \(\bshr =
\min(0,\ash+d)\) as in \cref{step:rel:shift}. The algorithm focuses on
computing an \((\ash-\bshr)\)-weak Popov basis \(\basis\) of
\(\relmod{\modulus}{\apo,-1}\), since according to \cref{cor:shift_amplitude}
this \(\basis\) is then also in \(\ash\)-weak Popov form.

Recall the quantity \(\dred = 2\deg(\apo) - d - \ash - 1 + \min(0,\ash+d) =
\deg(\apo) - 1 - \ashr + \bshr\) from \cref{thm:rel}. The first item in that theorem
considers \(\dred \le 0\), meaning \(\ashr - \bshr \ge \deg(\apo) -
1\). The case \(\ashr-\bshr > d\) (or, equivalently, \(\ash - \bshr > \deg(\apo)\))
is handled in complexity \(\bigO{1}\) at
\cref{step:rel:large_amp_1}, via the Hermite basis in
\cref{prop:ratrecon_reduction_general}. The case \(\deg(r) \le \ashr - \bshr \le
\deg(\apo)\) is then handled at
\cref{step:rel:large_amp_2,step:rel:large_amp_3} based on
\cref{prop:handle_dega}, at a cost of \(\bigO{\timepm{d}}\) operations for the
division with remainder at \cref{step:rel:divrem}. Since \(\deg(r) \le
\deg(\apo) - 1\) this covers all cases with \(\ashr - \bshr \ge \deg(\apo) -
1\), proving the first item of \cref{thm:rel}.

Now consider the second item, with \(\dred > 0\). Apart from the call
to \textproc{Relation-modN}, other operations at
\crefrange{step:rel:reds_start}{step:rel:reds_end} cost \(\bigO{\timepm{d}}\)
overall: there are essentially a power series quotient at \cref{step:rel:to_red2}
and polynomial multiplications at
\cref{step:rel:to_red4,step:rel:from_red2,step:rel:from_red1,step:rel:from_red},
with \cref{step:rel:from_red1} possibly also involving a division by
\(\modulus_1\). (Concerning the latter, as noted in
\cref{sec:rel_xgcd:rel_to_app}, one may alternatively compute the right column
of \(Q\) via middle products as in \cref{eqn:right_col_via_midprod}.) In these
computations, the involved degrees are always at most \(\max(d_2,d) \in
\bigO{d}\), including at \cref{step:rel:to_red4} since by definition of
\(\apo_3\) we have
\(\apo_3 \nodulus \rem x^{\delta} = x^{\delta-d_2} (\apo_2 \nodulus \rem
x^{d_2})\).

Correctness is a direct consequence of
\cref{prop:handle_dega,prop:ratrecon_reduction_general,prop:ratrecon_reduction_niceshift},
up to the minor following remark. \Cref{prop:ratrecon_reduction_niceshift}
requires two moduli of the same degree: applying it as such would
require \(\delta = \dred\). To allow any \(\delta \ge \dred\) and thus much
more flexibility in the choice of \(\nodulus\), we first transform the
approximant modulus \(\modulus_2 = x^{d_2}\) of \cref{step:rel:to_red2} into
the modulus \(\modulus_3 = x^{\delta}\) as in \cref{step:rel:to_red3}. Having
now a modulus of degree \(\delta\), we can apply
\Cref{prop:ratrecon_reduction_niceshift} to switch to modulus \(\nodulus\).
This approximation order increase from \(d_2\) to \(\delta\) is based on the
following straightforward result.

\begin{lemma}
  \label{lem:approx_increase_order}%
  Let \(d \in \NN\), \(\apo \in \pR_{<d}\), and \(\ash \in
  \ZZ\). Take some integer \(\delta \ge d\). Any \((\ash+\delta-d)\)-weak Popov
  basis \(P\) of \(\appmod{\delta}{x^{\delta-d} \apo, -1}\) has the form \(P =
  Q \diag{1,x^{\delta-d}}\), where  \(Q \in \mpR{2}{2}\) is
  an \(\ash\)-weak Popov basis of \(\appmod{d}{\apo,-1}\).
\end{lemma}

\begin{proof}
  The rows of \(P\) are in \(\appmod{\delta-d}{x^{\delta-d} \apo, -1} =
  \appmod{\delta-d}{0, -1}\), of which
  \(\diag{1,x^{\delta-d}}\) is a basis. Thus \(P = Q \diag{1,x^{\delta-d}}\) for some \(Q
  \in \mpR{2}{2}\). It is then easily checked that since  \(P\) is in
  \((\ash+\delta-d)\)-weak Popov form, the matrix \(Q\) must be in
  \(\ash\)-weak Popov form.
  \qed
\end{proof}

Regarding this transformation, we note that after \cref{step:rel:from_red4}, we
could have inserted an intermediate step which would compute the
\(\ash_3\)-weak Popov basis
\[
  \basis_3 \gets
  \begin{bmatrix}
    g_{00} & g_{00}\apo_3 \rem x^{\delta} \\
    g_{10} & g_{10}\apo_3 \rem x^{\delta}
  \end{bmatrix},
\]
of \(\appmod{\delta}{\apo_3,-1}\), and then \cref{step:rel:from_red2} would
become \(\basis_2 \gets \basis_3 \diag{1, x^{d_2-\delta}}\), yielding the same
basis. We preferred to write  \cref{step:rel:from_red2} in a way that clearly
indicates that we are dealing with a multiplication truncated at degree \(d_2\)
instead of the larger \(\delta\).

Observe also that in the case \(\nodulus = x^{\delta}\), at
\cref{step:rel:basis2} one may simply take
\(\basis_2 \gets \basis_4 \diag{1, x^{d_2 - \delta}}\).
%%[\begin{smallmatrix}
%%  g_{00} & x^{d_2-\delta}g_{01} \\
%%  g_{10} & x^{d_2-\delta}g_{11}
%%\end{smallmatrix}]

\subsection{Solving \problemName{pbm:xgcd} via a relation basis in Hermite form}
\label{sec:rel_xgcd:xgcd_to_rel}

We make explicit our above remark that solving a specific instance of
\problemName{pbm:rel} directly leads to an algorithm for
\problemName{pbm:xgcd}. The next lemma is a particular case of
\cref{lem:relbas_hermite_bases}; we give a self-contained proof, for
readability.

\begin{lemma}
  \label{lem:xgcd_to_rel}%
  Let \(d\in \ZZp\) and \(\apo, \bpo \in \pR\) be polynomials with \(0
  \le \deg(\apo) < \deg(\bpo) = d\). Let \(u,v,\pg\) be the unique polynomials
  in \(\pR\) with \(\pg\) monic such that \(u\apo + v\bpo = \pg =
  \gcd(\apo,\bpo)\), \(\deg(u) < \deg(\bpo/\pg)\), and \(\deg(v) <
  \deg(\apo/\pg)\). Then, the basis of \(\relmod{\bpo}{a,-1}\) in Hermite form
  is
  \[
    H =
    \begin{bmatrix}
      \frac{\bpo}{\lc{\bpo}\pg} & 0 \\
      u & \pg
    \end{bmatrix}
    \in \mpR{2}{2}.
  \]
  Note that \(H\) is also in \((2\ell-d+1)\)-Popov form for any \(0 \le
  \ell \le \deg(\pg)\).
\end{lemma}
\begin{proof}
  Each row of \(H\) is in  \(\relmod{\bpo}{a,-1}\) by construction: both
  \(\frac{\bpo}{\lc{\bpo}\pg} \apo\) and \(u \apo - \pg\) are multiples of
  \(\bpo\). From \(\det(H) = \bpo / \lc{\bpo}\) and
  \cref{cor:determinant_relmod}, we deduce that \(H\) is a basis of
  \(\relmod{\bpo}{a,-1}\). It is easily verified that \(H\) is in
  \((2\ell-d+1)\)-Popov form: the diagonal entries of \(H\) are monic, and
  its second row satisfies  \(\deg(u) + 2\ell - d + 1 \le \deg(\bpo) -
  \deg(\pg) + 2\ell - d \le \deg(\pg)\).
  \qed
\end{proof}

To solve \problemName{pbm:xgcd}, one firsts compute some
\((2\ell-d+1)\)-weak Popov basis \(P\) of \(\relmod{\bpo}{a,-1}\). Then,
according to \cref{lem:weak_to_popov_2x2}, the above \((2\ell-d+1)\)-Popov
basis \(H\) can be deduced from \(P\) using \(\bigO{\timepm{d}}\) field
operations. This matrix \(H\) yields \(\pg\) and \(u\), and the cofactor \(v\)
can be found as the quotient \(v = (\pg - u\apo) / \bpo\) again in time
\(\bigO{\timepm{d}}\).
This proves \cref{cor:xgcd}, as a direct consequence of \cref{thm:rel}. Indeed,
here the parameters are \(d = \deg(\bpo)\) and \(\ash = 2\ell-d+1\).
Therefore \(\min(0, \ash+d) = 0\), and the main parameter that rules the
complexity bound is \(\dred = 2\deg(\apo) - d - 1 -
\ash = 2(\deg(\apo) - \ell - 1)\).

\subsection{Reduction from approximant bases to relation bases with any modulus}
\label{sec:rel_xgcd:app_to_modN}

The reduction from \(\modulus\) to \(\nodulus\) achieved by
\cref{algo:RelationBasis-via-modN} requires \(\bpo=-1\). While this covers
\cref{pbm:ratrecon} by definition and also \cref{pbm:xgcd} as showed in
\cref{sec:rel_xgcd:xgcd_to_rel}, we cannot apply this reduction to  bases of
approximant modules as in \cref{pbm:app}, which typically involves \(\bpo \neq
-1\). One main motivation to fill this gap is that, as announced in
\cref{sec:intro}, as soon as the base field contains suitable points, the
interpolant basis algorithms in \cref{sec:interp} are faster than the
approximant basis ones in \cref{sec:approx}. Thus, we could lower the
complexity bounds for approximant bases by reducing them to interpolant
bases. In fact, they can be reduced to relations modulo any \(\nodulus\), by first
transforming the instance with \((\apo,\bpo)\) into another approximant
instance with \((\apor,-1)\):

\begin{lemma}
  \label{lem:app_reduce_b_to_mone}%
  Let \(d \in \mathbb{N}\), \(\ash \in \mathbb{Z}\), and \((\apo,\bpo) \in
  \pR_{<d}^2\). Assume \(\bpo \neq 0\), let \(v = \valuation{\bpo} \in
  \{0,\ldots,d-1\}\), and define \(\apor = (x^{-v} \bpo)^{-1} \apo \rem
  x^{d-v} \in \pR_{<d-v}\). For any
  \((\ash+v)\)-weak Popov basis \(\basis \in \mpR{2}{2}\) of
  \(\appmod{d-v}{\apor, -1}\), the matrix \(\basis [\begin{smallmatrix} x^v & 0
  \\ 0 & 1\end{smallmatrix}]\) is an \(\ash\)-weak Popov basis of
  \(\appmod{d}{\apo,\bpo}\).
\end{lemma}

\begin{proof}
  From the definitions in \cref{sec:prelim:polmat}, one easily verifies that
  \(\basis\) is in \((\ash+v)\)-weak Popov form if and only if \(\basis
  [\begin{smallmatrix} x^v & 0 \\ 0 & 1\end{smallmatrix}]\) is in \(\ash\)-weak
  Popov form. Since \(x^{-v}\bpo\) is invertible modulo \(x^{d-v}\),  one has
  \(\appmod{d-v}{\apo, x^{-v}\bpo} = \appmod{d-v}{\apor, -1}\). From the
  definition of \(\appmod{d}{\apo,\bpo}\) and the fact that \(\bpo\) has
  valuation \(v\), one deduces \(\appmod{d}{\apo,\bpo} = \{(x^v p, q) \mid (p,
  q) \in \appmod{d-v}{\apo, x^{-v}\bpo}\}\), which concludes the proof.
  \qed
\end{proof}

Since the transformed instance involves \(-1\), we can then apply the reduction
of \cref{algo:RelationBasis-via-modN} to rely on a relation basis modulo any
\(\nodulus\). There is still one caveat: the latter reduction asks that
\(\deg(\nodulus) \ge \dred = 2\deg(\apor) - (d-v) -
\ash - 1 + \min(0,\ash+d-v)\). Thus, if \(\ash\) is negative, and more
precisely if \(v - d \le \ash \le 0\), this may require \(\nodulus\) to have
degree significantly greater than the original modulus \(\modulus = x^{d-v}\).
For example, this could almost double the degree when  \(\ash \approx
v-d\), since we get \(\dred \approx 2(d-v)\) in the typical case \(\deg(\apor) \approx
d - v - 1\). This growth can be avoided: as the next lemma shows, if \(\ash\)
is negative, one can basically swap \(\apo\) and \(\bpo\) in order to negate
it.

\begin{lemma}
  \label{lem:invinput}%
  Define \(J = [\begin{smallmatrix} 0 & 1 \\ 1 &0 \end{smallmatrix}] \in
  \mKK{2}{2}\), and let \(\ash \in \ZZ\). A matrix \(\basis \in \mpR{2}{2}\) is
  in \(\ash\)-weak Popov form if and only if \(J\basis J\) is in
  \((1-\ash)\)-weak Popov form. In particular, \(\basis\) is an
  \(\ash\)-weak Popov basis of \(\relmod{\modulus}{\apo,\bpo}\) if and only if
  \(J\basis J\) is a \((1-\ash)\)-weak Popov basis of
  \(\relmod{\modulus}{\bpo,\apo}\), for any \(\apo,\bpo \in \pR\).
\end{lemma}

\begin{proof}
  Write \(\basis = [\begin{smallmatrix} p_{00} & p_{01} \\ p_{10} & p_{11}
    \end{smallmatrix}]\), so that \(J\basis J = [\begin{smallmatrix} p_{11} &
  p_{10} \\ p_{01} & p_{00} \end{smallmatrix}]\). Subtracting \(\ash\) from all
  sides of the \(\ash\)-weak Popov conditions on \(\basis\), which are
  \(\deg(p_{01}) < \deg(p_{00}) + \ash\) and \(\deg(p_{10}) + \ash \le
  \deg(p_{11})\), yields the conditions \(\deg(p_{10}) < \deg(p_{11}) +
  1-\ash\) and \(\deg(p_{01}) + 1-\ash \le \deg(p_{00})\), which are exactly
  the \((1-\ash)\)-weak Popov conditions on \(J \basis J\). This proves the
  first equivalence. The second one follows by observing that \(\basis J\) is a
  basis of \(\relmod{\modulus}{\bpo,\apo}\) if and only if \(\basis\) is a
  basis of \(\relmod{\modulus}{\apo,\bpo}\), and \(J \basis J\) generates the same
  \(\pR\)-module as \(\basis J\) since \(J\) is unimodular.
  \qed
\end{proof}

Thus, the steps to compute a basis of \(\appmod{d}{\apo,\bpo}\) via a relation
basis modulo some \(\nodulus\) would be:
\begin{itemize}[noitemsep,topsep=2pt]
  \item compute parameters \(\delta_\bpo = 2(d-v_\bpo -1) - (d-v_\bpo) -
    \ash - 1 + \min(0,\ash+d-v_\bpo)\) and
  \(\delta_\apo = 2(d-v_\apo -1) - (d-v_\apo) +
    \ash - 2 + \min(0,1-\ash+d-v_\apo)\);
  \item if \(\delta_\bpo \le \delta_\apo\), reduce to an instance of the form
    \(\appmod{d-v_\bpo}{\apor, -1}\) (\cref{lem:app_reduce_b_to_mone}), then
    solve it via a relation basis modulo any \(\nodulus\) of degree \(\ge
    \delta_\bpo\) (\cref{algo:RelationBasis-via-modN}); 
  \item if \(\delta_\bpo > \delta_\apo\), swap \(\apo\) and \(\bpo\)
    and reduce to an instance of the form \(\appmod{d-v_\apo}{\apor, -1}\) and
    shift \(1-\ash\) (\cref{lem:invinput,lem:app_reduce_b_to_mone}),
    then solve it
    via a relation basis modulo any \(\nodulus\) of degree \(\ge \delta_\apo\)
    (\cref{algo:RelationBasis-via-modN}),
    and swap columns and rows of the obtained basis following
    \cref{lem:invinput}.
\end{itemize}
Thanks to the possible swap, we only require \(\nodulus\) to have degree at 
least \(\min(\delta_\apo,\delta_\bpo)\). By definition, this quantity
is
\(\min(\delta_\apo,\delta_\bpo)
% =
% \min(
% d-v_\apo + \ash - 4 + \min(0,1-\ash+d-v_\apo),
% d-v_\bpo - \ash - 3 + \min(0,\ash+d-v_\bpo)
% )
\le
\min(
d-v_\apo + \ash - 4,
d-v_\bpo - \ash - 3
)
\le
% \min(d-\ash, d+\ash)
% =
d - \abs{\ash}
\). This explains the two
rightmost reductions in the graph of \cref{fig:constants_via_reductions}.

\begin{remark}
  \label{rmk:int_to_modN}%
  Ideas similar to those in this subsection can be used for computing
  interpolant bases of \(\intmod{\points}{\apo,\bpo}\) via a relation basis
  with another modulus \(\nodulus\). This would typically be used to reduce
  from the case of general points \(\points\) to one of the situations with
  faster algorithms (approximants, or interpolants at special points if
  available). We will not detail this reduction here, and will simply point out
  that it may increase the leading constant in the dominant term
  \(\bigO{\timepm{d}\log(d)}\) in the complexity bound, due to calls to fast
  multipoint evaluation and interpolation at points of \(\points\).
  \qed
\end{remark}
%% NOTE some more technical details
% The key observation is that, when \(\bpo\) is invertible modulo the modulus
% \(\modulus = \prod_{i=0}^{d-1}(x - \point_i)\), one has \(\relmod{\modulus}{\apo,
%   \bpo} = \relmod{\modulus}{-\apo\bpo^{-1}, -1}\), so the problem reduces to the
% case \(\bpo = -1\) by inverting the evaluations of \(\bpo\).

\subsection{An alternative approach for \problemName{pbm:xgcd}}
\label{sec:rel_xgcd:gcd_syz}

We mention another approach, based on \cref{lem:kernel_approx}, for solving
\problemName{pbm:xgcd}; let \(\apo,\bpo,\ell\) be the input. This lemma shows
that, after choosing some shift \(\ash\) and a modulus \(\modulus\) of
large enough degree, one can obtain \(\pg = \gcd(\apo,\bpo)\) from the
computation of an \(\ash\)-weak Popov basis \(\basis\) of
\(\relmod{\modulus}{\apo,\bpo}\). Observe that here we have a general instance
of \(\relmod{\modulus}{\apo,\bpo}\) with \(\bpo\) not necessarily equal to
\(-1\), so the above reductions do not apply, but on the other hand we can
freely choose the modulus \(\modulus\).

Let us now comment on the choice of shift and modulus. The main parameter in
the cost of computing a basis \(P\) of \(\relmod{\modulus}{\apo,\bpo}\) is the
degree \(d = \deg(\modulus)\). Thus, a good candidate shift is \(\ash =
\deg(\apo)-\deg(\bpo)\), as it minimizes the quantity in
\cref{eqn:order_yielding_syz}, which is \(\dsyz = \deg(\apo) + \deg(\bpo) + 1 -
2\deg(\pg)\) in that case. This shift furthermore guarantees that it is the
second row of \(\basis\) which has the form \([-\lambda\bpo/\pg \;\;
\lambda\apo/\pg]\), since the \(\ash\)-pivot of that vector is the
right-hand entry. One would then take any modulus \(\modulus\) of degree \(d
\ge \deg(\apo) + \deg(\bpo) + 1 - 2\ell\), corresponding e.g.\ to approximants
or interpolants for efficiency reasons. If there is a discrepancy between
\(\deg(\apo)\) and \(\deg(\bpo)\), one could perform an preliminary division
with remainder \(\bpo = q \apo + r\), bringing the constraint down to \(d \ge
\deg(r) + \deg(\apo) + 1 - 2\ell\). In the typical case \(\deg(r) = \deg(\apo)
- 1\), this means \(d \ge 2(\deg(\apo) - \ell)\). Observe that this main
parameter is essentially the same as the one \(\dred = 2(\deg(\apo) - \ell -
1)\) of the approach in \cref{sec:rel_xgcd:xgcd_to_rel}.

Starting from this method which yields the GCD, one can also compute the
cofactors \((u,v)\) at an additional cost of \(\bigO{\timepm{d}}\) field
operations. We only sketch this below,
as this leads to a complexity bound identical to that of the approach detailed
in the previous subsections.
%% Indeed, the main complexity
%% parameter is the same: we mentioned \(m+n+1-2\ell\) above, and in case of
%% discrepancy between \(m\) and \(n\) one may perform an initial division with
%% remainder to bring this down to \(2(n-\ell)\).
We still find this alternative method interesting because it highlights that
one can solve \problemName{pbm:xgcd} via an approximant basis of \((\apo,
\bpo)\) at order about \(m+n+1\), which is typically solved by first
considering the low degree coefficients of \(\apo\) and \(\bpo\). In contrast,
the reduction of \cref{sec:rel_xgcd:rel_to_app} used in
\cref{sec:rel_xgcd:main_proof,sec:rel_xgcd:xgcd_to_rel} incidentally reverses
the coefficients of \(\apo\) and \(\bpo\), and one is left with an approximant
basis of the same order but operating first on the low degree coefficients of
the reversals, that is, the high degree coefficients of \(\apo\) and \(\bpo\),
similarly to the half-gcd algorithm.

Choose the shift \(\ash = \deg(\apo)-\deg(\bpo)\)
and any modulus \(\modulus\) of degree \(d\), with \(d\) sufficiently large as
indicated above. Considering \(\gamma = \gcd(\modulus, \apo, \bpo)\), define
\[
  Q =
  \begin{bmatrix}
    u \modulus/\gamma & v \modulus/\gamma \\
    -\bpo/\pg & \apo/\pg
  \end{bmatrix}  \in \mpR{2}{2}
  \quad\text{and}\quad
  \basis =
  \begin{bmatrix}
    \pp & \pq \\
    -\bpo/\pg & \apo/\pg
  \end{bmatrix}  \in \mpR{2}{2},
\]
where \(\basis\) is an \(\ash\)-weak Popov basis of
\(\relmod{\modulus}{\apo,\bpo}\), hence \(\det(\basis) = \lc{p} \lc{\apo}
\modulus/\gamma\). The fact that \(\basis\) has a row equal to \([-\bpo/\pg
\;\; \apo/\pg]\) is ensured by \cref{lem:kernel_approx} (up to dividing by
a constant \(\lambda\)), and here the chosen shift ensures that it is the
second row. Since \(Q\) has determinant \(\modulus/\gamma\) by
construction, it is also a basis of \(\relmod{\modulus}{\apo,\bpo}\)
(\cref{cor:determinant_relmod}). Thus \(\basis\) and \(Q\) are
left-unimodularly equivalent, and the facts that \(Q\) and \(\basis\) have
the same second row and satisfy \(\det(\basis) =
\lc{p} \lc{\apo} \det(Q)\) ensure that  \([u \modulus/\gamma \;\; v
\modulus/\gamma] = [r_{0} \;\; r_{1}]\basis\) for the nonzero constant \(r_0
= (\lc{p} \lc{\apo})^{-1}\) and for some \(r_{1} \in \pR\) whose
degree is less than \(\deg(\modulus/\gamma)\).

It remains to compute \(r_1\), which will give access to the cofactors
\((u,v)\). From the above vector identity, one deduces explicit formulas:
\(r_1 = r_0 p (\bpo / \pg)^{-1} \bmod \modulus/\beta\) and \(r_1 = - r_0 q
(\apo / \pg)^{-1} \bmod \modulus/\alpha\), where \(\alpha = \gcd(\modulus, \apo)\)
and \(\beta = \gcd(\modulus, \bpo)\). Thus, by a CRT computation,  one can
deduce the actual polynomial \(r_1\), since \(\lcm(\modulus/\alpha,
\modulus/\beta) = \modulus/\gamma\).

In terms of computations, this approach seems to involve several instances of
GCD and modular inversion, for obtaining \(\alpha,\beta,\gamma\) and the
inverses modulo \(\modulus/\alpha\) or modulo \(\modulus/\beta\). However, note
that all these are actually done in complexity \(\bigO{\timepm{d}}\) if one
works for example with \(\modulus = x^{d}\) or with interpolants at \(d\)
points in geometric progression.

\section{Minimal approximant bases: Pad\'e approximation}
\label{sec:approx}

In this section, we first present the classical divide and conquer approximant
basis algorithm based on Beckermann and Labahn's recursive approach (see
\cref{lem:dnc_relbas}), and then we augment it with optimizations outlined at
the beginning of \cref{sec:approx:appbasis}. We provide thorough complexity
analyses of both variants, leading to \cref{thm:appbasis:cx} stated and
commented on in \cref{sec:intro}.

For the base case of recursions, we give ourselves an algorithm with the
following specification and running time \(\bigO{1}\). In practice, one would
typically use an iterative algorithm such as those from
\cite{m_pade,BeckermannLabahn1994}, and would use this one for all base cases
with \(d\) not exceeding some predetermined (constant) threshold.

\begin{algorithm}[h]
  \algoCaptionLabel{AppBasis2-base}{d, \apo, \bpo, \ash}
  \begin{algorithmic}[1]
    \Require \(d \in \{0,1,2\}\),
    polynomials \((\apo,\bpo) \in \pR_{<d}^2\),
    shift \(\ash \in \mathbb{Z}\)
    \Ensure the basis \(P \in \mpR{2}{2}\) of \(\appmod{d}{\apo,\bpo}\) in
    \(\ash\)-Popov form
  \end{algorithmic}
\end{algorithm}

\subsection{Polynomial multiplication and middle product}
\label{sec:approx:polmul_midprod}

In approximant basis computations, apart from sums and products of polynomials,
we will also use the \emph{middle product} operation:

\begin{definition}
  \label{dfn:middle_product}%
  For \(m,n \in \NN\), \(p \in \pR_{\le m}\) and \(q \in
  \pR_{< m+n}\), the middle product is
  \[
    \midprod{p, q, m, m+n}
    =
    \sum_{0 \le k < n} f_{m+k} x^k \in \pR_{<n},
    \quad \text{where } pq = \sum_{k \ge 0} f_k x^k.
  \]
  We denote by \(\timemp{m, n}\) the corresponding time function.
\end{definition}

While a middle product could be simply be found by computing a polynomial
product and selecting suitable coefficients from it, it turns out that
dedicated algorithms are faster, improving at least the leading constant factor
in complexity bounds \cite{HanrotQuerciaZimmermann2004}. We will encounter the
following situation with an additional degree constraint on \(p\).

\begin{lemma}
  \label{lem:midprod_with_gap}%
  For \(d,m,n \in \NN\) with \(d \le m\), let \(p \in \pR_{\le d}\) and \(q \in
  \pR_{< m+n}\). Then the middle product \(\midprod{p, q, m, m+n}\) can be
  computed in \(\timemp{d, n}\) operations in \(\field\).
\end{lemma}
\begin{proof}
  Write the polynomial division \(q = \tilde{q} x^{m-d} + r\), with
  \(\deg(r) < m-d\). Then \(pq = p\tilde{q} x^{m-d} + pr\) with \(\deg(pr) <
  m\), hence the coefficients \(m,m+1,\ldots,m+n-1\) of \(pq\) are exactly the
  coefficients \(d,d+1,\ldots,d+n-1\) of \(p\tilde{q}\). Thus \(\midprod{p, q, m,
  m+n} = \midprod{p, \tilde{q}, d, d+n}\). Besides, computing \(\tilde{q}\)
  from \(q\) does not use arithmetic operations.
  \qed
\end{proof}

We require some mild assumptions on the time functions for multiplications and
middle products.

\begin{assumption}
  \label{hyp:timepm}%
  We consider the following assumptions on the time functions \(\timepm{\cdot}\) and
  \(\timemp{\cdot}\):
  \begin{enumerate}[(i)]
    \item \label{hyp:timepm:nondecr}

    The function \(n \mapsto \frac{\timepm{n}}{n\log_2(n)}\) is nondecreasing.

    \item \label{hyp:timepm:sumdeg}

      Two polynomials \(p \in \pR_{\le m}\) and \(q \in \pR_{\le n}\) can be
      multiplied in \(\timepm{\lfloor \frac{m+n}{2} \rfloor} + \bigO{m+n}\)
      operations in \(\field\).

    \item \label{hyp:timepm:midprod}

      The middle product algorithm satisfies \(\timemp{m, n} \in \timepm{\lfloor
      \frac{m+n}{2} \rfloor} + \bigO{m+n}\).
\end{enumerate}
\end{assumption}

\noindent
We will assume that \(\timepm{n}\) is in \(\bigO{n^2}\), that is, the
multiplication algorithm is not slower than the classical one. We will use the
fact that \(\timepm{n + k}\) is in \(\timepm{n} + \bigO{n}\) for any fixed
constant \(k \in \NN\), and that \(\timepm{n + k}\) is in \(\timepm{n} +
\bigO{n\log(n)}\) in the case \(k \in \bigO{\log(n)}\).

\Cref{hyp:timepm:nondecr} implies in particular the \emph{superlinearity}
property \(\timepm{m} + \timepm{n} \le \timepm{m + n}\) for any \(m,n \in \NN\)
\cite[Sec.\,8.3 Eq.\,(9)]{GathenGerhard2013}. We also note that the weaker
assumption ``\(n \mapsto \timepm{n}/n\) is nondecreasing'' (as in \cite[Sec.\,8.3
Eq.\,(8)]{GathenGerhard2013}) may be more frequent in the literature. Yet, on
the one hand the stronger \cref{hyp:timepm:nondecr} is satisfied when
instantiating \(\timepm{\cdot}\) corresponding to any of the usual
multiplication algorithms (such as those listed in
\cite[Tab.\,8.6]{GathenGerhard2013}), and on the other hand it allows us to
refine complexity bounds for certain recursive computations. Specifically, the
next statement exhibits a leading constant \(\frac{1}{2}\) (see
\cref{app:dnc_factor_half} for more details):

\begin{quotation}
  assuming \cref{hyp:timepm:nondecr}, \(\sum_{i=0}^{k-1} 2^i \timepm{n/2^i}
  \leq \frac{1}{2}\timepm{n}(k+1)\) holds

  for any \(n = 2^k \in \ZZp\).
\end{quotation}
However, the assumption ``\(n \mapsto \timepm{n}/n\) is nondecreasing''
would yield the larger constant \(1\) in this upper bound; one can see this by
choosing \(\timepm{n} = n\), which is not forbidden by this weaker assumption.
In \cite{vdHoeven2025}, the same \hyptime{hyp:timepm:nondecr} was made, and the
above stated consequence was exploited in complexity analyses, e.g.\ to get the
factor \(\frac{1}{2}\) in \cite[Prop.\,6]{vdHoeven2025}. Here we will use the
following bound, valid for an arbitrary \(n\), which we prove in
\cref{app:dnc_factor_half}.

\begin{lemma}
  \label{lem:nondecr_cst_factor_general}%
  Let \(n \in \ZZp\) and \(k = \lfloor \log_2(n) \rfloor\). If
  \hyptime{hyp:timepm:nondecr} holds, then
  \[
    \sum_{i=0}^{k-1} 2^i \timepm{\left\lceil \frac{n}{2^i}\right\rceil}
    \;\;\le\;\; \frac{1}{2}\timepm{n}(k+5).
  \]
\end{lemma}

The assumption \cref{hyp:timepm:sumdeg} will help us analyze the complexity in
case of unbalanced degrees. One could always ensure it holds by increasing
\(\timepm{n}\) by a constant factor at most \(2\), since by definition \(pq\)
can be multiplied using at most \(\timepm{m+n}\) operations, which is within
\(2\timepm{\lfloor \frac{n+m}{2} \rfloor} + \bigO{n+m}\). Yet this increase should
be avoided if possible and is indeed not necessary in the case of the
multiplication algorithms listed in \cite[Tab.\,8.6]{GathenGerhard2013}. For
example, this can be seen for the classical algorithm with \(\timepm{n} =
2(n+1)^2\) by recalling that its cost in case of possibly distinct degrees is
\(2(m+1)(n+1)\) operations in \(\field\) \cite[Sec.\,2.3]{GathenGerhard2013};
there \cref{hyp:timepm:sumdeg} boils down to \(mn \le (\frac{m+n}{2})^2\).
FFT-based algorithms rely on a convolution operation such as \((pq) \rem (x^d -
1)\) or \((pq) \rem (x^d + 1)\), where the main parameter \(d\) in the
complexity is chosen so as to slightly exceed the degree \(n+m\) of the sought
product \(pq\) (see e.g.\ \cite[Thm.\,8.18 and Alg.\,8.20]{GathenGerhard2013}).

The assumption on \(\timemp{\cdot}\) in \cref{hyp:timepm:midprod} is
in the same spirit as the one on \(\timepm{\cdot}\) in
\cref{hyp:timepm:sumdeg}, and even brings no additional restriction if using
the usual polynomial multiplication algorithms. Indeed,
\cite[Appendix]{HanrotQuerciaZimmermann2004} shows that, as soon as
\(\timepm{\cdot}\) is realized through a bilinear multiplication algorithm,
\cref{hyp:timepm:midprod} is implied by \cref{hyp:timepm:sumdeg}.

\begin{remark}
  \label{rmk:generic_weaker_assumptions}%
  In most complexity analyses below, weaker assumptions are used when analyzing
  the generic case, because we have better control of the degrees of the
  manipulated polynomials and \(2 \times 2\) matrices in that case.
  For example, in the analyses presented in
  \cref{sec:approx:pmbasis,sec:approx:appbasis}, 
  \cref{hyp:timepm:sumdeg} is not needed, and
  \cref{hyp:timepm:midprod} is only used for the specific case where \(m \sim
  n/2\) (or precisely, \(m = \lceil d/4 \rceil\) and \(n = d - 2 \lceil d/4
  \rceil\) where \(d\) is the approximation order).
  \qed
\end{remark}

Finally, we note that the constants introduced in \cref{sec:intro:complexity},
\(\cstmatmul\) for \(2\times 2\) polynomial matrix multiplication and
\(\cstpolmulrep{k}\) for repeated multiplications by the same polynomial,
benefit from the assumptions above. Indeed, as explained in
\cite[Sec.\,2]{vdHoeven2025}, two matrices in \(\mpR{2}{2}_{\le m}\) and
\(\mpR{2}{2}_{\le n}\) can be multiplied using \(\cstmatmul \timepm{\lfloor
\frac{m+n}{2} \rfloor} + \bigO{m+n}\) operations in \(\field\),
and the same bound holds for a \(2 \times 2\) matrix middle product with
parameters \(m\) and \(m+n\). Similarly, for a fixed \(k \ge 1\), and given
polynomials \(p_1,\ldots,p_k \in \pR_{\le m}\), one can perform the \(k\)
multiplications \(p_1 q, \ldots, p_k q\) with \(q \in \pR_{\le n}\) in
\(\cstpolmulrep{k}\timepm{\lfloor \frac{m+n}{2} \rfloor} + \bigO{m+n}\)
operations in \(\field\), and the same bound holds for doing the \(k\) middle
products \(\midprod{p_i, q, m, m+n}\) with \(q \in \pR_{< m+n}\).
%% the latter (repeated MidProd) is a transpose of the former (repeated Mul)

\subsection{Complexity bounds for the classical algorithm}
\label{sec:approx:pmbasis}

\subsubsection{The classical divide and conquer algorithm}
\label{sec:approx:pmbasis:algorithm}

We now recall the algorithm due to Beckermann and Labahn for efficiently
computing approximant bases \cite[Sec.\,6]{BeckermannLabahn1994}, on which we
will base our complexity analyses. We specialize it to inputs formed by only
two polynomials, and modify it straightforwardly so as to avoid the explicit
use of FFT for polynomial multiplication (since it may not be feasible over
\(\field\)). Our description is thus close to
\cite[Alg.\,PM-Basis]{GiorgiJeannerodVillard2003}; in fact, whereas the latter
algorithm is generally faster than the former, both become essentially
equivalent when the input consists of only two polynomials. The recursion
splits the target approximation order as \(d = \lceil d/2 \rceil + \lfloor d/2
\rfloor\), like in \cite{GiorgiJeannerodVillard2003}, but unlike in
\cite{BeckermannLabahn1994} which uses \(d = d_r + (d-d_r)\) with \(d_r\) the
largest power of \(2\) smaller than or equal to \(d\) (this is indeed convenient
when working with FFT, and we will come back to it in \cref{sec:interp}).
How one splits \(d\) does not make a difference for correctness, which is is
based on the correctness of the base case and on the general recursion
property in \cref{lem:dnc_relbas}.

\begin{algorithm}[ht]
  \algoCaptionLabel{PM-Basis2}{d, \apo, \bpo, \ash}
  \begin{algorithmic}[1]

    \Require \(d \in \mathbb{N}\), polynomials \((\apo,\bpo) \in \pR_{<d}^2\), shift \(\ash \in \mathbb{Z}\)

    \Ensure a basis \(P \in \mpR{2}{2}_{\le d}\) of \(\appmod{d}{\apo,\bpo}\) in \(\ash\)-weak Popov form

    \State\InlineIf{\(d \le 2\)}{\Return \(\Call{algo:AppBasis2-base}{d, \apo,\bpo, \ash}\)}

    \State \(\dreca \gets \lfloor d / 2 \rfloor\); \(\drecb \gets d - \dreca\)

    \State \(\basis =
        [\begin{smallmatrix}
          p_{00} & p_{01} \\
          p_{10} & p_{11}
        \end{smallmatrix}]
        \gets \Call{algo:PM-Basis2}{\dreca, \apo, \bpo, \ash}\)

    \State \(\apor \gets \midprod{p_{00}, \apo, \dreca, d} + \midprod{p_{01}, \bpo, \dreca, d}\);
        \MyComment{\(\apor = (x^{-\dreca} (p_{00} \apo + p_{01} \bpo)) \rem x^{\drecb}\)}
    \Statex \(\bpor \gets \midprod{p_{10}, \apo, \dreca, d} + \midprod{p_{11}, \bpo, \dreca, d}\)
        \MyComment{\(\bpor = (x^{-\dreca} (p_{10} \apo + p_{11} \bpo)) \rem x^{\drecb}\)}%
        \label{step:pmbasis:residual}

    \State \(\basisr \gets \Call{algo:PM-Basis2}{\drecb, \apor,\bpor, \deg(p_{00})+\ash - \deg(p_{11})}\)

    \State \Return \(\basisr \basis\)
        \label{step:pmbasis:multiplication}
  \end{algorithmic}
\end{algorithm}

\subsubsection{Worst-case complexity bound}
\label{sec:approx:pmbasis:cx_worst_case}

Now we give a complexity analysis using general degree bounds on the
polynomials computed during the algorithm. Assume that we are not in the base
case, i.e., \(d > 2\). We write
\[
  \basis =
  \begin{bmatrix}
    p_{00} & p_{01} \\
    p_{10} & p_{11}
  \end{bmatrix}
  \text{ with degrees }
  \begin{bmatrix}
    \mathring{p}_{00} & \mathring{p}_{01} \\
    \mathring{p}_{10} & \mathring{p}_{11}
  \end{bmatrix},
  \text{ and }
  \basisr =
  \begin{bmatrix}
    q_{00} & q_{01} \\
    q_{10} & q_{11}
  \end{bmatrix}
  \text{ with degrees }
  \begin{bmatrix}
    \mathring{q}_{00} & \mathring{q}_{01} \\
    \mathring{q}_{10} & \mathring{q}_{11}
  \end{bmatrix}
  .
\]
Note that concerning possible zero entries \(p_{ij} = 0\) (which may only
happen for \(i\neq j\)), for use in upper bounds it will prove convenient to
set \(\mathring{p}_{ij} = 0\), and similarly for entries of \(\basisr\). From
\cref{sec:prelim:polmat,sec:prelim:relbas}, we have the following properties,
used throughout this section:
\begin{itemize}[noitemsep,topsep=2pt]
  \item since \(\basis\) is a basis of \(\appmod{\dreca}{\apo,\bpo}\),
    \(\det(\basis)\) divides \(x^{\dreca}\) (\cref{cor:determinant_relmod});
  \item since \(\basis\) is \(\ash\)-weak Popov, \(\mathring{p}_{01} +
    \mathring{p}_{10} < \mathring{p}_{00} + \mathring{p}_{11} =
    \deg(\det(\basis)) \le \dreca\) (\cref{lem:degree_bounds_2x2});
  \item similarly, \(\det(\basisr)\) divides \(x^{\drecb}\)
    and \(\mathring{q}_{01} + \mathring{q}_{10} < \mathring{q}_{00}
    + \mathring{q}_{11} \le \drecb\).
\end{itemize}
There are two main computations to analyze: that of the so-called
\emph{residual} polynomials at \cref{step:pmbasis:residual}, which can be seen
as a matrix-vector middle product, and the basis multiplication at
\cref{step:pmbasis:multiplication}, which is a product of two matrices in
\(\mpR{2}{2}\).

\begin{lemma}
  \label{lem:appbas_cx:residual}%
  The residual polynomials \(\apor\) and \(\bpor\) at
  \cref{step:pmbasis:residual} can be computed in \(2\timepm{\lceil 3d/4
  \rceil} +\bigO{d}\) operations in \(\field\), which is
  in \(\frac{3}{2}\timepm{d} +\bigO{d}\).
\end{lemma}
\begin{proof}
  Since \(\deg(p_{00}) = \mathring{p}_{00} \le \drecb\) and \(d =
  \dreca+\drecb\), we may apply \cref{lem:midprod_with_gap}: it states that
  computing \(\midprod{p_{00}, \apo, \dreca, d}\) costs
  \(\timemp{\mathring{p}_{00}, \drecb}\) operations in \(\field\). Similar
  bounds hold for the three other middle products, so that overall, using
  \hyptime{hyp:timepm:midprod}, the complexity of \cref{step:pmbasis:residual}
  is within
  \[
    \sum_{0 \le i,j \le 1} \timemp{\mathring{p}_{ij}, \drecb}
    \in \sum_{0 \le i,j \le 1} \timepm{\left\lfloor \frac{\mathring{p}_{ij} + \drecb}{2} \right\rfloor} + \bigO{d}.
  \]
  Now, gather terms in the sum so as to exploit the inequalities
  \(\mathring{p}_{00}+\mathring{p}_{11} \le \dreca\) and
  \(\mathring{p}_{01}+\mathring{p}_{10} \le \dreca\). The first one
  implies \(\lfloor (\mathring{p}_{00} + \drecb)/2 \rfloor + \lfloor (\mathring{p}_{11} + \drecb)/2
  \rfloor \le \lfloor \dreca/2 \rfloor + \drecb \le \lceil 3d / 4 \rceil\).
  %% in case this is useful: we can prove
  %% \lfloor \dreca/2 \rfloor + \drecb == floor(d/4) + ceil(d/2) <= 3d/4 + 1/4
  Thus, by superlinearity of \(\timepm{\cdot}\),
  \[
    \timepm{%
      \left\lfloor \frac{\mathring{p}_{00} + \drecb}{2} \right\rfloor
    }
    +
    \timepm{%
      \left\lfloor \frac{\mathring{p}_{11} + \drecb}{2} \right\rfloor
    }
    \le
    \timepm{\left\lceil \frac{3d}{4} \right\rceil}.
  \]
  Deriving a similar bound from \(\mathring{p}_{01}+\mathring{p}_{10} \le d/2\)
  leads to the first claimed cost bound. The second one then follows from
  \hyptime{hyp:timepm:nondecr}, which implies
  \(
    \timepm{\lceil 3d/4 \rceil} / \lceil 3d/4 \rceil
    \le \timepm{d} / d
  \),
  and therefore
  \(
    \timepm{\lceil 3d/4 \rceil}
    \le (3d/4 + 1) \timepm{d} / d
    = \frac{3}{4} \timepm{d} + \timepm{d} / d
  \),
  where the term \(\timepm{d} / d\) is in \(\bigO{d}\).
  \qed
\end{proof}

\begin{lemma}
  \label{lem:appbas_cx:multiplication}%
  The multiplication \(\basisr\basis\) at \cref{step:pmbasis:multiplication}
  can be done in \(4\timepm{\lfloor d/2 \rfloor} + \bigO{d}\) operations in
  \(\field\).
\end{lemma}
\begin{proof}
  Performing classical matrix multiplication, the dominant part of the
  complexity comes from the 8 polynomial multiplications \(p_{ij}
  q_{jk}\), for \(0 \le i,j,k \le 1\). By \hyptime{hyp:timepm:sumdeg} this can
  be done in \(
  \sum_{0\le i,j,k \le 1} \timepm{\lfloor (\mathring{p}_{ij} +
  \mathring{q}_{jk})/2 \rfloor} + \bigO{d} \).
  These 8 terms can be grouped in 4 pairs each involving 4 polynomials
  whose degrees sum to at most \(d\). For example,
  \(p_{00}q_{00}\) is grouped with \(p_{11}q_{11}\), and we know that
  \(
  (\mathring{p}_{00} + \mathring{q}_{00}) + (\mathring{p}_{11} + \mathring{q}_{11})
  =
  (\mathring{p}_{00} + \mathring{p}_{11}) + (\mathring{q}_{00} + \mathring{q}_{11})
  \le d
  \).
  Thus, using the superlinearity of \(\timepm{\cdot}\),
  \[
    \timepm{\left\lfloor \frac{\mathring{p}_{00} + \mathring{q}_{00}}{2} \right\rfloor}
    +
    \timepm{\left\lfloor \frac{\mathring{p}_{11} + \mathring{q}_{11}}{2} \right\rfloor}
    \le
    \timepm{\left\lfloor \frac{d}{2} \right\rfloor}.
  \]
  Similar bounds are derived for the other terms by grouping
  \(p_{01}q_{10}\) with \(p_{10}q_{01}\),
  \(p_{00}q_{01}\) with \(p_{11}q_{10}\),
  and
  \(p_{01}q_{11}\) with \(p_{10}q_{00}\).
  The conclusion then follows by summing the 4 obtained bounds.
  \qed
\end{proof}

\begin{proposition}
  \label{prop:pmbasis:cx_worst_case}%
  For any input polynomials \((\apo,\bpo)\) and shift \(\ash\), the call
  \(\Call{algo:PM-Basis2}{d, \apo, \bpo, \ash}\) costs
  \[
    {\textstyle\frac{7}{2}}\timepm{\lceil d/2 \rceil} \log_2(\lceil d/2 \rceil)
    + {\textstyle\frac{41}{4}}\timepm{d}
    + \bigO{d \log(d)}
  \]
  operations in \(\field\).
\end{proposition}
\begin{proof}
  Recall that if \(d \le 2\), the base case of the recursion costs
  \(\bigO{1}\). If the input order is \(d > 2\), this base case is called
  at most \(d\) times overall, and thus its total contribution to the complexity
  bound is only \(\bigO{d}\). Let \(\Cxity{d}\) be the number of field operations
  used by \(\Call{algo:PM-Basis2}{d, \apo, \bpo, \ash}\). With
  two recursive calls at orders at most \(\lceil d/2 \rceil\),
  and other operations analyzed in the two previous lemmas,
  we get the divide and conquer complexity relation
  %\[
  %  \Cxity{d} = \Cxity{\lfloor d/2 \rfloor} + \Cxity{\lceil d/2 \rceil} + 2\timepm{\lceil 3d/4 \rceil}+ 4\timepm{\lfloor d/2 \rfloor} + \bigO{d}.
  %\]
  \[
    \Cxity{d} = 2\Cxity{\lceil d/2 \rceil} + {\textstyle\frac{3}{2}} \timepm{d} + 4\timepm{\lceil d/2 \rceil} + \bigO{d}.
  \]
  We use only ceilings to simplify the analysis, notably via the property
  \(\lceil \lceil d/2^i \rceil / 2 \rceil = \lceil d/2^{i+1} \rceil\). Let \(k
  = \lceil \log_2(d) \rceil - 2\) be the largest integer such that \(\lceil d / 2^k
  \rceil > 2\). By unrolling the above relation until reaching calls at order
  \(\lceil d/2^{k+1} \rceil \le 2\) (for which we use \(\Cxity{2} = \bigO{1}\)), we
  obtain
  \begin{equation}
    \label{eqn:recurrence_approx}
    \Cxity{d} = \sum_{i=0}^{k} 2^i \left({\textstyle\frac{3}{2}} \timepm{\left\lceil \frac{d}{2^i} \right\rceil} + 4\timepm{\left\lceil \frac{d}{2^{i+1}} \right\rceil} \right)  + \bigO{d \log(d)}.
  \end{equation}
  %\begin{align*}
  %  \Cxity{d} & = \sum_{i=0}^{k} 2^i \left(2\timepm{\left\lceil \frac{3}{4} \left\lceil \frac{d}{2^i} \right\rceil \right\rceil} + 4\timepm{\left\lceil \frac{d}{2^{i+1}} \right\rceil} \right)  + \bigO{d \log(d)} \\
  %\end{align*}
  First, since \(k \le \lfloor \log_2(d) \rfloor - 1 = \lfloor \log_2(\lceil
  d/2 \rceil) \rfloor\), from \cref{lem:nondecr_cst_factor_general} we get
  \begin{equation}
    \label{eqn:use_lem_cst_factor}%
    \sum_{i=0}^{k} 2^i \timepm{\left\lceil \frac{d}{2^{i+1}} \right\rceil}
    =
    \sum_{i=0}^{k} 2^i \timepm{\left\lceil \frac{\lceil d/2 \rceil}{2^i} \right\rceil}
    \le
    {\textstyle\frac{1}{2}} \timepm{\lceil d/2 \rceil} (\log_2(\lceil d/2 \rceil) + 5),
    %% note: we could have a floor around the log_2, but this has no impact
  \end{equation}
  where we have used \(\lceil d/2^{i+1} \rceil = \lceil \lceil d/2 \rceil/2^i \rceil\).
  This directly gives a bound for the second part of the sum in \cref{eqn:recurrence_approx}.
  For its first part, we may use the same inequality after isolating the term \(i=0\):
  \begin{align*}
    \sum_{i=0}^{k} 2^i {\textstyle\frac{3}{2}} \timepm{\left\lceil \frac{d}{2^i} \right\rceil}
    & = {\textstyle\frac{3}{2}} \timepm{d} + 3 \sum_{i=0}^{k-1} 2^i \timepm{\left\lceil \frac{\lceil d/2 \rceil}{2^i} \right\rceil} \\
    & \le {\textstyle\frac{3}{2}} \timepm{d} + {\textstyle\frac{3}{2}} \timepm{\lceil d/2 \rceil} (\log_2(\lceil d/2 \rceil)  + 5).
  \end{align*}
  Combining the above bounds yields
  \[
    \Cxity{d} = {\textstyle\frac{7}{2}}\timepm{\lceil d/2 \rceil} \log_2(\lceil d/2 \rceil)
              + {\textstyle\frac{35}{2}}\timepm{\lceil d/2 \rceil}
              + {\textstyle\frac{3}{2}}\timepm{d}
              + \bigO{d \log(d)},
  \]
  and the claimed complexity bound follows.
  \qed
\end{proof}

\subsubsection{Complexity bound for uniform shifts and generic degrees}
\label{sec:approx:pmbasis:cx_generic}

Recall that the input shift \(\ash\) is said to be uniform if
\(\ash=0\). This case brings additional degree properties for the computed bases, on top of those
listed at the beginning of the previous subsection. For example, by definition
of a weak Popov form, one has \(\mathring{p}_{01} < \mathring{p}_{00}\) as well
as \(\mathring{p}_{10} \le \mathring{p}_{11}\). One technical point is that the
shift might not remain uniform all along the computation. Indeed, the algorithm
uses the diagonal degrees of the first recursive call to form the new shift for
the second recursive call yielding the basis \(\basisr\): using notation from
\cref{lem:dnc_relbas}, this new shift is \(\ashr =
\mathring{p}_{00} - \mathring{p}_{11}\). Therefore, if some degree
unbalancedness appeared in \(\basis\), it may be that the second recursive call
is computed with a non-uniform \(\ashr\), making the assumption of a
uniform input shift difficult to exploit when analyzing the overall complexity.

Yet, if the input polynomials \((\apo,\bpo)\) satisfy some genericity
condition, such a shift unbalancedness will not arise: roughly, the key idea is
that generically the shift will not drift towards more unbalancedness and, on
the contrary, it rather tends to become closer to uniform during the computation.
More precisely, starting from the uniform shift, generically the degrees
\(\mathring{p}_{00}\) and \(\mathring{p}_{11}\) will differ by at most \(1\),
and a similar property holds in all recursive calls. This follows from the
results in \cref{app:approx:genericity}, which state in particular that,
generically, \((\mathring{p}_{00},\mathring{p}_{11}) = \delta_{\dreca,0} =
(\lceil \dreca/2 \rceil, \lfloor \dreca/2 \rfloor)\). Thus, depending on the
parity of \(\dreca\), the second call involves either the uniform shift
\(\ashr=0\) again, or a shift with the smallest possible unbalancedness,
specifically \(\ashr = 1\). In the latter case \(\ashr = 1\),
subsequent computations will not increase this unbalancedness further; in other
words, as we will see in the next proof, the shift will always have amplitude
at most \(1\).

\begin{proposition}
  \label{prop:pmbasis:cx_generic}%
  Let \(d\in\NN\) and let \(\ash \in \{-1,0,1\}\) be a shift with amplitude
  \(0\) or \(1\). Let \((\apo,\bpo) \in \pR_{<d}^2\) be polynomials that
  satisfy the genericity condition of \cref{cor:approx_generic_mindeg}, that
  is, \(\bar\Delta_{d,\ash}\) does not vanish at the coefficients of
  \((\apo,\bpo)\). Then, the call \(\Call{algo:PM-Basis2}{d, \apo, \bpo,
  \ash}\) costs
  \[
    \textstyle
    \min(\frac{\cstmatmul+6}{4}, \frac{\cstmatmul}{2})\, \timepm{\lceil d/2 \rceil}\log_2(\lceil d/2 \rceil)
    + \frac{5\cstmatmul+42}{8} \timepm{d} + \bigO{d\log(d)}
  \]
  operations in \(\field\), where \(\cstmatmul \le 7\) is a constant for \(2
  \times 2\) polynomial matrix multiplication.
\end{proposition}
Note that this bound coincides with that in \cref{prop:pmbasis:cx_worst_case}
when one uses naive \(2 \times 2\) matrix multiplication with \(\cstmatmul=8\),
explaining why we report a single bound in
\cref{thm:appbasis:cx,prop:appbasis:gcdsensitive,prop:appbasis:shiftreduction}.
Still, one may observe from the proof below that the arguments provide a
slightly better constant in front of \(\timepm{d}\) than the one in the
statement above. If \(\cstmatmul \le 6\), the proof explicitly shows
\(5\cstmatmul/4\), and otherwise the generic degree bound on \(\basis\) allows
one to slightly simplify the analysis of \cref{prop:pmbasis:cx_worst_case} and
get a constant lower than \((5\cstmatmul+42)/8\). For the sake of readability,
we prefer to keep a single unified bound \((5\cstmatmul+42)/8\), equal to
\(41/4\) when \(\cstmatmul=8\).

\begin{proof}
  First consider the top-level node in the recursion tree, and observe that
  both recursive calls preserve the fact that the shift has amplitude \(0\) or
  \(1\). The first call is with the shift \(\ash\), and yields a
  \(\ash\)-weak Popov basis \(\basis\) of \(\appmod{\dreca}{\apo,\bpo}\).
  Since \(\dreca \ge 1\) and \(\ash\)
  has amplitude \(0\) or \(1\), we have \(-\dreca \le \ash \le \dreca\)
  and thus \(\basis\) has \(\ash\)-pivot degree \((\mathring{p}_{00},
  \mathring{p}_{11}) = (\lceil (\dreca-\ash)/2 \rceil, \lfloor
  (\dreca+\ash)/2 \rfloor)\) according to
  \cref{cor:approx_generic_mindeg}. Then, the second call is with the shift
  \(\ashr = \ash + \mathring{p}_{00} - \mathring{p}_{11} =
  \lceil (\dreca+\ash)/2 \rceil - \lfloor (\dreca+\ash)/2 \rfloor\),
  which has amplitude \(0\) or \(1\). The second part of
  \cref{cor:approx_generic_mindeg} ensures that the basis \(\basisr\) returned
  by the second call has \(\ashr\)-pivot degree \((\lceil
  (\drecb-\ashr)/2 \rceil, \lfloor (\drecb+\ashr)/2 \rfloor)\). One
  can easily check that this implies \(\deg(\basis) \le \lceil (\dreca+1)/2
  \rceil\) and \(\deg(\basisr) \le \lceil (\drecb+1)/2 \rceil\).

  Furthermore, \cref{cor:approx_generic_mindeg} also ensures that all recursive
  calls down the recursion tree enjoy the same properties: their respective
  input shifts have amplitude \(0\) or \(1\), and all intermediate approximant
  bases have the generically expected degrees. This generalizes the above
  degree bounds for \(\basis\) and \(\basisr\): any intermediate basis at some
  order \(k\) has degree at most \(\lceil (k+1)/2 \rceil\).

  This degree bound leads to the improved complexity bound, compared to the
  general case in \cref{sec:approx:pmbasis:cx_worst_case}. As seen above, both
  bases \(\basis\) and \(\basisr\) have degree at most \(\lceil (d+2)/4
  \rceil\). Thus, the product \(\basisr \basis\) at
  \cref{step:pmbasis:multiplication} of \cref{algo:PM-Basis2} can be computed
  using \(\cstmatmul \, \timepm{\lceil d/4 \rceil} + \bigO{d}\) operations in
  \(\field\). This improves upon the bound \(4 \timepm{\lceil d/2 \rceil} +
  \bigO{d}\) from \cref{lem:appbas_cx:multiplication} as soon as one uses a
  nontrivial \(2 \times 2\) polynomial matrix multiplication, such as
  Strassen's multiplication with \(\cstmatmul = 7\). Apart from this point, if
  we keep the same analysis as in the proof of \cref{prop:pmbasis:cx_worst_case}
  we arrive at the complexity equation
  \[
    \Cxity{d} = 2\Cxity{\lceil d/2 \rceil} + {\textstyle\frac{3}{2}} \timepm{d} + {\textstyle\frac{\cstmatmul}{2}} \timepm{\lceil d/2 \rceil} + \bigO{d},
  \]
  and the same arguments as in the proof of
  \cref{prop:pmbasis:cx_worst_case} yield the claimed cost bound
  with constants \((\cstmatmul+6)/4\)
  and \((5\cstmatmul+42)/8\).

  Now, if the context allows for a low constant \(\cstmatmul\), this may be
  exploited for the residual computation, following the approach in
  \cite[Sec.\,4.2]{vdHoeven2025}. Indeed, \cref{step:pmbasis:residual} asks to
  compute the high \(d/2\) coefficients of a product of \(\basis\) of degree
  about \(d/4\) with a \(2\times 1\) vector of degree \(\le 3d/4\) (where we
  have used \cref{lem:midprod_with_gap}). In the latter reference, it is
  observed that one may expand this vector into a \(2 \times 2\) matrix of
  degree \(\le d/2\) and rather compute the high \(d/4\) coefficients of a
  product of \(\basis\) of degree about \(d/4\) with this obtained \(2\times
  2\) matrix of degree \(\le d/2\). This is then a standard matrix middle
  product, computed in \(\cstmatmul \timepm{\lceil d/4 \rceil} + \bigO{d}\)
  operations in \(\field\). This leads to the equation
  \(
    \Cxity{d} = 2\Cxity{\lceil d/2 \rceil} + \cstmatmul \timepm{\lceil d/2 \rceil} + \bigO{d},
  \)
  and a simplified version of the proof of
  \cref{prop:pmbasis:cx_worst_case} yields the claimed bound
  with constants \(\cstmatmul/2\) and \(5\cstmatmul/4\), which
  is less than \((5\cstmatmul+42)/8\).
  \qed
\end{proof}

\subsection{Algorithmic optimizations and complexity bound refinements}
\label{sec:approx:appbasis}

We now augment the above algorithm with two kinds of optimizations. First, if
one or both input polynomials \(\apo\) and \(\bpo\) have positive valuation,
then one can reduce to a lower approximation order \(d\), at no arithmetic
cost. This will allow us to derive an improved complexity bound, presented in
\cref{sec:approx:appbasis:gcdsensitive}, which is sensitive to the order
\(\dsyz\) at which an approximant basis necessarily contains a syzygy of
\(\apo\) and \(\bpo\) (see \cref{lem:kernel_approx}). Second, if the shift
\(\ash\) has large amplitude, namely \(\abs{\ash} \ge
d\), then one can directly obtain the sought basis by computing a power series
expansion of either \(\bpo \apo^{-1}\) or \(\apo \bpo^{-1}\) at order \(d\). In
this case, the basis is in Hermite normal form (HNF), either lower or upper
depending on which entry of the shift is larger. A variant of this optimization
will allow us, in \cref{sec:approx:appbasis:unishift}, to handle more
efficiently the cases where the input shift is far from uniform, and also to
lift the requirement that the amplitude of the shift should be at most \(1\) in
the generic case.

Using these optimizations to design \cref{algo:AppBasis2-rec,algo:AppBasis2},
and combining the corresponding results detailed in
\cref{prop:appbasis:gcdsensitive,prop:appbasis:shiftreduction}, we obtain
\cref{thm:appbasis:cx} in \cref{sec:intro}. As outlined above, this gives a
complexity bound which is sensitive to \(\dsyz\) and to the amplitude of the
shift, and which is valid for the generic case even for non-uniform shifts. 

\subsubsection{Optimizations for positive valuations and for unbalanced shifts}
\label{sec:approx:appbasis:optimizations}

The first above-mentioned optimization is based on
\cref{lem:appbas:opti:valuation}, and implemented in
\crefrange{step:appbasis:opti:valuation:begin}{step:appbasis:opti:valuation:end}
of \cref{algo:AppBasis2-rec}, while the second one is based on
\cref{lem:appbas:opti:shift} and implemented in
\crefrange{step:appbasis:opti:shift:begin}{step:appbasis:opti:shift:end}.

Concerning the first optimization, one can consider the simpler case where one
of the input polynomials is zero, or both are zero: one can give the
approximant basis without any computation. Indeed, if \(\apo = \bpo = 0\), then
\(\appmod{d}{\apo,\bpo} = \pR^2\), whose \(\ash\)-Popov basis is the
identity matrix, for any shift. If \(\apo=0\) and \(\valuation{\bpo} = 0\),
then \(\appmod{d}{\apo,\bpo} = \pR \times x^d \pR\), whose
\(\ash\)-Popov basis is \([\begin{smallmatrix} 1 & 0 \\ 0 & x^{d}
\end{smallmatrix}]\) for any shift. The next lemma generalizes this to the case
where \(\apo\) or \(\bpo\) has positive valuation.

\begin{lemma}
  \label{lem:appbas:opti:valuation}%
  Let \(d \ge 0\), \((\apo,\bpo) \in \pR_{<d}^2\), and \(\ash \in
  \mathbb{Z}\). Define \(v_\apo = \min(d,\valuation{\apo})\) and \(v_\bpo =
  \min(d,\valuation{\bpo})\). If \(v_\apo \ge v_\bpo\), then the
  \(\ash\)-Popov basis of \(\appmod{v_\apo}{\apo,\bpo}\) is
  \([\begin{smallmatrix} 1 & 0 \\ 0 & x^{v_{\apo}-v_{\bpo}}
  \end{smallmatrix}]\). Similarly, if \(v_\bpo \ge v_\apo\), then the
  \(\ash\)-Popov basis of \(\appmod{v_\bpo}{\apo,\bpo}\) is
    \([\begin{smallmatrix} x^{v_{\bpo}-v_{\apo}} & 0 \\ 0 & 1
    \end{smallmatrix}]\).
\end{lemma}

\begin{proof}
  Assuming \(v_\apo \ge v_\bpo\), one easily verifies that
  \(\appmod{v_\apo}{\apo,\bpo} = \pR \times x^{v_{\apo}-v_{\bpo}} \pR\), whose
  shifted Popov basis is \([\begin{smallmatrix} 1 & 0 \\ 0 &
  x^{v_{\apo}-v_{\bpo}} \end{smallmatrix}]\) for any shift. The second case
  is similar.
  \qed
\end{proof}

This property can be exploited within the general recursive approach based on
\cref{lem:dnc_relbas}. Assume for example that \(v_\apo > v_\bpo\), and split
the approximation order as \(d = v_\apo + (d-v_\apo)\), so as to rely on two
subproblems at orders \(v_\apo\) and \(d-v_\apo\). The above lemma yields the
first approximant basis at order \(v_\apo\), for free: \(\basis =
[\begin{smallmatrix} 1 & 0 \\ 0 & x^{v_{\apo}-v_{\bpo}} \end{smallmatrix}]\).
Then one can deduce the residual polynomials \((\apor,\bpor)\) as in
\cref{step:pmbasis:residual}, which turns out not to require any arithmetic
operation either: \(\apor = x^{-v_{\apo}}\apo\), and \(\bpor =
(x^{-v_\bpo}\bpo) \rem x^{d - v_\apo}\). Then, the actual computation consists
in computing the second approximant basis \(\basisr\) at order \(d-v_\apo\),
for the updated shift \(\ash+\deg(1) - \deg(x^{v_{\apo}-v_{\bpo}}) =
\ash + v_{\bpo} - v_{\apo}\).  Finally, one returns the product \(\basisr
\basis\), which is found without any arithmetic operation. Proceeding similarly
for the case \(v_\apo > v_\bpo\), this explains the correctness of
\crefrange{step:appbasis:opti:valuation:begin}{step:appbasis:opti:valuation:end}
of \cref{algo:AppBasis2-rec}.

Now we turn to the second kind optimization, targeting unbalanced shifts.

\begin{lemma}
  \label{lem:appbas:opti:shift}
  Let \(d \ge 0\), \((\apo,\bpo) \in \pR_{<d}^2\), and \(\ash \in \mathbb{Z}\). If \(\ash
  \le -d\) and \(\valuation{\apo} = 0\), then the \(\ash\)-Popov basis of \(\appmod{d}{\apo,\bpo}\) is
  \([\begin{smallmatrix} x^{d} & 0 \\ -c & 1 \end{smallmatrix}]\), where \(c =
  (\bpo \apo^{-1}) \rem x^{d}\) is the truncation at order \(d\) of the power
  series expansion of \(b a^{-1}\). Similarly, if \(\ash \ge d\) and \(\valuation{\bpo} = 0\) then
  the \(\ash\)-Popov basis of \(\appmod{d}{\apo,\bpo}\) is
  \([\begin{smallmatrix} 1 & -c \\ 0 & x^{d} \end{smallmatrix}]\), where \(c =
  (\apo \bpo^{-1}) \rem x^{d}\).
\end{lemma}

\begin{proof}
  Assume \(\ash \le -d\), and let \(H = [\begin{smallmatrix} x^d & 0 \\ -c
  & 1 \end{smallmatrix}]\). Since \(d + \ash \le 0\) and \(\deg(c) < d\),
  \(H\) is in \(\ash\)-Popov form (and, as a matter of fact, also in
  lower HNF). By construction, each row of \(H\) is in
  \(\appmod{d}{\apo,\bpo}\). Furthermore, any \((q_0,q_1) \in
  \appmod{d}{\apo,\bpo}\) satisfies \(q_0 \apo + q_1 \bpo = 0 \bmod x^d\),
  which is equivalent to \(q_0 = -q_1 c \bmod x^d\), and the latter means that
  \((q_0,q_1)\) is a \(\pR\)-linear combination of the rows of \(H\).  Hence
  \(H\) is the \(\ash\)-Popov basis of \(\appmod{d}{\apo,\bpo}\). The
  case \(\ash \ge d\) is similar.
  \qed
\end{proof}

Thus, for such unbalanced shifts, we compute the approximant basis \(H\) at the
cost of one power series division at precision \(d\). The latter can be done
via power series inversion in \(3\timepm{d} + \bigO{d}\) \cite[Thm.\,9.4 and
Exe.\,9.6]{GathenGerhard2013} followed by a power series multiplication in
\(\timepm{d}\); this has been accelerated to \(\frac{5}{2} \timepm{d} +
\bigO{d}\) in \cite[Sec.\,4]{HanrotQuerciaZimmermann2004} by exploiting the
middle product and an idea from \cite{KarmMarkstein1997}.

\begin{remark}
  One may wonder what differs if, instead of computing the above explicit basis
  \(H\), one uses the recursive approach of \algoName{algo:PM-Basis2} to obtain
  an approximant basis \(\basis\). First, \(\basis\) has almost the same form
  as \(H\): its diagonal entries are also \(x^d\) and \(1\) up to
  multiplication by a constant, and its bottom-left entry may have degree up to
  \(d\). In other words, \(\basis\) is equal to \(H\) up to left multiplication
  by an invertible lower triangular matrix in \(\mKK{2}{2}\). Thus, using
  \(\bigO{d}\) additional field operations, \(\basis\) reveals the power series
  quotient \(c\). For this reason, one cannot expect the recursion of
  \algoName{algo:PM-Basis2} to be faster than the direct computation of \(c\) with
  the fastest available method.
  Further, since the unbalancedness of the shift propagates to recursive calls,
  all intermediate approximant bases have a similar form. This form makes the
  residual computation and basis multiplication at
  \cref{step:pmbasis:residual,step:pmbasis:multiplication} faster (at the level
  of constant factors), and one can see that these steps involve operations
  very close to those performed in the Newton-iteration-based power series
  division to find \(c\). Yet, a major difference remains in how the recursion
  is organized: the latter power series division uses a single recursive call
  at order \(d/2\) while \algoName{algo:PM-Basis2} uses two such calls. In the
  complexity bound, this translates as a gain of a factor \(\log(d)\) for the
  former.
  \qed
\end{remark}

\begin{algorithm}[ht]
  \algoCaptionLabel{AppBasis2-rec}{d, \apo, \bpo, \ash}
  \begin{algorithmic}[1]

    \Require \(d \in \mathbb{N}\), polynomials \((\apo,\bpo) \in \pR_{<d}^2\), shift \(\ash \in \mathbb{Z}\)

    \Ensure a basis \(P \in \mpR{2}{2}_{\le d}\) of \(\appmod{d}{\apo,\bpo}\) in \(\ash\)-weak Popov form

    \State\InlineIf{\(d \le 2\)}{\Return \(\Call{algo:AppBasis2-base}{d, \apo,\bpo, \ash}\)}

    %% optimisation: \apo or \bpo has positive valuation
    \LComment{\(\apo\) or \(\bpo\) has positive valuation: reduce to lower order}
        \label{step:appbasis:opti:valuation:begin}
    \State \(v_\apo \gets \min(d,\valuation{\apo})\); \(v_\bpo \gets \min(d,\valuation{\bpo})\)
    \State \InlineIf{\(v_\apo = v_\bpo = d\)}{\Return the \(2 \times 2\) identity matrix}
            \MyComment{\(\apo = \bpo = 0\)}
    \If{\(v_\apo > 0\) \OR{} \(v_\bpo > 0\)}
      \If{\(v_\apo > v_\bpo\) \OR{} \(v_\apo = v_\bpo > 0\)}
            \State \(\basis \gets [\begin{smallmatrix} 1 & 0 \\ 0 & x^{v_{\apo}-v_{\bpo}} \end{smallmatrix}]\);
                   \(\drecb \gets d-v_{\apo}\);
                   \(\apor \gets x^{-v_{\apo}}\apo\);
                   \(\bpor \gets (x^{-v_\bpo}\bpo) \rem x^{\drecb}\);
                   \(\ashr \gets \ash + v_{\bpo}-v_{\apo}\)
      \ElsIf{\(v_\bpo > v_\apo\)}
            \State \(\basis \gets [\begin{smallmatrix} x^{v_{\bpo}-v_{\apo}} & 0 \\ 0 & 1 \end{smallmatrix}]\);
                   \(\drecb \gets d-v_{\bpo}\);
                   \(\apor \gets (x^{-v_{\apo}}\apo) \rem x^{\drecb}\);
                   \(\bpor \gets x^{-v_\bpo}\bpo\);
                   \(\ashr \gets \ash+v_{\bpo}-v_{\apo}\);
      \EndIf
      \State \(\basisr \gets \Call{algo:AppBasis2-rec}{\drecb, \apor, \bpor, \ashr}\)
      \State \Return \(\basisr \basis\)
            \label{step:appbasis:opti:valuation:end}
            \MyComment{if going beyond this point, \(\valuation{\apo} = \valuation{\bpo} = 0\)}
    \EndIf

    %% optimisation: Hermite shift
    \LComment{the shift has amplitude \(\ge d\): return the basis in lower or upper HNF}
        \label{step:appbasis:opti:shift:begin}
    \State\InlineIf{\(\ash \le -d\)}{\Return \([\begin{smallmatrix} x^{d} & 0 \\ -c & 1 \end{smallmatrix}]\), where \(c \gets (\bpo \apo^{-1}) \rem x^{d}\)}
    \State\InlineIf{\(\ash \ge d\)}{\Return \([\begin{smallmatrix} 1 & -c \\ 0 & x^{d} \end{smallmatrix}]\), where \(c \gets (\apo \bpo^{-1}) \rem x^{d}\)}
        \label{step:appbasis:opti:shift:end}

    \LComment{divide and conquer}

    \State \(\dreca \gets \lfloor d / 2 \rfloor\); \(\drecb \gets d - \dreca\)

    \State \(\basis =
        [\begin{smallmatrix}
          p_{00} & p_{01} \\
          p_{10} & p_{11}
        \end{smallmatrix}]
        \gets \Call{algo:AppBasis2-rec}{\dreca, \apo, \bpo, \ash}\)

    \State \(\apor \gets \midprod{p_{00}, \apo, \dreca, d} + \midprod{p_{01}, \bpo, \dreca, d}\)
        \MyComment{\(\apor = (x^{-\dreca} (p_{00} \apo + p_{01} \bpo)) \rem x^{\drecb}\)}
    \State \(\bpor \gets \midprod{p_{10}, \apo, \dreca, d} + \midprod{p_{11}, \bpo, \dreca, d}\)
        \MyComment{\(\bpor = (x^{-\dreca} (p_{10} \apo + p_{11} \bpo)) \rem x^{\drecb}\)}%
        %\label{step:appbasis:residual}

    \State \(\basisr \gets \Call{algo:AppBasis2-rec}{\drecb, \apor,\bpor, \ash+\deg(p_{00})-\deg(p_{11})}\)

    \State \Return \(\basisr \basis\)
        %\label{step:appbasis:multiplication}
  \end{algorithmic}
\end{algorithm}

The above two additions to \algoName{algo:PM-Basis2} yield
\cref{algo:AppBasis2-rec}. As highlighted above, the complexity of this
modified algorithm is the same as or better than the former, depending on input
properties.
%\crefrange{step:appbasis:opti:valuation:begin}{step:appbasis:opti:valuation:end}
%reduce the approximation order at no cost, while
%\crefrange{step:appbasis:opti:shift:begin}{step:appbasis:opti:shift:end}
%compute the output basis via a single power series division.
Carrying out a more precise analysis of the general impact of the added steps
on efficiency is left as a perspective. This would mean to seek a complexity
bound which depends on the amount of valuations and shift unbalancedness
encountered in recursive calls. This is closely related, in the XGCD context,
to an analysis of the half-gcd algorithm which would take into account the
successive degrees of the polynomials in the remainder sequence (see e.g.\
\cite[Rem.\,16]{vdHoeven2025} on that topic). Although interesting, this
appears to be quite technical, and would only bring improvements for situations
that do not happen generically.

As announced above, we will however exploit these additions for complexity
bound refinements on two specific aspects, presented in the next subsections.

\subsubsection{Refinement: sensitivity to the GCD degree}
\label{sec:approx:appbasis:gcdsensitive}

We start with an example, illustrating that thanks to the improved handling of
input polynomials with positive valuation, the algorithm now performs better
when the quantity \(\dsyz\) from \cref{lem:kernel_approx} is small compared to
\(d\), which roughly happens when the two input polynomials have a large GCD.

\begin{example}
  \label{ex:approx_gcd_sensitive}%
  Consider the base field \(\field = \ZZ/7\ZZ\) and take the polynomials
  \begin{align*}
    \apo & = 2x^{14} + 3x^{13} + 5x^{12} + 4x^{11} + 4x^{10} + x^8 + 3x^6 + 4x^5 + 5x^4 + 2x^3 + 3x^2 + 2x + 4 \\
    \text{and } \bpo & = x^{14} + 5x^{13} + 5x^{11} + 6x^9 + 3x^8 + 5x^7 + 6x^6 + x^5 + 4x^4 + x^3 + 3x^2 + 3x + 6, \\
    \text{with } \pg & = \gcd(\apo,\bpo) = x^{10} + 3x^9 + 2x^8 + 4x^7 + x^4 + 6x^3 + 2x^2 + 3x + 6.
  \end{align*}
  Here, \(n = 14\), \(m = 14\), \(\ell = 10\). We use the uniform shift
  \(\ash = 0\) and compute a weak Popov approximant basis at order
  \(d = m+n+1 = 29\); we have \(\dsyz = 29 - 20 = 9\).

  Consider the top-level recursive calls, at respective orders \(\lfloor 29/2
  \rfloor = 14\) and \(\lceil 29/2 \rceil = 15\). Due to the large gcd between
  \(\apo\) and \(\bpo\), we have \(\dsyz \le \dreca\), so the first recursive
  call yields a basis \(\basis\) which has one row of the form
  \([-\lambda\bpo/\pg \;\; \lambda\apo/\pg]\) for some
  \(\lambda\in\field\setminus\{0\}\) (see \cref{lem:kernel_approx}); in the
  present case, it is the second row. Thus the residual polynomial \(\bpor\)
  corresponding to that row is zero, since it is by definition
  \[
    \bpor =
    \left(x^{-\dreca} \left(-\frac{\lambda\bpo}{\pg} \apo + \frac{\lambda\apo}{\pg} \bpo\right)\right) \rem x^{d-\dreca}
    = 0.
  \]
  Specifically, here we have
  \begin{align*}
    \basis & =
    \begin{bmatrix}
      x^{10} + 5x^6    &  3x^9 + 5x^8 + 6x^6 \\
      3x^4 + 6x^3 + 4x^2 + 3   & x^4 + 2x^3 + 5x^2 + 5
    \end{bmatrix} , \text{ and}
    \\
    \apor & = 2x^{10} + 6x^9 + 4x^8 + x^7 + 2x^4 + 5x^3 + 4x^2 + 6x + 5.
  \end{align*}
  Once \(\apor\) and \(\bpor\) are computed, the optimizations introduced at
  \crefrange{step:appbasis:opti:valuation:begin}{step:appbasis:opti:valuation:end}
  of \cref{algo:AppBasis2-rec} ensure that the second recursive call does not use
  any field operation and returns the basis \(\basisr = [\begin{smallmatrix}
  x^{\drecb} & 0 \\ 0 & 1 \end{smallmatrix}]\). Then, due to the form of
  \(\basisr\), the final product \(\basisr \basis\) is obtained for free.  In
  short, the second half of the recursion tree was totally avoided, and this
  happens without a priori knowledge of the GCD \(\pg\) or of its degree
  \(\ell\).
  %%: the algorithm computes the first basis \(\basis\), and
  %%from there the residual \((\apor,\bpor)\), and it is the fact that
  %%\(\bpor=0\) which significantly simplifies subsequent computations by making
  %%the algorithm choose the branch of
  %%\crefrange{step:appbasis:opti:valuation:begin}{step:appbasis:opti:valuation:end}.

  To get a finer understanding of complexity aspects, let us go further in the
  description of the recursion tree, represented in
  \cref{fig:ex:approx_gcd_sensitive}.

  \begin{figure}[htb]
    \centering
    \begin{tikzpicture}[
      every node/.style = {minimum width = 2em, draw, rectangle},
      level 1/.style = {sibling distance = 50mm},
      level 2/.style = {sibling distance = 50mm},
      level 3/.style = {sibling distance = 20mm},
      level 4/.style = {sibling distance = 10mm},
      ]
      \node {\(\substack{\texttt{call C}_0 \\ d = 29}\)}
        child {node {\(\substack{\texttt{call C}_{11} \\ \lfloor d/2 \rfloor = 14}\)}
          child {node {\(\substack{\texttt{call C}_{21} \\ d_1 = 7}\)}
            child {node {\(\substack{\texttt{call C}_{31} \\ \lfloor d_1/2 \rfloor = 3}\)}
              child {node {1}}
              child {node {2}}
            } % c31
            child {node {\(\substack{\texttt{call C}_{32} \\ \lceil d_1/2 \rceil = 4}\)}
              child {node {2}}
              child {node {2}}
            } % c32
          } % c21
          child {node {\(\substack{\texttt{call C}_{22} \\ \lceil \hat{d}_1 \rceil = 7}\)}
            child {node {\(\substack{\texttt{call C}_{33} \\ d_2 = 3}\)}
              child {node {1}}
              child {node {2}}
              } % c33
            child {node [dashed] {\(\substack{\texttt{call C}_{34} \\ \lceil \hat{d}_2 \rceil = 4}\)}}
          } % c34
        } % c11
        child {node [dashed] {\(\substack{\texttt{call C}_{12} \\ \lceil d/2 \rceil = 15}\)}
        }; % c0
    \end{tikzpicture}
    \caption{\em Recursion tree for \cref{ex:approx_gcd_sensitive}.  The
      notation \(d_1,d_2,\hat{d}_1,\hat{d}_2\) is used to help reading the
      proof of \cref{prop:appbasis:gcdsensitive}.
      As a consequence of \cref{lem:kernel_approx,lem:appbas:opti:valuation}, and
      because of the (a priori unknown) large GCD between \(\apo\) and
      \(\bpo\) implying \(\dsyz = 9\), the recursive calls \(\texttt{C}_{34}\)
      and \(\texttt{C}_{12}\) do not use any field operation and output an
      approximant basis of the form
      \([\begin{smallmatrix}
        x^d & 0 \\
        0 & 1
      \end{smallmatrix}]\).}
    \label{fig:ex:approx_gcd_sensitive}
  \end{figure}

  The above-discussed first recursive call, labelled \(\texttt{C}_{11}\), is
  itself divided into two calls both at order \(7\). Since \(7 < \dsyz = 9\),
  we do not expect the valuation optimization to bring any significant gain in
  the call \(\texttt{C}_{21}\). It yields the basis
  \[
    \basis_{21} =
    \begin{bmatrix}
      x^4 + x^3 + 6x^2 + 4x + 5 &  4x^3 + x^2 + 2x + 6 \\
      x^3 + 6x^2 + 4x           &      x^3 + 3x^2 + 2x
    \end{bmatrix},
  \]
  and the input for the call \(\texttt{C}_{22}\) follows: the residual polynomials are
  \[
    \apor_{22} = 2x^5 + x^4 + x^3 + 5x^2 + 4
    \quad\text{and}\quad
    \bpor_{22} = 3x^4 + 4x^3 + 6x^2 + 2x + 4,
  \]
  and the updated shift is \(\ashr_{22} = 0+4-3 = 1\).
  Now, recall that a syzygy should appear as soon as we reach overall order
  \(\dsyz = 29 - 20 = 9\), and we already have processed order
  \(7\). So within the call \(\texttt{C}_{22}\), a syzygy should appear as soon
  as we reach order \(\dsyz - 7 = 2\). Therefore, for the calls
  \(\texttt{C}_{33}\) and \(\texttt{C}_{34}\), the situation is similar to the
  top-level calls \(\texttt{C}_{11}\) and \(\texttt{C}_{12}\) discussed before:
  the call \(\texttt{C}_{33}\) yields a basis which contains a syzygy, hence
  the corresponding residual polynomial is zero, and the call
  \(\texttt{C}_{34}\) will not use any arithmetic operation.

  The overall order \(9\) is reached during the processing of the second leaf
  at order \(2\) of the call \(\texttt{C}_{33}\). The latter call yields
  the following basis, along with the input polynomials for
the call \(\texttt{C}_{34}\):
  \[
    \basis_{33} =
    \begin{bmatrix}
      x^2 &  6x^2 \\
      2 & x + 5
    \end{bmatrix}, \;\;
    \apor_{34} = 5x^3 + 4x^2 + 6x + 5, \;\;
    \bpor_{34} = 0, \;\;
    \ashr_{34} = 1+2-1 = 2.
  \]
  Then, the call \(\texttt{C}_{34}\) yields the trivial basis
  \(\basis_{34} =
  [\begin{smallmatrix}
    x^4 & 0 \\
    0 & 1
  \end{smallmatrix}]
  \),
  so that the basis multiplication finalizing the call \(\texttt{C}_{22}\) is
  free and gives \(\basis_{22} = \basis_{34} \basis_{33} =
  [\begin{smallmatrix}
      x^6 &  6x^6 \\
      2 & x + 5
  \end{smallmatrix}]\).
  Finally, the call \(\texttt{C}_{11}\) outputs the product \(\basis_{11} =
  \basis_{22} \basis_{21}\), which is the matrix \(\basis\) described above.
  Observe that one could exploit the valuation in the first row of
  \(\basis_{22}\) to compute this product faster.
  % This valuation is here \(6\), and for other polynomials with the same
  % \(m,n,\ell\) it would be at least \(5 = 14 - \dsyz\) according
  % to \cref{lem:kernel_approx}.

  Altogether, on this example, the algorithm performs recursive calls as
  expected (not using the valuation optimization) in the two subtrees of the
  calls \(\texttt{C}_{21}\) and \(\texttt{C}_{33}\). At this stage, after
  completing \(\texttt{C}_{33}\), it has reached total approximation order
  \(7+3=10\), which is just above \(\dsyz = 9\), and therefore all subsequent,
  not yet encountered nodes of the recursion tree (i.e.\ here
  \(\texttt{C}_{12}\) and \(\texttt{C}_{34}\)) are simplified: they do not
  involve further recursive calls and directly return trivial bases for free.
  Apart from the work in calls \(\texttt{C}_{21}\) and \(\texttt{C}_{33}\), the
  costly tasks are the computation of residuals (to find the input for calls
  \(\texttt{C}_{22}\), \(\texttt{C}_{34}\), and \(\texttt{C}_{12}\)) and the
  multiplication of bases in order to obtain \(\basis_{11} = \basis_{22} \basis_{21}\).
  \qed
\end{example}

\begin{proposition}
  \label{prop:appbasis:gcdsensitive}%
  Let \(d \in \mathbb{N}\) and \(\ash \in \mathbb{Z}\). Let \((\apo,\bpo)
  \in \pR_{<d}^2\), both nonzero, with respective degrees \((n,m)\). Consider
  the quantity from \cref{lem:kernel_approx},
  \[
    \dsyz = \max(m+n, 2n - \ash, 2m + \ash) + 1 - 2\deg(\gcd(\apo,\bpo)),
  \]
  and suppose \(\dsyz \le d\). Then \cref{algo:AppBasis2-rec} computes an
  \(\ash\)-weak Popov basis of \(\appmod{d}{\apo,\bpo}\) using
  \[
    \textstyle
    \min(\frac{\cstmatmul+6}{4}, \frac{\cstmatmul}{2}) \timepm{\lceil \dsyz/2 \rceil}\log_2(\lceil \dsyz/2 \rceil)
    + \frac{5\cstmatmul+74}{8} \timepm{\dsyz} + {\textstyle\frac{5}{2}} \timepm{d} + \bigO{d\log(d)}
  \]
  operations in \(\field\), where \(\cstmatmul = 8\) in the general case, and
  \(\cstmatmul \le 7\) is a constant for \(2 \times 2\) polynomial matrix
  multiplication in the case where \(\ash\) has amplitude at most \(1\)
  and \((\apo \rem x^{\dsyz},\bpo \rem x^{\dsyz})\) satisfies the genericity
  condition of \cref{cor:approx_generic_mindeg}, that is,
  \(\bar\Delta_{\dsyz,\ash}\) does not vanish at the first \(\dsyz\)
  coefficients of \(\apo\) and \(\bpo\).
\end{proposition}

\begin{proof}
  Concerning the coprimeness assumption in \cref{lem:kernel_approx}, note that
  the first steps of \cref{algo:AppBasis2-rec} make sure that the valuation
  of \(\apo\) and \(\bpo\) is zero; this process thus ensures the sought
  coprimeness while reducing the order \(d\) and leaving \(\dsyz\) unchanged.
  Observe also that \(\dsyz > 0\) by construction, since \(2\deg(\gcd(\apo,\bpo))
  \le m+n\). If \(\dsyz = d\), the result comes from
  \cref{prop:pmbasis:cx_worst_case,prop:pmbasis:cx_generic}. Now suppose
  \(\dsyz < d\). \cref{fig:proof:approx_gcd_sensitive} illustrates the
  recursion tree and should help reading the proof.  The orders
  \(d_1,\ldots,d_r\) can be defined recursively as follows. Let \(d_{0} = d\)
  and \(\hat{d}_0 = d_{0} / 2^{i_0}\). Suppose that for some \(j \ge 1\) we
  have built \(d_0,\ldots,d_{j-1} \in \ZZp\) and
  \(\hat{d}_0,\ldots,\hat{d}_{j-1} > 0\), with \(\dsyz - (d_1 + \cdots +
  d_{j-1}) < \lceil \hat{d}_{j-1} \rceil\) (for \(j=0\), this is our assumption
  \(\dsyz < d\)).  If \(\dsyz \le d_1 + \cdots + d_{j-1}\), the process ends:
  \(r = j-1\). Else, either \(\lceil \hat{d}_{j-1} \rceil = 2\) and we define
  \(d_j = 2 = \dsyz - (d_1 + \cdots + d_{j-1}) + 1\) and \(\hat{d}_j = 0\)
  (then the process ends with \(r=j\)), or \(\lceil \hat{d}_{j-1} \rceil \ge
  3\) and we let \(\hat{d}_j = \lceil \hat{d}_{j-1} \rceil / 2^{i_j}\) and
  \(d_j = \lfloor \hat{d}_j \rfloor\), where \(i_j \in \ZZp\) is the smallest
  integer such that \(d_j \le \max(2, \dsyz - (d_1 + \cdots + d_{j-1}))\). Note
  that \(i_j > 0\) follows from \(\lceil \hat{d}_{j-1} \rceil \ge 3\) and
  \(\lceil \hat{d}_{j-1} \rceil > \dsyz - (d_1 + \cdots + d_{j-1})\). By
  construction, \(d_1+\cdots+d_r \in \{\dsyz,\dsyz+1\}\).

  \begin{figure}[ht]
    \centering
    \begin{tikzpicture}[
      every node/.style = {minimum width = 2em, draw, rectangle},
      level 1/.style = {sibling distance = 25mm},
      level 2/.style = {sibling distance = 25mm},
      %level 3/.style = {sibling distance = 20mm},
      level 4/.style = {sibling distance = 42mm},
      multi/.style={edge from parent/.style={black,draw,dashed}},
      norm/.style={edge from parent/.style={black,draw,solid}}
      ]
      \node (root) {\footnotesize \begin{minipage}{0.2\textwidth} \centering call at order \(d = \lceil \hat{d}_0 \rceil\) \\ \(\dsyz < d\) \end{minipage}}
        child
        {
          child {
            child[multi]
            {
              child[norm] {
                  node[ellipse] (leafa) {\footnotesize \begin{minipage}{0.2\textwidth} \centering normal call at order \(d_1\) \\ \(d_1 = \lfloor d / 2^{i_1} \rfloor \le \dsyz\) \end{minipage}}
                  edge from parent node[left=0.1cm,draw=none] {\footnotesize \(\lfloor d / 2^{i_1} \rfloor\)}
                }
              child[norm] {
                  node[solid] (call) {\footnotesize \begin{minipage}{0.2\textwidth} \centering call at order \(\lceil \hat{d}_1 \rceil\) \\ \(\dsyz - d_1 < \lceil \hat{d}_1 \rceil\) \end{minipage}}
                  child { edge from parent[dashed,double]
                    child[norm] {
                      node[ellipse] (leafb) {\footnotesize \begin{minipage}{0.21\textwidth} \centering normal call at order \(d_2\) \\ \(= \lfloor \lceil \hat{d}_1 \rceil / 2^{i_2} \rfloor \le \dsyz - d_1\) \end{minipage}}
                      edge from parent node[left=0.1cm,draw=none] {\footnotesize \(\lfloor \lceil \hat{d}_1 \rceil / 2^{i_2} \rfloor\)}
                    }
                    child {
                      node[solid] {\footnotesize \begin{minipage}{0.2\textwidth} \centering call at order \(\lceil \hat{d}_2 \rceil\) \\ \(\dsyz - d_1 - d_2 < \lceil \hat{d}_2 \rceil\) \end{minipage}}
                      child { edge from parent[dashed,double] }
                      edge from parent node[right=0.1cm,draw=none] {\footnotesize \(\lceil \hat{d}_2 \rceil = \lfloor \lceil \hat{d}_1 \rceil / 2^{i_2-1} \rfloor - d_2\)}
                    }
                  }
                  edge from parent node[right=0.1cm,draw=none] {\footnotesize \(\lceil \hat{d}_1 \rceil = \lceil \lfloor d / 2^{i_1-1} \rfloor / 2 \rceil = \lfloor d / 2^{i_1-1} \rfloor - d_1\)}
                }
              %edge from parent node[left=0.1cm,draw=none] {\footnotesize depth \(i_1 - 3\)}
            }
            edge from parent node[left=0.1cm,draw=none] {\footnotesize \(\lfloor d/4 \rfloor\)}
          }
          child {node[dashed] {\footnotesize free} edge from parent node[right=0.1cm,draw=none] {\footnotesize \(\lceil \lfloor d/2 \rfloor / 2 \rceil\)}}
          edge from parent node[left=0.1cm,draw=none] {\footnotesize \(\lfloor d/2 \rfloor\)}
        }
        child {node [dashed] {\footnotesize free} edge from parent node[right=0.1cm,draw=none] {\footnotesize \(\lceil d/2 \rceil\)}
        };

      \draw[<->,thin] ([xshift=-1cm]leafa.north west) -- node[fill=white,draw=none] {\footnotesize depth \(i_1\)} +(up:5.2cm);
      \draw[<->,thin] ([xshift=-1cm]leafa.south west) -- node[fill=white,draw=none] {\footnotesize depth \(i_2\)} +(down:2.1cm);
    \end{tikzpicture}
    \caption{\em The recursion tree when \(\dsyz < d\),
      simplified thanks to
      \crefrange{step:appbasis:opti:valuation:begin}{step:appbasis:opti:valuation:end}
      of \cref{algo:AppBasis2-rec} and the properties in
      \cref{lem:kernel_approx,lem:appbas:opti:valuation}. The ``free''
      calls return an approximant basis without using field operations. On the
      contrary, the ``normal'' calls involve an approximation order less
      than the local \(\dsyz\), so their cost is not expected to
      significantly benefit from the optimizations at
      \crefrange{step:appbasis:opti:valuation:begin}{step:appbasis:opti:valuation:end}.
      Below the call at order \(\lceil \hat{d}_1 \rceil\), the double dashed
      edge hides a subtree involving free calls similar to that below the
      initial call at order \(d\).
      Below the call at order \(\lceil \hat{d}_2 \rceil\), the double dashed
      edge indicate similar subtrees at orders \(\lceil \hat{d}_3 \rceil, \ldots,
      \lceil \hat{d}_r \rceil\) until reaching a leaf (\(\hat{d}_r \le 2\))
      or until overall order \(\dsyz\) has been reached (\(d_1+\ldots+d_r \ge \dsyz\)).}
    \label{fig:proof:approx_gcd_sensitive}
  \end{figure}

  The computation involves ``normal'' calls to \cref{algo:AppBasis2-rec} at
  orders \(d_1,\ldots,d_r\). Since the cost bounds for this algorithm in
  \cref{prop:pmbasis:cx_worst_case,prop:pmbasis:cx_generic} are superlinear,
  this costs less than a single call to \cref{algo:AppBasis2-rec} at order
  \(d_1+\cdots+d_r \le \dsyz+1\). This yields the bound
  \[
    \textstyle
    \min(\frac{\cstmatmul+6}{4}, \frac{\cstmatmul}{2}) \timepm{\lceil \dsyz/2 \rceil}\log_2(\lceil \dsyz/2 \rceil)
    + \frac{5\cstmatmul+42}{8} \timepm{\dsyz} + \bigO{d\log(d)},
  \]
  with \(\cstmatmul=8\) if we are in the general context of
  \cref{prop:pmbasis:cx_worst_case}, and \(\cstmatmul \le 7\) in the generic
  case of \cref{prop:pmbasis:cx_generic}. It remains to take into account the
  extra operations, excluding those in these ``normal'' calls already accounted
  for. At all levels of the recursion tree, one has to compute only one pair
  of residual polynomials and at most one basis multiplication. Indeed, the
  only costly bases multiplications are at each of the \(r\) levels where one
  splits the order into \(d_i = \lfloor \hat{d}_i \rfloor\) and \(\lceil
  \hat{d}_i \rceil\), since the multiplications at other levels are free due to
  a left-hand operand of the form \([\begin{smallmatrix} x^j & 0 \\ 0 & 1
    \end{smallmatrix}]\) or \([\begin{smallmatrix} 1 & 0 \\ 0 & x^j
  \end{smallmatrix}]\).

  According to \cref{lem:appbas_cx:multiplication}, the multiplication which
  gathers the outputs of the recursive calls at orders \(d_i = \lfloor
  \hat{d}_i \rfloor\) and \(\lceil \hat{d}_i \rceil\) costs \(4 \timepm{d_i} +
  \bigO{d_i}\) operations. Since \(d_1+\cdots+d_r \le \dsyz+1\), by
  superlinearity of \(\timepm{\cdot}\), altogether the bases multiplications
  cost \(4 \timepm{\dsyz} + \bigO{\dsyz}\) operations.

  For the residuals, we consider the worst case where they all arise at the
  largest possible orders \(\lceil d / 2^i \rceil\) for \(0 \le i \le k\),
  where \(k = \lceil \log_2(d) - 2 \rceil\). According to
  \cref{lem:appbas_cx:residual}, computing them costs \((\sum_{i=0}^{k}
  {\textstyle\frac{3}{2}}\timepm{\lceil d/2^i \rceil}) + \bigO{d}\) operations
  in \(\field\). Isolating the term \(i=0\) of the sum, which is
  \({\textstyle\frac{3}{2}}\timepm{d}\), the rest of the sum is bounded by
  \(\sum_{i=1}^{k} {\textstyle\frac{3}{2}}\timepm{\lceil d/2^i \rceil} \le
  {\textstyle\frac{3}{2}} \timepm{d + k}\), thanks to the superlinearity of
  \(\timepm{\cdot}\) and the inequality \(\sum_{i=1}^{k} \lceil d/2^i
  \rceil \le d + k\). Since \(\timepm{d + \lceil \log_2(d) - 2 \rceil}\)
  is in \(\timepm{d} + \bigO{d \log(d)}\), we conclude that the computation of
  all residuals is in \({\textstyle\frac{5}{2}} \timepm{d} + \bigO{d
  \log(d)}\).
  \qed
\end{proof}

\begin{remark}
  Refinements of the above analysis should be feasible when it comes to
  improving the constant factors in front of the terms in \(\timepm{\dsyz}\)
  and \(\timepm{d}\). The above proof, concerning the extra computations for
  residuals and bases multiplications (outside of the ``normal'' calls), does
  not exploit the fact that the manipulated bases have a special form. Indeed,
  after the overall order \(\dsyz\) has been exceeded, all bases
  multiplications involve a left-hand operand of the form either
  \([\begin{smallmatrix} x^j & 0 \\ 0 & 1
    \end{smallmatrix}]\) or \([\begin{smallmatrix} 1 & 0 \\ 0 & x^j
  \end{smallmatrix}]\),
  and the computed bases have gradually higher valuation in either the first or
  the second row, accordingly. In the generic case, this allows one to compute
  the bases multiplications in \(\cstmatmul \timepm{\lceil\dsyz/2 \rceil}\) instead of \(4
  \timepm{\dsyz}\). It should also be possible to exploit this valuation to
  give a slightly better complexity bound for the computation of residuals,
  thus improving the constant for the term in \(\timepm{d}\).
  \qed
\end{remark}

\subsubsection{Refinement: non-uniform shifts}
\label{sec:approx:appbasis:unishift}

If the input shift has small amplitude, one expects the same for recursive calls.
For example, as we have seen in \cref{sec:approx:pmbasis:cx_generic}, if one
starts with a shift of amplitude at most \(1\), then generically all recursive
calls also involve a shift of amplitude at most \(1\); this is a favorable case
for efficiency, since one has better control of the degrees of the entries of
the computed bases. Here, if the input shift is not uniform, we use
\cref{lem:appbas:opti:shift} to design a preliminary step which computes a
first approximant basis at order \(\abs{\ash}\). Then, we are
essentially left with an instance with a uniform shift. In particular, this
allows us to give a complexity bound for the generic case when the input shift
is arbitrary, whereas the above results assumed shift amplitude at most \(1\).
Adding this preliminary step yields \algoName{algo:AppBasis2}.

\begin{algorithm}[ht]
  \algoCaptionLabel{AppBasis2}{d, \apo, \bpo, \ash}
  \begin{algorithmic}[1]

    \Require \(d \in \mathbb{N}\), polynomials \((\apo,\bpo) \in \pR_{<d}^2\), shift \(\ash \in \mathbb{Z}\)

    \Ensure a basis \(P \in \mpR{2}{2}_{\le d}\) of \(\appmod{d}{\apo,\bpo}\) in \(\ash\)-weak Popov form

    \State\InlineIf{\(d \le 2\)}{\Return \(\Call{algo:AppBasis2-base}{d, \apo,\bpo, \ash}\)}

    %% optimisation: \apo or \bpo has positive valuation
    \LComment{\(\apo\) or \(\bpo\) has positive valuation: reduce to lower order}
    \State \(v_\apo \gets \min(d,\valuation{\apo})\); \(v_\bpo \gets \min(d,\valuation{\bpo})\)
    \State \InlineIf{\(v_\apo = v_\bpo = d\)}{\Return the \(2 \times 2\) identity matrix}
            \MyComment{\(\apo = \bpo = 0\)}
    \If{\(v_\apo > 0\) \OR{} \(v_\bpo > 0\)}
      \If{\(v_\apo > v_\bpo\) \OR{} \(v_\apo = v_\bpo > 0\)}
            \State \(\basis \gets [\begin{smallmatrix} 1 & 0 \\ 0 & x^{v_{\apo}-v_{\bpo}} \end{smallmatrix}]\);
                   \(\drecb \gets d-v_{\apo}\);
                   \(\apor \gets x^{-v_{\apo}}\apo\);
                   \(\bpor \gets (x^{-v_\bpo}\bpo) \rem x^{\drecb}\);
                   \(\ashr \gets \ash+v_{\bpo}-v_{\apo}\)
      \ElsIf{\(v_\bpo > v_\apo\)}
            \State \(\basis \gets [\begin{smallmatrix} x^{v_{\bpo}-v_{\apo}} & 0 \\ 0 & 1 \end{smallmatrix}]\);
                   \(\drecb \gets d-v_{\bpo}\);
                   \(\apor \gets (x^{-v_{\apo}}\apo) \rem x^{\drecb}\);
                   \(\bpor \gets x^{-v_\bpo}\bpo\);
                   \(\ashr \gets \ash+v_{\bpo}-v_{\apo}\)
      \EndIf
      \State \(\basisr \gets \Call{algo:AppBasis2}{\drecb, \apor, \bpor, \ashr}\)
      \State \Return \(\basisr \basis\)
            \MyComment{if going beyond this point, \(\valuation{\apo} = \valuation{\bpo} = 0\)}
    \EndIf

    %% optimisation: Hermite shift
    \LComment{the shift has amplitude \(\ge d\): return the basis in lower or upper HNF}
    \State\InlineIf{\(\ash \le -d\)}{\Return \([\begin{smallmatrix} x^{d} & 0 \\ -c & 1 \end{smallmatrix}]\), where \(c \gets (\bpo \apo^{-1}) \rem x^{d}\)}
    \State\InlineIf{\(\ash \ge d\)}{\Return \([\begin{smallmatrix} 1 & -c \\ 0 & x^{d} \end{smallmatrix}]\), where \(c \gets (\apo \bpo^{-1}) \rem x^{d}\)}

    %% shift already uniform: call recursive approach directly
    \LComment{the shift is uniform: direct call to the divide and conquer approach}
    \State\InlineIf{\(\ash = 0\)}{\Return \(\Call{algo:AppBasis2-rec}{d, \apo, \bpo, 0}\)}
        \MyComment{shift is already uniform}

    %% optimisation: unbalanced shift
    \LComment{the shift has amplitude \(> 0\): preliminary computation to make it uniform}
        \label{step:appbas:unbalancedshift:start}
    \If{\(\ash < 0\)} \MyComment{\(0 < - \ash < d\)}
      \State \(\dreca \gets -\ash\) ; \(\drecb \gets d - \dreca\);
             \(\basis \gets [\begin{smallmatrix} x^{\dreca} & 0 \\ -c & 1 \end{smallmatrix}]\),
             where \(c \gets (\bpo \apo^{-1}) \rem x^{\dreca}\);
             \(\apor \gets \apo \rem x^{\drecb}\)
      \State \(\bpor \gets \sum_{0 \le k < \drecb} \bpo_{\dreca+k}\, x^{k} - \midprod{c, \apo, \dreca, d}\)
            \MyComment{residual \(\bpor = (x^{-\dreca} (\bpo - c \apo)) \rem x^{\drecb}\)}
            \label{step:appbas:unbalancedshift:resa}
    \Else \MyComment{\(0 < \ash < d\)}
      \State \(\dreca \gets \ash\) ; \(\drecb \gets d - \dreca\);
             \(\basis \gets [\begin{smallmatrix} 1 & -c \\ 0 & x^{\dreca} \end{smallmatrix}]\),
            where \(c = (a b^{-1}) \rem x^{\dreca}\);
            \(\bpor \gets \bpo \rem x^{\drecb}\)
      \State \(\apor \gets \sum_{0 \le k < \drecb} \apo_{\dreca+k}\, x^{k} - \midprod{c, \bpo, \dreca, d}\)
            \MyComment{residual \(\apor = (x^{-\dreca} (\apo - c\bpo)) \rem x^{\drecb}\)}%
            \label{step:appbas:unbalancedshift:resb}
    \EndIf
    \State \(\basisr \gets \Call{algo:AppBasis2-rec}{\drecb, \apor, \bpor, 0}\)
           \MyComment{divide and conquer with uniform input shift}
        \label{step:appbas:unbalancedshift:call}
    \State \Return \(\basisr \basis\)
        \label{step:appbas:unbalancedshift:end}
  \end{algorithmic}
\end{algorithm}

\begin{proposition}
  \label{prop:appbasis:shiftreduction}%
  Let \(d \in \mathbb{N}\), \((\apo,\bpo) \in \pR_{<d}^2\), and \(\ash \in
  \mathbb{Z}^2\). Let \(\dreca = \min(d,\abs{\ash})\) be the
  minimum of the approximation order \(d\) and of the amplitude of
  \(\ash\).  Then \cref{algo:AppBasis2} computes an \(\ash\)-weak
  Popov basis of \(\appmod{d}{\apo,\bpo}\) using
  \begin{itemize}[noitemsep,topsep=2pt]
    \item one call to \algoName{algo:AppBasis2-rec} at order \(d-\dreca\) and
      input shift \(0\);
    \item one division of power series, one residual computation, and one basis
      multiplication at a total cost of
      \(\frac{5}{2} \timepm{\dreca} + \timepm{\lfloor d/2 \rfloor} + \timepm{\dreca + \lfloor \frac{d-\dreca}{2} \rfloor} + \bigO{d} \in \bigO{\timepm{d}}\)
      operations in \(\field\).
  \end{itemize}
\end{proposition}

\begin{proof}
  The correctness of the algorithm follows from the results in
  \cref{sec:approx:appbasis:optimizations}. Note that, if the quantity
  \(\dreca\) from the statement is in \(\{1,\ldots,d-1\}\), then it coincides
  with the quantity denoted by \(\dreca\) at
  \crefrange{step:appbas:unbalancedshift:start}{step:appbas:unbalancedshift:end}
  of the algorithm.  In this case, the fact that the call to
  \algoName{algo:AppBasis2-rec} at \cref{step:appbas:unbalancedshift:call}
  is with the uniform shift \(0\) is
  justified by the shape of the basis \(\basis\). For example, in the case \(0
  < - \ash < d\), then \(\dreca = -\ash\) and \(\basis =
  [\begin{smallmatrix} x^{\dreca} & 0 \\ -c & 1
  \end{smallmatrix}]\), so that the updated shift in the recursive approach
  (see \cref{lem:dnc_relbas}) is \(\ashr = \ash+\deg(x^{\dreca})-
  \deg(1) = 0\). The case \(0 < \ash < d\) is similar.
  Apart from this call, one has to compute \(c\) which is a division of power
  series at precision \(\dreca\) and costs \(\frac{5}{2} \timepm{\dreca} +
  \bigO{\dreca}\) operations; to compute the residual at
  \cref{step:appbas:unbalancedshift:resa} or
  \cref{step:appbas:unbalancedshift:resb}, at a cost of \(\timemp{\dreca,\drecb} +
  \bigO{\drecb} \subseteq \timepm{\lfloor d/2 \rfloor} + \bigO{d}\) operations
  thanks to \hyptime{hyp:timepm:midprod};
  and to compute the basis
  multiplication \(\basisr \basis\) at \cref{step:appbas:unbalancedshift:end}.
  The latter costs \(\bigO{\timepm{d}}\) operations, and one may exploit the
  specific shape of \(\basis\) to refine the constant in this
  \(\bigO{\cdot}\): for example in the case \(0 < - \ash < d\) one has
  \[
    \basisr \basis =
    \begin{bmatrix}
      q_{00} & q_{01} \\
      q_{10} & q_{11}
    \end{bmatrix}
    \begin{bmatrix}
      x^{\dreca} & 0 \\
      -c & 1
    \end{bmatrix}
    =
    \begin{bmatrix}
      q_{00} x^{\dreca} - c q_{01} & q_{01} \\
      q_{10} x^{\dreca} - c q_{11} & q_{11}
    \end{bmatrix}.
  \]
  Thus, \(\basisr \basis\) can be computed in \(\timepm{\lfloor \frac{\dreca +
  \deg(q_{01})}{2} \rfloor} + \timepm{\lfloor \frac{\dreca + \deg(q_{11})}{2}
\rfloor} + \bigO{d}\) operations, using \hyptime{hyp:timepm:sumdeg}. Since
\(\basisr\) is a weak Popov approximant basis at order \(\drecb\), one has
\(\deg(q_{00}) + \deg(q_{11}) \le \drecb\) (see
\cref{sec:approx:pmbasis:cx_worst_case}) as well as \(\deg(q_{01}) <
\deg(q_{00})\) (since the shift is uniform), yielding the upper bound
\(\timepm{\dreca + \lfloor \frac{\drecb}{2} \rfloor} + \bigO{d}\) by
superlinearity.
\qed
\end{proof}

\section{Minimal interpolant bases: Cauchy interpolation}
\label{sec:interp}

In this section, we focus on the fast computation of interpolant bases
(\cref{pbm:int}), that is, bases of \(\intmod{\points}{\apo,\bpo}\) for special
sets of points \(\points \in \field^d\).

In this section we use bold letters to denote tuples of field elements,
typically for the points at which we evaluate or interpolate, or the values of
polynomials at these points. For such a tuple \(\points \in
\field^d\), we use the slice notation for subtuples: \(\points_{i:j} =
(\point_i,\point_{i+1},\ldots,\point_{j-1})\).
Given tuples \(\aevs = (\aev_i)_{0 \le i < d}\) and \(\bevs = (\bev_i)_{0 \le i
< d}\) in \(\field^d\), we overload our notation and define
\[
  \intmod{\points}{\aevs,\bevs}
  = \{(p,q) \in \pR^2 \mid \aev_i \pp(\point_i) + \bev_i \pq(\point_i) = 0 \text{ for } 0 \le i < d\}.
\]
This is convenient because we will often consider polynomials represented by
their evaluations, and one easily observes that \(\intmod{\points}{\aevs,\bevs}
= \intmod{\points}{\apo,\bpo}\) holds when the tuples above are the evaluations
\(\aevs = (\apo(\point_0),\ldots,\apo(\point_{d-1}))\) and \(\bevs =
(\bpo(\point_0),\ldots,\bpo(\point_{d-1}))\). Note that switching between
polynomial input \((\apo,\bpo)\) and evaluated input \((\aevs,\bevs)\) only
costs \(\bigO{\timepm{d}}\) operations in \(\field\) in the case of points
\(\points\) in geometric progression, including radix-2 FFT points; see
\cref{sec:interp:eval_interp_extrap}.

For the base case of recursions, we give ourselves an algorithm with the
following specification and running time \(\bigO{1}\). In practice, one would
typically use an iterative algorithm such as those from
\cite{m_pade,BeckermannLabahn1994}, and would use this one for all base cases
with \(d\) not exceeding some predetermined (constant) threshold.

\begin{algorithm}[h]
  \algoCaptionLabel{IntBasis2-base}{\points, \aevs, \bevs, \ash}
  \begin{algorithmic}[1]
    \Require distinct points \(\points = (\point_i)_{0 \le i < d}\) for some \(d \in \{0,1,2\}\),
    values \((\aevs,\bevs) \in (\field^d)^2\),
    shift \(\ash \in \mathbb{Z}\)
    \Ensure the basis \(\basis \in \mpR{2}{2}\) of
    \(\intmod{\points}{\aevs,\bevs}\) in \(\ash\)-Popov form
  \end{algorithmic}
\end{algorithm}

We start by stating some results about evaluation, interpolation, and
extrapolation at special sets of points (\cref{sec:interp:eval_interp_extrap}).
Then, as outlined in \cref{sec:intro:approx_interp}, we present and analyze the
classical divide and conquer interpolant basis algorithm
(\cref{sec:interp:pmintbasis}), and describe an optimized variant of it based
on an evaluated representation (\cref{sec:interp:eval_intbasis}). Finally, we
give the sketch of an approach to further optimize the latter algorithm
(\cref{sec:interp:eval_intbasis_opti}), but leaving as a perspective the study
of the conditions under which it is applicable.

Our analyses are restricted to some sets of points, for which it is known how
to perform \(d\)-point evaluation or \(d\)-point extrapolation in
\(\bigO{\timepm{d}}\). In particular, this excludes general points, for which
the best known cost bounds are in \(\bigO{\timepm{d}\log(d)}\) and the
algorithms in this section would lead to a complexity of \(\bigO{\timepm{d}
\log(d)^2}\); for this reason, one should rather rely on the reductions from
\cref{sec:rel_xgcd} for such points. Still, these reductions require \(\bpo =
-1\), whereas \cref{pbm:int} involves an arbitrary polynomial \(\bpo\). To
circumvent this, similarly to the approximant case in
\cref{sec:rel_xgcd:app_to_modN}, one needs first to process the input
to ensure either \(\apo=-1\) or \(\bpo=-1\); for example, in the
favorable case where all entries of \(\bevs\) are nonzero, one has
\(\intmod{\points}{\apo, \bpo} = \intmod{\points}{-\apo\bpo^{-1}, -1}\).

\begin{remark}
  \label{rmk:interp_optimizations}%
  One could derive optimizations analogous to those in
  \cref{sec:approx:appbasis}, so as to improve the sensitivity of the
  interpolant basis algorithms presented below to the shift \(\ash\), to the
  number of zero evaluations in  \(\aevs\) and \(\bevs\), and also to the
  degree of the GCD of \(\apo\) and \(\bpo\). One would then need to carefully
  analyze these optimizations, notably to take into account the fact that the
  powers of \(x\) that appear in
  \cref{lem:appbas:opti:valuation,lem:appbas:opti:shift} are replaced here by
  polynomials of the form \(\prod_{i \le k < j} (x - \point_k)\),
  and multiplying by those is not free anymore. We leave this as a
  perspective.
  \qed
\end{remark}

\subsection{Evaluation \& Interpolation \& Extrapolation at special sets of points}
\label{sec:interp:eval_interp_extrap}

In this section, we first recall three basic operations about polynomials and
their values at points: evaluation, interpolation, and extrapolation. Then, we
state results that are either known or easily deduced from the literature
concerning the complexity of performing these operations for special sets of
points, focusing on those that we will need in our interpolant basis
algorithms.

\begin{definition}[evaluation]
  \label{dfn:evaluate}%
  Let \(\points = (\point_0, \ldots, \point_{\dout-1})\in \field^\dout\).
  The evaluation of a polynomial \(p \in \pR\) at \(\points\) is the tuple of
  values \(\evaluate(\points, p) = (p(\point_0), \ldots,
  p(\point_{\dout-1}))\in \field^\dout\), denoted by \(p(\points)\).
  This extends to a polynomial matrix \(\basis \in \mpR{2}{2}\) by
  evaluating each entry:
  \[
    \evaluate(\points, \basis) = (\basis(\point_0), \ldots, \basis(\point_{\dout-1}))
    = \basis(\points) \in (\mKK{2}{2})^\dout.
  \]
\end{definition}

\begin{definition}[interpolation]
  \label{dfn:interpolate}%
  Let \(\points = (\point_0, \ldots, \point_{\dinp-1})\) be
  pairwise distinct elements of \(\field\). For given values \(\bm{p} = (p_0,
  \ldots, p_{\dinp-1}) \in \field^\dinp\), \(\interpolate(\points,
  \bm{p})\) is the unique polynomial \(p \in \pR_{<\dinp}\) such that
  \(p(\points) = \bm{p}\). This extends to matrix values \(\evbasis \in
  (\mKK{2}{2})^\dinp\), with \(\interpolate(\points, \evbasis)\) being the unique
  polynomial matrix \(\basis \in \mpR{2}{2}_{<\dinp}\) such that
  \(\basis(\points) = \evbasis\).
\end{definition}

\begin{definition}[extrapolation]
  \label{dfn:extrapolate}%
  For \(\dinp, \dout \in \NN\), pairwise distinct points \(\ptsinp \in
  \field^\dinp\), values \(\bm{p} \in \field^\dinp\), and points \(\ptsout \in
  \field^\dout\), the extrapolation operation is
  \begin{align*}
    \extrapolate(\ptsinp, \ptsout, \bm{p}) & = \evaluate(\ptsout, \interpolate(\ptsinp, \bm{p})) \\
                                           & = p(\ptsout) \in \field^\dout
  \end{align*}
  where \(p \in \pR_{<\dinp}\) is the unique polynomial such that \(p(\ptsinp)
  = \bm{p}\).
  This extends to matrices: for the same points \(\ptsinp\) and \(\ptsout\),
  and for values \(\evbasis = (\basis_i)_{0 \le i < \dinp} \in (\mKK{2}{2})^\dinp\),
  \begin{align*}
    \extrapolate(\ptsinp, \ptsout, \evbasis) & = \evaluate(\ptsout, \interpolate(\ptsinp, \evbasis)) \\
                                             & = \basis(\ptsout) \in (\mKK{2}{2})^\dout
  \end{align*}
  where \(\basis \in \mpR{2}{2}_{<\dinp}\) is the unique polynomial matrix
  such that \(\basis(\ptsinp) = \evbasis\).
\end{definition}

Thanks to the reductions presented in \cref{sec:intro:rel_xgcd,sec:rel_xgcd},
one can solve \cref{pbm:ratrecon,pbm:xgcd,pbm:app} by computing an interpolant
basis with input points \(\points\) that we can choose.
For this reason, it is sensible to consider special sets of points, for which
some of the above operations can be performed efficiently, say in
\(\bigO{\timepm{\dinp+\dout}}\). We focus on three families: points in
arithmetic progression, points in geometric
progression, and radix-2 FFT points.

\begin{lemma}[Arithmetic progression]
  \label{lem:arithmetic_progression}%
  Assume the field \(\field\) has characteristic \(0\) or at least \(d\),
  and consider points in arithmetic progression \(\points = (\point_0 + i
  \ratio)_{0 \le i < d}\) for a given \(r \in \field \setminus \{0\}\).
  Given nonnegative integers \(k, \ell, \dinp, \dout\) such that either
  \(k+\dinp \le \ell\) and \(\ell+\dout \le d\) hold, or \(\ell+\dout \le k\)
  and \(k+\dinp \le d\) hold, and given a tuple of values \(\bm{p} \in
  \field^\dinp\), one can compute \(\extrapolate(\points_{k:k+\dinp},
  \points_{\ell:\ell+\dout}, \bm{p})\) using
  \(\timemp{\dinp-1,\dout} + \bigO{\dinp+\dout}\) operations in \(\field\).
  %% NOTE in the end, we do not use this: this could be mentioned
  %% to better understand the leading constant in the complexity for
  %% the arithmetic progression case (either calling directly algo:Eval-IntBasis2
  %% with such points, or through a reduction to geometric/FFT points),
  %% but this is too much trouble. In the reference, complexity is analyzed for
  %% a number of points that is a power of 2. Claiming that it always holds seems
  %% risky, without proof, and a proof would bring us too far.
  % \begin{itemize}[noitemsep,topsep=2pt]
  % \item Given a nonnegative integer \(\dout \le d\) and a polynomial \(p \in
  %   \pR_{< \dout}\), \(\evaluate(\points_{0:\dout}, p)\) is computed in
  %   \(\frac{1}{4} \timepm{\dout} \log_2(\dout) + \bigO{\timepm{\dout}}\) operations in
  %   \(\field\), assuming the subproduct tree for \(\points_{0:\dout}\) has
  %   already been computed.
  % \item Given nonnegative integers \(\ell,\dout\) with \(\ell+\dout \le d\)
  %   and given a tuple of values \(\bm{p} \in \field^\dout\),
  %   \(\interpolate(\points_{\ell:\ell+\dout}, \bm{p})\) is computed in
  %   \(\frac{1}{4} \timepm{\dout} \log_2(\dout) + \bigO{\timepm{\dout}}\)
  %   operations in \(\field\), assuming the subproduct tree for
  %   \(\points_{0:\dout}\) has already been computed
  % \end{itemize}
  %% NOTE [Newton basis not used in the end]
  % \item Evaluate (Newton basis) arithmetic: mullow of a polynomial of degree
  %   \(\le \dinp-1\) and a polynomial of degree \(< dout\)
  %   \cite[Sec.\,3]{Gerhard2000}
  %   \cite[Sec.\,4.1]{BostanSchost2005}
  % \item Interpolate (Newton basis) arithmetic: \cite[Sec.\,4.1]{BostanSchost2005}
\end{lemma}
\begin{proof}
  %% NOTE draft notes about part that we do not use
  % -> Evaluate from Newton basis is \(\timepm{\dout} + \bigO{\dout}\)
  % \cite[Sec.\,4.1]{BostanSchost2005}.
  % -> Subproduct tree is constant \(1/4\).
  % -> \cite[Thm.\,1]{BostanSchost2005} gives \(1/2 \timepm{\dout}\log(\dout) +
  % \timepm{\dout}\) for conversion between bases (in degree \(\dout\)), with our
  % analysis in \cref{app:dnc_factor_half} this is  \(1/4
  % \timepm{\dout}\log(\dout) + \timepm{\dout}\).
  % ->  Evaluate and interpolate arithmetic only used at the very beginning and
  %   very end. A rough estimate is the one in the table from the paper, but that
  %   might be improved in our case. Interpolate arithmetic: interpolation to
  %   Newton basis \cite[Sec.\,4.1]{BostanSchost2005} in \(\timepm{\dout} +
  %   \bigO{\dout}\), and then change of basis has a log but that second step depends
  %   on that output degree! note also that the subproduct tree may be
  %   computed once for all both for the initial evaluation of \(\apo,\bpo\) and
  %   for the final interpolation of a \(2 \times 2\) matrix.
  We use the algorithm in \cite[Sec.\,3]{BostanGaudrySchost2007}, which we
  sketch here for completeness, adapted to our notation and straightforwardly
  extended to the case where the numbers \(\dinp\) and \(\dout\) of input and
  output values may differ.
  % NOTE (not really needed in the end, since we see below where something needs to be invertible)
  % Our assumptions ensure that, for all \(i\neq j\)
  % both in \(\{k, k+1, \ldots, k+\dinp-1\} \cup \{\ell, \ell+1, \ldots, \ell +
  % \dout - 1\}\), the pairwise difference \(\point_i - \point_j =
  % (i-j) r\) is nonzero as an element of \(\field\). This meets the
  % invertibility requirements in \cite[Thm.\,5]{BostanGaudrySchost2007}, and
  % gives meaning to the fractions written below.
  Let \(p = \interpolate(\points_{k:k+\dinp}, \bm{p}) \in \pR_{<\dinp}\), so that we aim
  to compute \(p(\points_{\ell:\ell+\dout})\). By the Lagrange formula,
  \begin{align*}
    p = \sum_{0 \le i < \dinp} \bm{p}_i \prod_{0 \le j < \dinp, j \neq i}
    \frac{x - \point_{k+j}}{\point_{k+i} - \point_{k+j}}
    & = \sum_{0 \le i < \dinp} \bm{p}_i \prod_{0 \le j < \dinp, j \neq i} \frac{x - \point_{k+j}}{(i-j) \ratio} \\
    & = \sum_{0 \le i < \dinp} \frac{\bm{q}_i}{\ratio^{\dinp-1}} \prod_{0 \le j < \dinp, j \neq i} (x - \point_{k+j}),
  \end{align*}
  where \(\bm{q}\) is the appropriate scaling of \(\bm{p}\), obtained in
  \(\bigO{\dinp}\). It follows that, for all \(0 \le u < \dout\),
  \[
    p(\point_{\ell+u}) = \sum_{0 \le i < \dinp} \frac{\bm{q}_i}{\ratio^{\dinp-1}} \prod_{0 \le j < \dinp, j \neq i} (\point_{\ell+u} - \point_{k+j})
    % = \sum_{0 \le i < \dinp} \bm{q}_i \prod_{0 \le j < \dinp, j \neq i} (\ell + u - k - j)
    = \sigma_u \sum_{0 \le i < \dinp} \frac{\bm{q}_i}{\ell + u - k - i},
  \]
  where \(\sigma_{u} = \prod_{0 \le j < \dinp} (\ell + u - k - j)\).
  Note that our assumptions ensure that, as integers, \(\ell+u \neq k+i\) and
  \(-d < \ell + u - k - i < d\); hence \(\ell + u - k - i\) is invertible as an
  element of \(\field\). All \(\sigma_u\)'s, for \(0 \le u < \dout\), are
  computed using a total of \(\bigO{\dinp+\dout}\) operations in \(\field\). Up
  to these scaling factors, it remains to compute a middle product:
  \begin{align*}
    & \sum_{0 \le u < \dout} \frac{p(\point_{\ell+u})}{\sigma_u} x^u
    = \\
    & \quad \midprodBig{\sum_{0 \le i < \dinp} \bm{q}_{i} x^i, \sum_{0 \le i < \dinp-1+\dout} \frac{x^i}{\ell-k+i-\dinp+1}, \dinp-1, \dinp - 1 + \dout},
  \end{align*}
  which costs \(\timemp{\dinp-1, \dout}\) operations in \(\field\) by
  definition (see \cref{sec:approx:polmul_midprod}).
  \qed
\end{proof}

\begin{lemma}[Geometric progression]
  \label{lem:geometric_progression}%
  Assume that an element \(\ratio \in \field\) of order at least \(d\) exists
  and is known, and consider points in geometric progression \(\points =
  (\point_0 \ratio^i)_{0 \le i < d}\), for some given \(\point_0 \in
  \field\setminus\{0\}\). The three items below assume that \(\ratio^k\),
  \(\ratio^\ell\), and \(\ratio^{\ell-k}\) are known, respectively; otherwise,
  they can be computed efficiently by binary exponentiation.
  \begin{itemize}[noitemsep,topsep=2pt]
    \item Given nonnegative integers \(k,\dinp\) with \(k+\dinp \le d\)
      and given a tuple of values \(\bm{p} \in \field^\dinp\),
      one can compute
      \(\interpolate(\points_{k:k+\dinp}, \bm{p})\) using
      \(2\timepm{\dinp} + \bigO{\dinp}\) operations in \(\field\).

    \item Given nonnegative integers \(\ell,\dout\) with \(\ell+\dout \le d\) and
      given a polynomial \(p \in \pR_{< \dinp}\),
      \(\evaluate(\points_{\ell:\ell+\dout}, p)\) is computed in \(\timemp{\dinp-1,
      \dout} + \bigO{\dinp+\dout}\) operations in \(\field\).
      % with \(\dinp \in \bigO{\dout}\),

    \item Given nonnegative integers \(k, \ell, \dinp, \dout\) such that either
      \(k+\dinp \le \ell\) and \(\ell+\dout \le d\) hold, or \(\ell+\dout \le
      k\) and \(k+\dinp \le d\) hold, and given values \(\bm{p} \in
      \field^\dinp\), \(\extrapolate(\points_{k:k+\dinp},
      \points_{\ell:\ell+\dout}, \bm{p})\) can be computed using
      \(\timemp{\dinp-1,\dout} + \bigO{\dinp+\dout}\) operations in \(\field\).
  \end{itemize}
  %% NOTE [Newton basis not used in the end]
  % \item Evaluate (Newton basis) geometric: mullow of a polynomial of degree
  %   \(\le \dinp-1\) and a polynomial of degree \(< dout\) \cite[Sec.\,5.1]{BostanSchost2005}
  % \item Interpolate (Newton basis) geometric: \cite[Sec.\,5.1]{BostanSchost2005}
\end{lemma}
\begin{proof}
  For interpolation, using \(\point_k = \point_0 \ratio^k\) brings us to the
  classical case of the initial points
  \(\points_{k:k+\dinp} = (\point_k, \point_k \ratio, \point_k \ratio^2,
  \ldots, \point_k \ratio^{\dinp-1})\) of a geometric progression. The cost of
  interpolating \(\bm{p}\) is then stated in
  \cite[Sec.\,5.3]{BostanSchost2005}. The same reference presents the
  evaluation of \(p\) at \(\points_{\ell:\ell+\dout} = (\point_\ell,
  \point_\ell \ratio, \point_\ell \ratio^2, \ldots, \point_\ell
  \ratio^{\dout-1})\) using the chirp transform
  \cite{Bluestein1970,AhoSteiglitzUllman1975}.
  Its main step is a middle
  product of the form \(\midprod{\bar{p}, q, \dinp-1, \dinp-1+\dout}\) which
  costs \(\timemp{\dinp-1, \dout}\) operations in \(\field\), where
  \(\bar{p}\in\pR_{<\dinp}\) and \(q \in \pR_{<\dinp-1+\dout}\) are built in \(\bigO{\dinp+\dout}\)
  from the input \(p\) and \(\point_\ell\) and \(\ratio\).

  For the extrapolation, we propose an approach very similar to that for
  arithmetic progressions in \cite[Sec.\,3]{BostanGaudrySchost2007} and
  sketched in \cref{lem:arithmetic_progression}. Let \(p =
  \interpolate(\points_{k:k+\dinp}, \bm{p}) \in \pR_{<\dinp}\), so that we aim
  to compute \(p(\points_{\ell:\ell+\dout})\). By the Lagrange formula,
  for \(0 \le u < \dout\),
  \begin{align*}
    p(\point_{\ell+u}) = \sum_{0 \le i < \dinp} \bm{p}_i \prod_{0 \le j < \dinp, j \neq i} \frac{\point_{\ell+u} - \point_{k+j}}{\point_{k+i} - \point_{k+j}}
    & = \sum_{0 \le i < \dinp} \bm{p}_i \prod_{0 \le j < \dinp, j \neq i} \frac{\ratio^{\ell+u-k-j} - 1}{\ratio^{i-j} - 1} \\
    & = \sigma_u \sum_{0 \le i < \dinp} \frac{\bm{q}_i}{\ratio^{\ell+u-k-i} - 1},
  \end{align*}
  where we have used the scaling \(\bm{q} \in \field^\dinp\) of \(\bm{p}\) defined as
  \(\bm{q}_i = \bm{p}_i / \prod_{0 \le j < \dinp, j \neq i} (\ratio^{i-j} -
  1)\), as well as the scaling factors \(\sigma_u = \prod_{0 \le j < \dinp}
  (\ratio^{\ell+u-k-j} - 1)\).

  Observe that \(\bm{q}\) is obtained from \(\bm{p}\) in linear time
  \(\bigO{\dinp}\), and all \(\sigma_u\)'s, for \(0 \le u < \dout\), are
  computed using a total of \(\bigO{\dinp+\dout}\) operations in \(\field\)
  (plus the computation of \(\ratio^{\ell-k}\), in case it is not already
  known). Note also that the above inverse of \(\ratio^{\ell+u-k-i} -
  1\) exists, since by assumption \(\ell+u\) and \(k+i\) are distinct and
  both in \(\{0,\ldots,d-1\}\), and \(\ratio\) has order at least \(d\).
  Apart from linear-time scalings, the main task
  to obtain the sought extrapolated values
  is to compute a middle product:
  \begin{align*}
    & \sum_{0 \le u < \dout} \frac{p(\point_{\ell+u})}{\sigma_u} x^u
    = \\
    & \quad \midprodBig{\sum_{0 \le i < \dinp} \bm{q}_{i} x^i, \sum_{0 \le i < \dinp-1+\dout} \frac{x^i}{\ratio^{\ell-k+i-\dinp+1} - 1}, \dinp-1, \dinp - 1 + \dout},
  \end{align*}
  which costs \(\timemp{\dinp-1, \dout}\) operations in \(\field\) by
  definition (see \cref{sec:approx:polmul_midprod}).
  \qed
\end{proof}

Finally, we turn to points for the radix-2 FFT: for a given \(N\)-root of unity
\(\ratio\) with \(N\) power of \(2\), we consider the first \(N\) powers of
\(\ratio\). As is classical in this context, we do not consider them in
geometric progression order \((\ratio^i)_{0 \le i < N}\), but rather with
exponents ordered by increasing bit-reversed index \((\ratio^{\bitrev{i}})_{0
\le i < N}\), where \(\bitrev{\cdot} : \{0,\ldots,N-1\} \to \{0,\ldots,N-1\}\)
is the bit reversal permutation. For example, if \(N \ge 8\), we have
\((\bitrev{i})_{0 \le i < N} = (0, N/2, N/4, 3N/4, N/8, 5N/8, 3N/8, 7N/8,
\ldots)\); the rest of the sequence can be determined by the fact that
\((\bitrev{i})_{k \le i < 2k} = (\bitrev{i} + N/(2k))_{0 \le i < k}\) for any
\(k\) which is a power of \(2\) with \(2k \le N\).

\begin{lemma}[radix-2 FFT points]
  \label{lem:radix2_fft}%
  Assume that \(N\) is a power of \(2\) and that an \(N\)-th principal
  root of unity \(\ratio \in \field\) exists and is known,
  and consider the FFT points \(\points = (\ratio^{\bitrev{i}})_{0 \le i < N}
  \in \field^N\).
  The items below assume that \(\point_{jk} = \ratio^{\bitrev{jk}}\)
  is known, as well as some points of the form \((\ratio^{N / 2^i})_{1 \le i
  \le \log_2(N)}\); otherwise, these points can all be computed using
  \(\bigO{\log_2(N)}\) operations in \(\field\).
  \begin{itemize}[noitemsep,topsep=2pt]
    %% this is correct but only valid for starting point at power of 2,
    %% which is not flexible enough for our needs
    % \item Given nonnegative integers \(k,\dinp\) such that either \(k = 0\) and
    %   \(\dinp \le N\), or \(k\) is a power of \(2\) with \(\dinp \le k < N\),
    %   and given a tuple of values \(\bm{p} \in \field^\dinp\), one can compute
    %   \(\interpolate(\points_{k:k+\dinp}, \bm{p})\) using
    %   \(\frac{3}{2} \dinp \log_2(\dinp) + \bigO{\dinp}\) operations in
    %   \(\field\).
    \item Given nonnegative integers \(j,k,\dinp\) such that
      \(0 \le jk < N\) and \(k > 0\) is a power of \(2\) with \(\dinp \le k\),
      and given a tuple of values \(\bm{p} \in \field^\dinp\), one can compute
      \(\interpolate(\points_{jk:jk+\dinp}, \bm{p})\) using
      \(\frac{3}{2} \dinp \log_2(\dinp) + \bigO{\dinp}\) operations in
      \(\field\).

    \item Given nonnegative integers \(\ell,\dout\) with \(\ell+\dout \le N\) and
      given a polynomial \(p \in \pR_{< \dinp}\), one can compute
      \(\evaluate(\points_{\ell:\ell+\dout}, p)\) using \(\frac{3}{2}
      \dout \log_2(\dout) + \bigO{\dinp+\dout}\) operations in \(\field\).

    %% this is correct but only valid for starting point at power of 2,
    %% which is not flexible enough for our needs
    % \item Given nonnegative integers \(k, \ell, \dinp, \dout\) such that
    %   \(\ell+\dout \le N\) and either \(k = 0\) and \(\dinp \le N\), or \(k\)
    %   is a power of \(2\) with \(\dinp \le k < N\), and given values \(\bm{p}
    %   \in \field^\dinp\), \(\extrapolate(\points_{k:k+\dinp},
    %   \points_{\ell:\ell+\dout}, \bm{p})\) can be computed using \(\frac{3}{2}
    %   (\dinp+\dout) \log_2(\dinp + \dout) + \bigO{\dinp+\dout}\) operations in
    %   \(\field\).

    \item Given nonnegative integers \(j, k, \ell, \dinp, \dout\) such that
      \(\ell+\dout \le N\) and \(0 \le jk < N\) and \(k > 0\) is a power of
      \(2\) with \(\dinp \le k\), and given values \(\bm{p}
      \in \field^\dinp\), \(\extrapolate(\points_{jk:jk+\dinp},
      \points_{\ell:\ell+\dout}, \bm{p})\) can be computed using \(\frac{3}{2}
      (\dinp+\dout) \log_2(\dinp + \dout) + \bigO{\dinp+\dout}\) operations in
      \(\field\).
  \end{itemize}
\end{lemma}
\begin{proof}
  We first consider interpolation. If \(j=0\) we may directly rely on the
  inverse truncated Fourier transform, which finds
  \(\interpolate(\points_{0:\dinp}, \bm{p})\) using \(\frac{3}{2} \dinp
  \log_2(\dinp) + \bigO{\dinp}\) operations in \(\field\) \cite[Thm.\,2 and
  Rmk.\,4]{vdH:issac04}.
  %% NOTE remark 4 is here to say that only a linear (without log) number
  %% of shifted additions are actually shifted (i.e., involve a multiplication by 2 or 1/2)
  %% -> see also Coxon 2022, An in-place truncated Fourier transform, introduction and Section 4
  The same algorithm solves the case \(j>0\), thanks to some scaling based on
  the observation that, since \(k\) is a power of \(2\), one has \(\bitrev{jk +
  i} = \bitrev{jk} + \bitrev{i}\) for any \(0 \le i < k\). Since the
  considered points \(\points_{jk:jk+\dinp}\) have an index of the form \(jk +
  i\) for some \(i < \dinp \le k\), they can be written as \(\point_{jk+i} =
  \ratio^{\bitrev{jk+i}} = \ratio^{\bitrev{jk}} \ratio^{\bitrev{i}}\);
  concisely, \(\points_{jk:jk+\dinp} = \ratio^{\bitrev{jk}}
  \points_{0:\dinp}\). Thus, one obtains the sought interpolant in the same
  cost \(\frac{3}{2} \dinp \log_2(\dinp) + \bigO{\dinp}\), by first
  computing \(p(x) = \interpolate(\points_{0:\dinp}, \bm{p})\) as above and
  then returning its scaling \(p(x / \ratio^{\bitrev{jk}})\).

  The cost bound for extrapolation follows by combining those for
  interpolation and evaluation, in order to obtain
  \(
    % \extrapolate(\points_{jk:jk+\dinp}, \points_{\ell:\ell+\dout}, \bm{p}) =
    \evaluate(\points_{\ell:\ell+\dout}, \interpolate(\points_{jk:jk+\dinp}, \bm{p}))
  \).

  It remains to prove the cost bound for evaluation. We do this in
  several stages, explaining why this bound essentially follows from the
  algorithm described in \cite[Sec.\,3]{vdH:issac04}. First, if \(\ell=0\),
  we start by reducing the degree of \(p\) in \(\bigO{\dinp + \dout}\)
  operations, by setting \(k\) as the smallest power of \(2\) such that \(k \ge
  n\), and computing \(\tilde{p} = p \bmod (x^k - 1)\); this polynomial satisfies
  \(\deg(\tilde{p}) < k = 2^{\lceil \log_2(\dout) \rceil}\) and \(\tilde{p}(\points_{0:\dout}) =
  p(\points_{0:\dout})\), so that these evaluations can be obtained in
  \(\frac{3}{2} \dout \log_2(\dout) + \bigO{\dout}\) by the truncated Fourier
  transform algorithm \cite[Thm.\,1]{vdH:issac04}.
  Second, if \(\ell\) is a power of \(2\) and \(\dout \le \ell\), we reduce to
  the case \(\ell=0\) through a scaling similar to the one done for
  interpolation above:
  \(p(\points_{\ell:\ell+\dout}) = p(\ratio^{N/(2\ell)} \points_{0:\dout})\),
  and the cost of computing the scaled polynomial \(p(\ratio^{N/(2\ell)} x)\)
  is in \(\bigO{\dinp}\). Third, if \(\ell+\dout\) is a power of \(2\), we also
  reduce to the first case \(\ell=0\) but applied to the principal
  \((\ell+\dout)\)-th root of unity \(\bar{\ratio} = \ratio^{-N/(\ell+\dout)}\) instead of
  \(\ratio\): indeed, for \(0 \le i < \ell+\dout\) we have
  \(\bitrev{i} = \bitrevlen{i}{\ell+\dout} N / (\ell+\dout)\) and \(1 +
  \bitrevlen{\ell+\dout-1-i}{\ell+\dout} = \ell+\dout -
  \bitrevlen{i}{\ell+\dout}\), and therefore
  \(\point_{i} = \ratio^{\bitrev{i}} = \bar{\ratio}^{-\bitrevlen{i}{\ell+\dout}} =
  \bar{\ratio}^{1+\bitrevlen{\ell+\dout-1-i}{\ell+\dout}}\); it follows that the
  evaluations of \(p(x)\) at \(\points_{\ell:\ell+\dout}\) are the reversal of
  the evaluations of \(p(\bar{\ratio} x)\) at \(\bar{\points}_{0:\dout}\) for
  \(\bar{\points} = (\bar{\ratio}^i)_{0 \le i < \ell+\dout}\).

  From now on, suppose we are not in one of the three cases above, and let
  \(k\) be the smallest power of \(2\) such that \(\ell+\dout \le 2k\); in
  particular, \(k \ge 2\) since \(\ell > 0\) and \(\ell+\dout\) is not a power
  of \(2\). The fourth case is when \(\ell \le k\): then we split the list of
  evaluation points as \(\points_{\ell:\ell+\dout} = \points_{\ell:k} \cup
  \points_{k:\ell+\dout}\); since \(k\) is a power of \(2\), and since
  \(\ell+\dout - k \le k\), evaluating \(p(x)\) at both sublists can be done by
  the third and second cases above, respectively, for a total of
  \[
    \textstyle
    \frac{3}{2} (k - \ell) \log_2(k - \ell) + \frac{3}{2}(\ell+\dout-k) \log_2(\ell+\dout-k) + \bigO{\dinp + (k - \ell) + (\ell+\dout-k)}
  \]
  operations in \(\field\), which is in \(\frac{3}{2} \dout \log_2(\dout) +
  \bigO{\dinp + \dout}\). Finally, the last case is when \(\ell > k\): since
  \(\points_{k:2k} = \ratio^{N/(2k)} \points_{0:k}\), we use a scaling to
  evaluate \(p(\ratio^{N/(2k)} x)\) at \(\points_{\ell-k:\ell-k+\dout}\), which
  brings us back to the previous case since now \(\ell - k \le k\).
  \qed
\end{proof}

\subsection{Classical divide and conquer algorithm}
\label{sec:interp:pmintbasis}

\subsubsection{Algorithm and context}
\label{sec:interp:pmintbasis:algo}

We start by a careful analysis of the divide and conquer approach due to
Beckermann and Labahn
\cite{BeckermannLabahn1994,BeckermannLabahn97}; we follow the presentation of
the algorithm in \cite[Sec.\,3.2]{HyunNeigerSchost2019}, specialized to
\(2\times 1\) input. It is analogous to \algoName{algo:PM-Basis2} presented in
\cref{sec:approx:pmbasis:algorithm}, and is also based on the general
recursiveness property recalled in \cref{lem:dnc_relbas}.  Note that the
assumption of pairwise distinct input points allows us to make use of
\cref{rmk:modified_relbas} concerning the residual computation: specifically,
we take \(\gamma=1\), using notation from that remark.

\begin{algorithm}[ht]
  \algoCaptionLabel{PM-IntBasis2}{\points,\aevs,\bevs,\ash}
  \begin{algorithmic}[1]
    \Require pairwise distinct points \(\points \in \field^d\) for \(d \in \NN\),
             values \((\aevs,\bevs) \in (\field^d)^2\),
             shift \(\ash \in \mathbb{Z}\)

   \Ensure a basis \(\resbasis \in \mpR{2}{2}_{\le d}\) of
    \(\intmod{\points}{\aevs,\bevs}\) in \(\ash\)-weak Popov form

    \State\InlineIf{\(d \le 2\)}{\Return \(\Call{algo:IntBasis2-base}{\points, \aevs,\bevs, \ash}\)}
    \State \(\dreca \gets  \lfloor d/2 \rfloor\); \(d_2 \gets d - \dreca\)
    \State \(\basis \gets \Call{algo:PM-IntBasis2}{\points_{0:\dreca}, \aevs_{0:\dreca}, \bevs_{0:\dreca}, \ash}\)
        \label{step:intbasis:firstbasis}
    \State \((\basis_{i})_{\dreca \le i < d} \gets \evaluate(\basis,\points_{\dreca:d})\)
        \label{step:intbasis:residual:mpe}
    \State \((\arevs, \brevs) \in (\field^{\drecb})^2 \gets\) values such that
    \(([\begin{smallmatrix}\arev_i \\ \brev_i\end{smallmatrix}])_{0 \le i < \drecb} =
    (\basis_{i} \, [\begin{smallmatrix} \aev_i \\ \bev_i \end{smallmatrix}])_{\dreca \le i < d}\)
        \label{step:intbasis:residual:mul}
    \State \(\basisr \gets \Call{algo:PM-IntBasis2}{\points_{\dreca:d}, \arevs, \brevs, \ash + \deg(p_{00}) - \deg(p_{11})}\)
        \label{step:intbasis:secondbasis}
    \State \Return \(\basisr \basis\)
        \label{step:intbasis:basis}
  \end{algorithmic}
\end{algorithm}

The degree properties of the bases computed during the run of this algorithm
are exactly the same as those of \algoName{algo:PM-Basis2} that were given in
\cref{sec:approx:pmbasis:cx_worst_case} (worst-case bound) and
\cref{sec:approx:pmbasis:cx_generic} (generic case bound). For this reason, the
complexity of \cref{algo:PM-IntBasis2} can be analyzed through steps that are
very similar to the ones in those sections. Specifically, the basis
multiplication step at \cref{step:intbasis:basis} of \cref{algo:PM-IntBasis2}
has the same cost  as the one at \cref{step:pmbasis:multiplication} of
\algoName{algo:PM-Basis2} since the arguments in
\cref{sec:approx:pmbasis:cx_worst_case,sec:approx:pmbasis:cx_generic} also
apply here verbatim. On the other hand, \cref{step:intbasis:residual:mul} only
uses \(\bigO{d}\) operations. Thus, the only task remaining here is to estimate
the cost of the basis evaluation at \cref{step:intbasis:residual:mpe}; compared
to the analysis of \algoName{algo:PM-Basis2} in
\cref{sec:approx:pmbasis:cx_worst_case,sec:approx:pmbasis:cx_generic}, this
cost will appear in place of the cost of the residual computation, studied in
\cref{lem:appbas_cx:residual}.

\subsubsection{Analysis for points in geometric progression}
\label{sec:interp:pmintbasis:geometric}

Interestingly, for input points in geometric progression, in the worst case we
obtain the exact same complexity bound as that for \algoName{algo:PM-Basis2},
stated in \cref{prop:pmbasis:cx_worst_case}. In the generic case with a
small-amplitude shift, the bound is also close to that in
\cref{prop:pmbasis:cx_generic}, the difference coming from how of the residual
is computed.

\begin{proposition}
  \label{prop:pmintbasis:geometric_complexity}%
  Let \(\points \in \field^d\) be points in geometric progression as in
  \cref{lem:geometric_progression}. For input values \((\aevs,\bevs) \in
  (\field^d)^2\) and a shift \(\ash \in \ZZ\), the call
  \(\Call{algo:PM-IntBasis2}{\points, \aevs, \bevs, \ash}\) costs
  \[
    \textstyle
    \frac{2\cstmatmul+3\cstpolmulrep{4}}{8} \timepm{\lceil d/2 \rceil} \log_2(\lceil d/2 \rceil)
    + \frac{5\cstmatmul+42}{8} \timepm{d}
    + \bigO{d \log(d)}
  \]
  operations in \(\field\), where \(\cstmatmul=8\) and \(\cstpolmulrep{4} = 4\)
  in the general case, and
  \(\cstmatmul \le 7\) and \(\cstpolmulrep{4} \le 4\)
  are constants for \(2 \times 2\) polynomial matrix multiplication
  and fixed operand \(4\)-polynomial multiplication
  in the case where \(\ash \in \{-1,0,1\}\) and the genericity condition
  \(\bar\Gamma_{\points,\ash}(\aevs,\bevs) \neq 0\) of
  \cref{cor:interp_generic_mindeg} holds.
\end{proposition}

\begin{proof}
  As mentioned above, the cost of the basis multiplication at
  \cref{step:intbasis:basis} of \cref{algo:PM-IntBasis2} is exactly the same as
  the cost of the corresponding step in \cref{algo:PM-Basis2}: the analyses in
  \cref{sec:approx:pmbasis} show that it costs \(\frac{\cstmatmul}{2}
  \timepm{\lfloor d/2 \rfloor} + \bigO{d}\) operations in \(\field\), with
  \(\cstmatmul\) as in the statement above.

  Now, for the residual evaluations, the main task is the basis evaluation at
  \cref{step:intbasis:residual:mpe}. In the worst case, we take individual
  degrees into account: write \(\basis = [\begin{smallmatrix}
    p_{00} & p_{01} \\ p_{10}& p_{11}
  \end{smallmatrix}]\) for the basis at \cref{step:intbasis:firstbasis}, with
  \(p_{ij}\) of degree \(\mathring{p}_{ij}\) for each \((i,j) \in \{0,1\}^2\). 
  The second item in \cref{lem:geometric_progression} states that evaluating
  \(p_{ij}\) at \(\drecb\) points costs \(\timemp{\mathring{p}_{ij}-1,\drecb} +
  \bigO{d}\) operations in \(\field\). This brings us to the same cost bound 
  \(\sum_{0\le i,j \le 1} \timemp{\mathring{p}_{ij}, \drecb} + \bigO{d}\) as in
  the proof of \cref{lem:appbas_cx:residual} concerning the residual
  computation of \cref{algo:PM-Basis2}. As a consequence, overall, we get the
  same worst-case bound as in \cref{prop:pmbasis:cx_worst_case}, that is, with
  constant factors \(7/2\) and \(41/4\), which are those in the statement here
  when \(\cstmatmul=8\) and \(\cstpolmulrep{4} = 4\).

  Under the genericity condition, as in the proof of
  \cref{prop:pmbasis:cx_generic}, one has \(\deg(\basis) \le \dinp\) with
  \(\dinp = \lceil (\dreca+1)/2 \rceil = 
  % \lfloor \dreca/2 \rfloor + 1 =
  \lfloor d/4 \rfloor + 1\). Furthermore, the \(4\) entries of \(\basis\) are
  evaluated at the same \(\drecb\) points from the input geometric progression.
  Hence, according to \cref{lem:geometric_progression}, these evaluations
  consist in \(4\) middle products \(\midprod{\bar{p}_{ij}, q,
  \dinp-1, \dinp-1+\drecb}\) for a fixed polynomial \(q \in
  \pR_{<\dinp-1+\drecb}\) built in \(\bigO{d}\) from the points, and where
  \(\bar{p}_{ij}\) is built in \(\bigO{d}\) from \(p_{ij}\). In total,
  this costs \(\cstpolmulrep{4}\timepm{\lfloor (\dinp-1 + \drecb) / 2 \rfloor}
  + \bigO{d}\) following the remarks in \cref{sec:approx:polmul_midprod}, and
  since \(\lfloor (\dinp-1 + \drecb) / 2 \rfloor = \lfloor (\lfloor d/4 \rfloor
  + \lceil d/2 \rceil) / 2 \rfloor \le 3 \lceil d/2 \rceil/4\),
  by superlinearity of \(\timepm{\cdot}\) this is in
  \(3\cstpolmulrep{4}/4 \cdot \timepm{\lceil d / 2 \rceil}
  + \bigO{d}\). We thus get the complexity equation
  \(
    \Cxity{d} = 2\Cxity{\lceil d/2 \rceil} + (2\cstmatmul + 3\cstpolmulrep{4})/4 \cdot \timepm{\lceil d/2 \rceil} + \bigO{d},
  \)
  and the same arguments as in the proof of \cref{prop:pmbasis:cx_worst_case}
  yield the claimed cost bound with constants
  \((2\cstmatmul+3\cstpolmulrep{4})/8\) and 
  \((10\cstmatmul+15\cstpolmulrep{4})/16\).
  The latter is less than \((5\cstmatmul+42)/8\).
  \qed
\end{proof}

\subsubsection{Analysis for radix-2 FFT points}
\label{sec:interp:pmintbasis:fft}

Here we consider FFT points \(\points \in \field^d\) as in
\cref{lem:radix2_fft}, thus with \(N \ge d\). As above, we only need to focus
on the cost of the basis evaluation at \cref{step:intbasis:residual:mpe}. This
amounts to evaluating four polynomials of degree at most \(\dreca\) at the
\(\drecb\) points of the sublist \(\points_{\dreca:d}\), which costs \(6 \drecb
\log_2(\drecb) + \bigO{d}\) operations in \(\field\) according to
the second item of \cref{lem:radix2_fft}. This can be rewritten as
\(\frac{2}{3} \timepm{\lceil d/2 \rceil} + \bigO{d}\), since \(\drecb = \lceil
d/2 \rceil\) and we are in a context where polynomial multiplication can be
done via evaluation \& interpolation at the FFT points in \(\points\); see for
example \cite[Thm.\,8.18]{GathenGerhard2013}.
%% NOTE M(drecb) = 9/2 2*drecb log_2(2*drecb) + O(d) = 9 drecb log_2(drecb) + O(d)
This remark also applies to the basis multiplication \(\basisr \basis\) at
\cref{step:intbasis:basis}, which costs \(4 \timepm{\lceil \lenP / 2 \rceil}  +
\bigO{d}\) via FFT evaluation \& interpolation, for any known bound \(\lenP \ge
\deg(\basisr\basis) + 1\). This leads to the next result.

\begin{proposition}
  \label{prop:pmintbasis:fft_complexity}%
  Let \(\points \in \field^d\) be the first \(d\) points in a length-\(N\) list
  of radix-2 FFT points in bit-reversed order as in \cref{lem:radix2_fft}, for
  \(N\) some power of \(2\) such that \(N \ge d\).
  For any input values \((\aevs,\bevs) \in
  (\field^d)^2\) and any shift \(\ash\), the call
  \(\Call{algo:PM-IntBasis2}{\points, \aevs, \bevs, \ash}\) uses
  \(
    {\textstyle\frac{7}{3}}\timepm{\lceil d/2 \rceil} \log_2(\lceil d/2 \rceil)
    + \bigO{d \log(d)}
  \)
  operations in \(\field\) in the worst case. If \(\ash \in \{-1,0,1\}\) and
  the genericity condition \(\bar\Gamma_{\points,\ash}(\aevs,\bevs) \neq 0\) of
  \cref{cor:interp_generic_mindeg} holds, this reduces to
  \(
    {\textstyle\frac{4}{3}} \timepm{\lceil d/2 \rceil}\log_2(\lceil d/2 \rceil)
    + \bigO{d\log(d)}
  \).
\end{proposition}

\begin{proof}
  Starting here with the complexity equation
  \begin{align*}
    \Cxity{d}
    & = 2\Cxity{\lceil d/2 \rceil} + 4\timepm{\lfloor d/2 \rfloor} + {\textstyle\frac{2}{3}}\timepm{\lceil d/2 \rceil} + \bigO{d} \\
    & = 2\Cxity{\lceil d/2 \rceil} + {\textstyle\frac{14}{3}}\timepm{\lceil d/2 \rceil} + \bigO{d},
  \end{align*}
  the worst-case bound follows from an analysis is similar to that in
  \cref{prop:pmbasis:cx_worst_case}. Indeed, let \(k = \lceil \log_2(d) \rceil
  - 2\) be the largest integer such that \(\lceil d / 2^k \rceil > 2\). By
  unrolling the above relation until reaching calls at order \(\lceil d/2^{k+1}
  \rceil \le 2\) (for which we use \(\Cxity{2} = \bigO{1}\)), we obtain
  \begin{equation*}
    % \label{eqn:recurrence_approx}
    \Cxity{d} = \frac{14}{3} \sum_{i=0}^{k} 2^i \timepm{\left\lceil \frac{d}{2^{i+1}} \right\rceil} + \bigO{d \log(d)}.
  \end{equation*}
  % where we have used \(\lceil d/2^{i+1} \rceil = \lceil \lceil d/2 \rceil/2^i \rceil\).
  From here, we conclude thanks to the bound on the sum given in \cref{eqn:use_lem_cst_factor}.
  Under the genericity assumption, the basis \(\basisr \basis\) has
  degree less than \(\lenP = \lceil (d+1)/2 \rceil + 1\) (see the proof of
  \cref{prop:pmbasis:cx_generic} where this is detailed for approximants), so
  that this product can be computed using \(4\timepm{\lceil \lenP / 2 \rceil} +
  \bigO{d} = 2\timepm{\lceil d / 2
  \rceil} + \bigO{d}\) operations. Adapting the above analysis with \(14/3 =
  4+2/3\) replaced by \(8/3 = 2+2/3\) gives the announced bound with constant
  factor \(4/3\).
  \qed
\end{proof}

Compared to the similar analyses of \algoName{algo:PM-Basis2} in
\cref{sec:approx:pmbasis:cx_worst_case,sec:approx:pmbasis:cx_generic}, here we
are in the case \(\cstmatmul = 4\), and the gain on the overall constant factor
comes from the fact the residual computation in those sections is costlier than
the basis evaluation here. Note that even in this context where one can use FFT
methods for computing the residual in \algoName{algo:PM-Basis2}, the latter
algorithm remains at a disadvantage: this residual will require a basis
evaluation just like the one here for \algoName{algo:PM-IntBasis2}, but will
also need to perform an additional evaluation of the input polynomials and then
an interpolation to recover the middle products.

In light of \cref{rmk:modified_relbas}, one could reuse the already known
evaluations of \(\firstbasis\) in order to multiply \(\secondbasis
\firstbasis\); this saves part of the time spent in this product.
%% NOTE this saves a third: computing \(\secondbasis \firstbasis\) takes \(8/3 \timepm{d/4} +
%% \bigO{d}\); the factor \(8/3 = 2+2/3\) becomes \(2 = 4/3 + 2/3\), and the
%% overall factor is simply \(1\).
Going further in this direction, this process of FFT multiplication also yields
evaluations of \(\secondbasis \firstbasis\), which could be returned as part of
the output: the advantage would be to retrieve some evaluations of
\(\secondbasis\) at \cref{step:intbasis:secondbasis}, and therefore fewer extra
evaluations would need to be computed in order to multiply \(\secondbasis
\firstbasis\). The algorithm in \cref{sec:interp:eval_intbasis} incorporates,
among others, these two ways of reducing the amount of evaluations to be done,
and quantifies the speed-up that this provides.

\subsubsection{Explicit basis when the output degree is maximal}
\label{sec:interp:pmintbasis:extremaldegree}

We end \cref{sec:interp:pmintbasis} with a remark that will prove useful for
designing variants of \cref{algo:PM-IntBasis2}. It shows that when the
interpolant basis output by this algorithm has the maximal possible degree,
we can give an explicit and simple description of this basis from the
input data.

\begin{lemma}
  \label{lem:intbasis_extremaldegree}%
  In the context of \cref{algo:PM-IntBasis2}, let \(\modulus = \prod_{0 \le k
  < d} (x - \point_k)\), \(\apo = \interpolate(\points, \aevs)\), and \(\bpo =
  \interpolate(\points, \bevs)\), and write \(\basis \in \mpR{2}{2}_{\le d}\)
  for the output basis. This basis reaches the maximal possible degree
  \(\deg(\basis) = d\) if and only if one of these two situations holds:
  \begin{align*}
    & \gcd(\apo,\modulus)=1
    \text{ and }
    \basis =
    \begin{bmatrix}
      \modulus & 0 \\
      -\bpo/\apo \bmod \modulus & 1
    \end{bmatrix},
    \\
    \text{or}\quad
    & \gcd(\bpo,\modulus)=1
    \text{ and }
    \basis =
    \begin{bmatrix}
      1 & -\apo/\bpo \bmod \modulus \\
      0 & \modulus
    \end{bmatrix}.
  \end{align*}
\end{lemma}
\begin{proof}
  We prove this inductively. Consider first the base case \(d \le 2\). Then
  \cref{algo:PM-IntBasis2} returns the \(\ash\)-Popov basis \(\basis\) of
  \(\intmod{\points}{\aevs,\bevs}\); we assume \(\deg(\basis) = d\), and we let
  \(\ddegP_{ij} = \deg(p_{ij})\). The \(\ash\)-Popov form implies \(d =
  \deg(\basis) = \max(\ddegP_{00}, \ddegP_{11})\), and we also know that
  \(\ddegP_{00} + \ddegP_{11} \le d\) thanks to
  \cref{lem:degree_bounds_2x2,cor:determinant_relmod}. Hence, either
  \(\ddegP_{00} = d\) and \(\ddegP_{11} = 0\), or  \(\ddegP_{00} = 0\) and
  \(\ddegP_{11} = d\). Suppose that we are in the first case. The
  \(\ash\)-Popov form definition implies that \(p_{11} = 1\) and \(p_{01} =
  0\), hence also \(\det(\basis) = p_{00}\). Since \(\det(\basis) = p_{00}\)
  divides \(\modulus\) by \cref{cor:determinant_relmod}, and both are monic of
  degree \(d\), we get \(p_{00} = \modulus\). We deduce that \(\aevs\) has no
  zero entries, or equivalently that \(\gcd(\apo,\modulus)=1\),
  since otherwise there would exist an element \((p,0)
  \in \intmod{\points}{\aevs,\bevs}\) with \(\deg(p) < d\) and this element
  cannot be generated by \(\basis\), which is absurd. It remains to
  observe that the relation \(p_{10} \apo +
  p_{11} \bpo = 0 \bmod \modulus\), with here \(\ddegP_{10} < \ddegP_{00} = d\)
  and \(p_{11} = 1\), yields \(p_{10} = -\bpo \apo^{-1} \bmod \modulus\). The
  proof in the other case \(\ddegP_{00} = 0\) and \(\ddegP_{11} = d\) is done
  using symmetric arguments.

  Now take \(d > 2\) and assume that the statement holds for any input
  instance of \cref{algo:PM-IntBasis2} with fewer than \(d\) points. Consider
  an instance with \(d\) points such that \cref{algo:PM-IntBasis2} returns a
  basis \(\secondbasis\firstbasis\) of
  \(\intmod{\points}{\aevs,\bevs}\) with \(\deg(\secondbasis\firstbasis) = d\).
  The interpolant bases \(\firstbasis\) and \(\secondbasis\), obtained from
  recursive calls at \(\dreca = \lfloor d/2 \rfloor\) and \(\drecb = \lceil d/2 \rceil\) points respectively, are such that
  \(\deg(\firstbasis) \le \dreca\) and \(\deg(\secondbasis) \le \drecb\).
  On the other hand, we have \(\dreca + \drecb = d =
  \deg(\secondbasis\firstbasis) \le \deg(\firstbasis) + \deg(\secondbasis)\).
  Thus \(\deg(\firstbasis) = \dreca < d\) and \(\deg(\secondbasis) = \drecb <
  d\). So, we can apply the induction hypothesis to both \(\secondbasis\) and
  \(\firstbasis\), which tells us that there exist \(\apor \in \pR_{<\dreca}\)
  and \(\bpor \in \pR_{<\drecb}\) such that
  \[
    \firstbasis =
    \begin{bmatrix}
      1 & \apor \\
      0 & \modulus_1
    \end{bmatrix}
    \text{ or }
    \firstbasis =
    \begin{bmatrix}
      \modulus_1 & 0 \\
      \apor & 1
    \end{bmatrix}
    \quad\text{ and }\quad
    \secondbasis =
    \begin{bmatrix}
      1 & \bpor \\
      0 & \modulus_2
    \end{bmatrix}
    \text{ or }
    \secondbasis =
    \begin{bmatrix}
      \modulus_2 & 0 \\
      \bpor & 1
    \end{bmatrix},
  \]
  where \(\modulus_1 = \prod_{0 \le k < \dreca} (x - \point_k)\) and
  \(\modulus_2 = \prod_{\dreca \le k < d} (x - \point_k)\). Degree
  considerations show that only two of the four situations are possible,
  specifically those that lead to
  \[
    \secondbasis \firstbasis =
    \begin{bmatrix}
      \modulus & 0 \\
      \gamma & 1
    \end{bmatrix}
    \text{ or }
    \secondbasis \firstbasis =
    \begin{bmatrix}
      1 & \gamma \\
      0 & \modulus
    \end{bmatrix}
    \text{ for some } \gamma \in \pR_{<d}.
  \]
  Using the same arguments as in the above case \(d\le 2\), in both cases one
  can relate \(\gamma\) to \((\apo,\bpo,\modulus)\), leading precisely to the
  matrices described in the statement, which concludes the proof.
  \qed
\end{proof}

\subsection{Acceleration using extrapolations and evaluated representations}
\label{sec:interp:eval_intbasis}

Following the remarks in \cref{sec:interp:pmintbasis:fft}, we design a variant
of the above algorithm called \nameref{algo:Eval-IntBasis2}, where we store and
reuse evaluations of \(\secondbasis\) and \(\firstbasis\) in order to speed-up
the basis multiplication \(\secondbasis \firstbasis\). For this, instead of
representing these interpolant bases as matrices of polynomials, we represent
them through their evaluations at sufficiently many points. These points are
selected from the input list \(\points\), which allows us to exploit the
evaluations of the first basis \(\firstbasis\) at \(\points_{\dreca:d}\), that
the approach asks to compute in any case.
In this evaluated representation, when we need to evaluate bases at some new
points, we make use of extrapolation: we never explicitly interpolate the
values into polynomials. In the case of FFT points, this makes no difference,
since FFT extrapolation starts with an interpolation step (see
\cref{lem:radix2_fft}), so that one could equivalently store at all times both
the polynomials and their values. However, avoiding interpolation and keeping
an always-evaluated representation does have a significant impact in the case
of points in geometric progression, thanks to the results in
\cref{lem:geometric_progression}. Indeed, for these points, direct
extrapolation is quite faster than interpolation followed by evaluation.

The difference is even more striking for points in arithmetic progressions: in
fact, \algoName{algo:Eval-IntBasis2} has exactly the same cost bound for
arithmetic and geometric progressions. This is because this algorithm only uses
extrapolations and this specific operation is performed with the same
efficiency for both types of points, see
\cref{lem:arithmetic_progression,lem:geometric_progression}. In contrast,
calling \nameref{algo:PM-IntBasis2} on an arithmetic progression would lead to
a cost bound in \(\bigO{\timepm{d} \log(d)^2}\).

One important note is that \algoName{algo:Eval-IntBasis2} takes as input
evaluations \(\aevs\) and \(\bevs\) and returns evaluations of the sought
interpolant basis. Thus, if one starts from polynomials \(\apo\) and \(\bpo\)
or wants the basis represented on the monomial basis, as in
\problemName{pbm:int}, then one must account for the cost of evaluating
\(\apo\) and \(\bpo\) at the points \(\points\) or interpolating the output
basis evaluations. If one deals with points in arithmetic progression, these
steps have some impact on the leading constant in front of the term in
\(\timepm{d}\log(d)\), unlike for FFT points or geometric progressions.

\subsubsection{Algorithm}
\label{sec:interp:eval_intbasis:algo}

To pick the right number of evaluation points for representing the bases that
are computed along the algorithm, we need some properties on the degrees of
their entries. A first observation is that, according to \cref{lem:dnc_relbas},
from the diagonal degrees of the bases \(\firstbasis\) and \(\secondbasis\)
obtained by recursive calls, we are able to deduce the diagonal degrees of the
output basis \(\secondbasis\firstbasis\) without doing any field operation. In our
context where we aim to store only evaluations of the computed bases, this
means that we can propagate through the recursion exact information on diagonal
degrees of these bases, even though we never interpolate the data back to
polynomials.

A crucial consequence of knowing diagonal degrees is that, even though we do
not know the polynomials \(p_{00}\) and \(p_{11}\) explicitly, we have direct
access to the updated shift \(\ash + \deg(p_{00}) - \deg(p_{11})\) for the
second recursive call (see  \cref{step:intbasis:secondbasis} of
\cref{algo:PM-IntBasis2}). Another consequence is that we can deduce bounds on
the off-diagonal degrees using weak Popov form properties, so that we have
degree bounds for all basis entries and we know how many evaluation points
suffice for an evaluated representation. Note that the output basis degree
might reach the number of points available in \(\points\), in which case one
would need one extra evaluation point; yet, this case is easily detected and
can be handled thanks to the explicit basis description in
\cref{lem:intbasis_extremaldegree}.

This leads to \algoName{algo:Eval-IntBasis2}, described in
\cref{algo:Eval-IntBasis2}. It does the exact same computations as
\algoName{algo:PM-IntBasis2}, but maintains an evaluated representation for the
intermediate and output interpolant bases. It contains two subprocedures:
\begin{itemize}[noitemsep,topsep=2pt]
  \item \textproc{BasisDegBounds}, which uses known diagonal degrees to deduce
    bounds on off-diagonal degrees, based on the definition of a shifted weak
    Popov form;
  \item \textproc{BasisExtrapolate}, which performs
    \(\extrapolate(\ptsinp, \ptsout, \evbasis)\) for some basis
    \(\evbasis\) given in evaluated representation, using known degree bounds on the
    basis entries while making sure to handle the case of
    a basis of maximal degree correctly based on \cref{lem:intbasis_extremaldegree}.
\end{itemize}
The correctness thus follows directly from that of
\algoName{algo:PM-IntBasis2}, along with easy verifications of the degree
considerations ensuring that bases are evaluated at sufficiently many points
during the whole computation. Another easy verification is about the inequality
\(d-\lenP \ge \dreca - \lenP_1\), ensuring that we do not try to access
evaluations of \(\firstbasis\) that we do not have at
\cref{algo:Eval-IntBasis2:basis_multiply}; this is implied by \(\lenP \le
\drecb+\lenP_1\), itself deduced from \(\lenP \le \lenP_1 + \lenP_2\). The
latter roughly translates the fact that \(\deg(\secondbasis\firstbasis) \le
\deg(\secondbasis)+\deg(\firstbasis)\) and appears in the construction of
\(\lenP\) at \cref{algo:Eval-IntBasis2:lenP_mul}.

\begin{algorithm}[htbp]
  \algoCaptionLabel{Eval-IntBasis2}{\points,\aevs,\bevs,\ash}
  \begin{algorithmic}[1]
    \Require pairwise distinct points \(\points \in \field^d\) for \(d \in \NN\),
             values \((\aevs,\bevs) \in (\field^d)^2\),
             shift \(\ash \in \mathbb{Z}\)%

   \Ensure a list of matrices \(\evbasis \in (\mKK{2}{2})^{\lenP}\) with \(\lenP \in \{0,\ldots,d\}\),
    and integers \(\ddegP_0,\ddegP_1 \in \{0,\ldots,d\}\),
    % and such that \(\lenP = \min(d, 1+\max(\textproc{BasisDegBounds}(d,\ddegP_0,\ddegP_1,\ash))) \in \{0,\ldots,d\}\)
    which are the evaluations \(\evbasis = \resbasis(\points_{d-\lenP:d})\)
    and diagonal degrees \((\ddegP_0,\ddegP_1)\)
    of an \(\ash\)-weak Popov basis \(\basis \in \mpR{2}{2}_{\le d}\) of
    \(\intmod{\points}{\aevs,\bevs}\)
    with either \(\deg(\basis) = \lenP = d\) or \(\deg(\basis) < \lenP\)

    \If{\(d \le 2\)}
      \State \(\basis = [\begin{smallmatrix} p_{00} & p_{01} \\ p_{10}& p_{11} \end{smallmatrix}] \in \mpR{2}{2}_{\le 2} \gets \Call{algo:IntBasis2-base}{\points, \aevs,\bevs, \ash}\);
      ~\(\lenP \gets \min(2, \deg(\resbasis) + 1)\)
      \label{algo:Eval-IntBasis2:base_lenP}
      \State \Return \(\resbasis(\points_{d-\lenP:d})\) and \((\deg(p_{00}),\deg(p_{11}))\)
      \label{algo:Eval-IntBasis2:base_return}
    \EndIf

    \State \(\dreca \gets \lfloor d/2 \rfloor\); ~\(\drecb \gets d - \dreca\)

    \State \(\evfirstbasis, (\ddegP_0,\ddegP_1) \gets \Call{algo:Eval-IntBasis2}{\points_{0:\dreca}, \aevs_{0:\dreca}, \bevs_{0:\dreca}, \ash}\)
    \MyComment{indexing \(\evfirstbasis = (\firstbasis_i)_{\dreca - \lenP_1 \le i < \dreca}\), length \(\lenP_1\)}
    \label{algo:Eval-IntBasis2:firstcall}

    \State \((\firstbasis_{i})_{\dreca \le i < d} \in (\mKK{2}{2})^{\drecb} \gets
    \textproc{BasisExtrapolate}(\points_{\dreca-\lenP_1:\dreca}, \points_{\dreca:d}, \evfirstbasis, \dreca, \ddegP_0, \ddegP_1, \ash)\)
      \label{algo:Eval-IntBasis2:residual_extrapolate}

    \State \((\arevs, \brevs) \in (\field^{d-\dreca})^2 \gets\) values such that
    \(([\begin{smallmatrix}\arev_i \\ \brev_i\end{smallmatrix}])_{0 \le i < \drecb} =
    (\firstbasis_{i} \, [\begin{smallmatrix} \aev_i \\ \bev_i \end{smallmatrix}])_{\dreca \le i < d}\)
    \label{algo:Eval-IntBasis2:residual_multiply}

    \State \(\evsecondbasis, (\ddegQ_0,\ddegQ_1) \gets \Call{algo:Eval-IntBasis2}{\points_{\dreca:d}, \arevs, \brevs, \ash + \ddegP_{0} - \ddegP_{1}}\)
    \MyComment{indexing \(\evsecondbasis = (\secondbasis_i)_{d - \lenP_2 \le i < d}\), length \(\lenP_2\)}
    \label{algo:Eval-IntBasis2:secondcall}

    \State \(\lenP \gets \min(d, 1+\max(\textproc{BasisDegBounds}(d,\ddegQ_0+\ddegP_0,\ddegQ_1+\ddegP_1,\ash)))\)
    \label{algo:Eval-IntBasis2:lenP_wpopov}

    \State\InlineIfElse{\(\dreca \in \{\ddegP_0, \ddegP_1\}\) \OR{} \(\drecb \in \{\ddegQ_0,\ddegQ_1\}\)}{%
      \(\lenP \gets \min(\lenP, \lenP_1 + \lenP_2)\)%
    }{%
      \(\lenP \gets \min(\lenP, \lenP_1 + \lenP_2 - 1)\)%
    }
    \label{algo:Eval-IntBasis2:lenP_mul}

    \State \((\secondbasis_{i})_{d-\lenP \le i < d-\lenP_2} \gets
    \textproc{BasisExtrapolate}(\points_{d-\lenP_2:d}, \points_{d-\lenP:d-\lenP_2}, \evsecondbasis, \drecb, \ddegQ_0, \ddegQ_1, \ash + \ddegP_{0} - \ddegP_{1})\)
    \label{algo:Eval-IntBasis2:basis_extrapolate}

    \State \Return \((\secondbasis_i \firstbasis_i)_{d-\lenP \le i < d}\) and \((\ddegP_0+\ddegQ_0, \ddegP_1+\ddegQ_1)\)
    \MyComment{\(\firstbasis_i\) is known since \(d-\lenP \ge \dreca - \lenP_1\)}
    \label{algo:Eval-IntBasis2:basis_multiply}

    \Statex
    \Subprocedure \(\textproc{BasisDegBounds}(d,\ddegP_0,\ddegP_1,\ash)\)

    \Require nonnegative integers \(d,\ddegP_0,\ddegP_1 \in \NN\),
             integer \(\ash \in \ZZ\)
     \MyComment{diagonal degrees \(\ddegP_0,\ddegP_1\)}

    \Ensure integer tuple in \(\ZZ^4\) with entries in \(\{-1,0,\ldots,d\}\)
    \MyComment{degree bounds on \(p_{00}, p_{01}, p_{10}, p_{11}\)}

    \State\InlineIf{\((\ddegP_0,\ddegP_1) = (d,0)\)}{%
      \Return \((d,-1,d-1,0)\)%
    }
    \MyComment{basis \([\begin{smallmatrix} \modulus & 0 \\ p_{01} & 1 \end{smallmatrix}]\),
    see \cref{lem:intbasis_extremaldegree}}

    \State\InlineElseIf{\((\ddegP_0,\ddegP_1) = (0,d)\)}{%
      \Return \((0,d-1,-1,d)\)%
    }
    \MyComment{basis \([\begin{smallmatrix} 1 & p_{01} \\ 0 & \modulus \end{smallmatrix}]\),
    see \cref{lem:intbasis_extremaldegree}}

    \State\InlineElse{% {\(\max(\ddegP_0,\ddegP_1) < d\)}{%
      \Return \((\ddegP_0, \max(-1, \min(d-1, \ddegP_0+\ash-1)), \max(-1, \min(d-1, \ddegP_1-\ash)), \ddegP_1)\)%
    }

    \Statex
    \Subprocedure \(\textproc{BasisExtrapolate}(\ptsinp,\ptsout,\evbasis,d,\ddegP_0,\ddegP_1,\ash)\)

    \Require pairwise distinct points \(\ptsinp = (\point_{i})_{0 \le i < \dinp} \in \field^\dinp\),
             target points \(\ptsout \in \field^\dout\),
             % = (\varpi_0,\ldots,\varpi_{\dout-1})
             list of matrix values \(\evbasis = (\basis_{i})_{0 \le i < \dinp} \in (\mKK{2}{2})^\dinp\),
             nonnegative integers \((d,\ddegP_0,\ddegP_1)\),
             integer \(\ash \in \ZZ\)
             such that \(\ddegP_0 + \ddegP_1 \le d\) and
             \(\dinp \ge \min(d, 1+\max(\textproc{BasisDegBounds}(d,\ddegP_0,\ddegP_1,\ash))) \in \{0,\ldots,d\}\)

    \Ensure list in \((\mKK{2}{2})^{\dout}\) of \(\dout\) matrix values

    \If{\(\max(\ddegP_0,\ddegP_1) < d\)} \MyComment{typical case \(\deg(P) < d\): do \(\extrapolate(\ptsinp, \ptsout, \evbasis)\)}

      \State \((\lenP_{00}, \lenP_{01}, \lenP_{10}, \lenP_{11}) \gets (1,1,1,1) + \textproc{BasisDegBounds}(d,\ddegP_0,\ddegP_1,\ash)\)
      \MyComment{\(\deg(p_{jk}) < \lenP_{jk} \le \dinp\)}
      \State \InlineFor{\(j, k \in \{0,1\}\)}{%
        \((E_{i,jk})_{0 \le i < \dout} \in \field^\dout
        \gets \extrapolate(\ptsinprange{\dinp-\lenP_{jk}:\dinp}, \ptsout, (\basis_{i,jk})_{\dinp-\lenP_{jk} \le i < \dinp})\)
      }
      \State \Return \((E_{i})_{0 \le i < \dout}\),
      where \(E_i \in \mKK{2}{2}\) has entry \((j,k)\) equal to \(E_{i,jk}\)
      \label{algo:BasisExtrapolate:direct}

    \EndIf

    \State \((u_i)_{0 \le i < \dout} \in \field^\dout \gets \evaluate(\ptsout, \modulus)\),
    where \(\modulus = \prod_{0 \le i < \dinp} (x - \point_i)\)

    \If{\(\ddegP_0 = d\)}
    \MyComment{case \((\ddegP_0,\ddegP_1) = (d,0)\) and \(\dinp = d\):
    evaluate \([\begin{smallmatrix} \modulus & 0 \\ p_{01} & 1 \end{smallmatrix}]\),
    see \cref{lem:intbasis_extremaldegree}}

      \State \Return \(([\begin{smallmatrix} u_i & 0 \\ v_i & 1 \end{smallmatrix}])_{0 \le i < \dout}\),
      where \((v_i)_{0 \le i < \dout} \in \field^\dout \gets \extrapolate(\ptsinp, \ptsout, (\basis_{i,10})_{0 \le i < \dinp})\)
      \label{algo:BasisExtrapolate:maxdeg0}

    \Else \MyComment{case \((\ddegP_0,\ddegP_1) = (0,d)\) and \(\dinp = d\):
    evaluate \([\begin{smallmatrix} 1 & p_{01} \\ 0 & \modulus \end{smallmatrix}]\),
    see \cref{lem:intbasis_extremaldegree}}

      \State \Return \(([\begin{smallmatrix} 1 & v_i \\ 0 & u_i \end{smallmatrix}])_{0 \le i < \dout}\),
      where \((v_i)_{0 \le i < \dout} \in \field^\dout \gets \extrapolate(\ptsinp, \ptsout, (\basis_{i,01})_{0 \le i < \dinp})\)
      \label{algo:BasisExtrapolate:maxdeg1}

    \EndIf
  \end{algorithmic}
\end{algorithm}

\subsubsection{Complexity bounds for points in arithmetic or geometric progressions}
\label{sec:interp:eval_intbasis:arith_geom}

We start with a complexity analysis for arithmetic or geometric progressions,
two types of points that share the same bound for the cost of extrapolation.

\begin{proposition}
  \label{prop:evalintbasis:geometric_complexity}%
  Let \(\points \in \field^d\) be points either in arithmetic progression as in
  \cref{lem:arithmetic_progression} or in geometric progression as in
  \cref{lem:geometric_progression}. For any input
  values \((\aevs,\bevs) \in (\field^d)^2\) and any shift \(\ash\), the call
  \(\Call{algo:Eval-IntBasis2}{\points, \aevs, \bevs, \ash}\) costs
  \[
    3\timepm{\lceil d/2 \rceil} \log_2(\lceil d/2 \rceil)
    + {\textstyle\frac{21}{2}}\timepm{d}
    + \bigO{d \log(d)}
  \]
  operations in \(\field\) in the worst case, which reduces to
  \[
    \textstyle
    \frac{5\cstpolmulrep{4}}{8} \timepm{\lceil d/2 \rceil}\log_2(\lceil d/2 \rceil)
    + \frac{25\cstpolmulrep{4}}{16}\timepm{d}
    + \bigO{d\log(d)}
  \]
  if \(\ash \in \{-1,0,1\}\) and the genericity condition
  \(\bar\Gamma_{\points,\ash}(\aevs,\bevs) \neq 0\) of
  \cref{cor:interp_generic_mindeg} holds.
\end{proposition}

\begin{proof}
  The two calls to \textproc{BasisExtrapolate} are the only steps that
  contribute to the terms in \(\timepm{d} \log_2(d)\) and \(\timepm{d}\).
  Consider the first such call at
  \cref{algo:Eval-IntBasis2:residual_extrapolate}. We first analyze the typical
  case \(\max(\ddegP_0,\ddegP_1) < \dreca\). Extrapolating the entry \((j,k)\)
  from \(\ptsinprange{\dinp-\lenP_{jk}:\dinp}\) of length \(\lenP_{jk}\) to
  \(\ptsout\) of length \(\drecb\) costs \(\timemp{\lenP_{jk}-1, \drecb} +
  \bigO{d}\) by \cref{lem:arithmetic_progression,lem:geometric_progression} (we
  ignore the possibility that \(\lenP_{jk}-1 = -1\), which does not impact the
  next arguments). In total, the call to \textproc{BasisExtrapolate} thus costs
  \[
    \sum_{jk \in \{00,01,10,11\}}
    \timepm{\left\lfloor \frac{\lenP_{jk}-1 + \lceil d/2 \rceil}{2}
    \right\rfloor} + \bigO{d}
    \;\;\subseteq\;\; 2\timepm{\left\lceil \frac{3d}{4} \right\rceil} + \bigO{d},
  \]
  according to \hyptime{hyp:timepm:midprod}, where the upper bound follows from
  the superlinearity of \(\timepm{\cdot}\) and from the bounds
  \(\lenP_{00}-1+\lenP_{11}-1 \le \dreca = \lfloor d/2 \rfloor\) and
  \(\lenP_{01}-1+\lenP_{10}-1 < \dreca\). Now, in the maximal degree
  case, say with \((\ddegP_0,\ddegP_1) = (\dreca,0)\) (the other case \((\ddegP_0,\ddegP_1) = (0,\dreca)\) is
  handled similarly), the evaluation of \(\modulus\) to obtain \((u_i)_{0 \le i
  < \drecb}\) is done in linear time \(\bigO{d}\), and the
  extrapolation of a single polynomial costs \(\timemp{\dreca-1, \drecb}\),
  which is within the above bound.
  A similar analysis for the second call to \textproc{BasisExtrapolate} shows
  that here too the typical case \(\max(\ddegQ_0,\ddegQ_1) < \drecb\) is
  costlier than the maximal degree cases, and that in this typical case the
  total cost is
  \begin{align*}
    & \sum_{jk \in \{00,01,10,11\}}
    \timepm{\left\lfloor \frac{\lenP_{jk}-1 + \lenP - \lenP_2}{2} \right\rfloor} + \bigO{d}
    \\
    & \quad \subseteq\;\;
    \timepm{\left\lceil \frac{d}{4} \right\rceil + \lenP_1}
    + \timepm{\left\lfloor \frac{d}{4} \right\rfloor + \lenP_1}
    + \bigO{d},
  \end{align*}
  with the \(\lenP_{jk}\)'s now related to the second basis \(\secondbasis\),
  and we used \(\lenP_{00}-1+\lenP_{11}-1 \le \drecb = \lceil d/2 \rceil\),
  \(\lenP_{01}-1+\lenP_{10}-1 < \drecb\), and \(\lenP - \lenP_2 \le \lenP_1\)
  to derive the right-hand side bound.

  With the worst-case bound \(\lenP_1 \le \lfloor d/2 \rfloor\), we get the
  same cost bound for the second extrapolation as for the first one,
  hence the complexity recurrence \(\Cxity{d} = 2\Cxity{\lceil d/2 \rceil} + 3
  \timepm{d} + \bigO{d}\). Using the same analysis as in the proof of
  \cref{prop:pmbasis:cx_worst_case},
  we obtain the general bound.

  Now, under genericity, we have the refined bounds \(\lenP_1 \le
  \lceil (\dreca+1)/2 \rceil = \lfloor d/4 \rfloor + 1\) and \(\lenP_2 \le
  \lceil (\drecb+1)/2 \rceil \le \lceil d/4 \rceil + 1\), as in the proof of
  \cref{prop:pmbasis:cx_generic}. Thanks to these bounds, we can assume that
  \textproc{BasisExtrapolate} enters the typical case for both computed bases.
  Similarly to what was done in \cref{prop:pmintbasis:geometric_complexity},
  we can use these uniform bounds for all matrix entries in order to exploit
  the constant \(\cstpolmulrep{4}\). The basis evaluation at
  \cref{algo:Eval-IntBasis2:residual_extrapolate} asks to extrapolate the four
  matrix entries of \(\evfirstbasis\) from the same
  \(\ptsinprange{\dinp-\lenP_{1}:\dinp}\) of length \(\le \lfloor d/4
  \rfloor+1\) to the same \(\ptsout\) of length \(\drecb\). This boils down to
  four middle products with a fixed operand, according to
  \cref{lem:arithmetic_progression,lem:geometric_progression}, which costs
  % \(\cstpolmulrep{4}\timemp{\lfloor d/4 \rfloor, \drecb} + \bigO{d}\)
  \(\cstpolmulrep{4}\timepm{\lfloor (\lfloor d/4 \rfloor + \drecb)/2 \rfloor} +
  \bigO{d}\) operations in \(\field\). Following the bounds in
  \cref{prop:pmintbasis:geometric_complexity}, this is in \(3\cstpolmulrep{4}/4
  \cdot \timepm{\lceil d / 2 \rceil} + \bigO{d}\). Similarly, the basis
  multiplication at \cref{algo:Eval-IntBasis2:basis_multiply} asks to 
  extrapolate the four matrix entries of \(\evsecondbasis\) from the same
  \(\lenP_2 \le \lceil d/4 \rceil + 1\) points to the same \(\lenP - \lenP_2
  \le \lenP_1 \le \lfloor d/4 \rfloor + 1\) points. Thanks to
  \cref{lem:arithmetic_progression,lem:geometric_progression} this costs
  % \(\cstpolmulrep{4}\timemp{\lceil d/4 \rceil, \lfloor d/4 \rfloor + 1} + \bigO{d}\)
  \(\cstpolmulrep{4}\timepm{\lfloor (\lceil d/4 \rceil + \lfloor d/4 \rfloor +
  1)/2 \rfloor} + \bigO{d}\). Since \(\lfloor (\lceil d/4 \rceil + \lfloor d/4
  \rfloor + 1)/2 \rfloor \le \lceil d/4 \rceil\), this is in 
  \(\cstpolmulrep{4}/2 \cdot \timepm{\lceil d/2 \rceil} + \bigO{d}\)
  operations in \(\field\). Altogether, we get the complexity recurrence
  \(\Cxity{d} = 2\Cxity{\lceil d/2 \rceil} + 5\cstpolmulrep{4}/4 \cdot
  \timepm{\lceil d/2 \rceil} + \bigO{d}\). Doing an analysis similar to the
  one in the proof of \cref{prop:pmbasis:cx_worst_case}, we get the second
  announced bound.
  \qed
\end{proof}

\subsubsection{Complexity bounds for radix-2 FFT points}
\label{sec:interp:eval_intbasis:fft}

Here, because we wish to call FFT extrapolation as described in
\cref{lem:radix2_fft}, we mostly focus on \(d\) being a power of
\(2\), and we list the FFT points in reverse order, \(\points =
(\ratio^{d-1-\bitrevlen{i}{d}})_{0 \le i < d} =
(\ratio^{\bitrevlen{d-1-i}{d}})_{0 \le i < d}\), to accommodate the fact that
\algoName{algo:Eval-IntBasis2} stores evaluations at points at the end of this
list. Thanks to this, during a call to this algorithm with input points
\(\points\), all extrapolations that occur fit in the framework of
\cref{lem:radix2_fft}. Indeed they have input points of the form
\(\points_{jk-\dinp:jk}\), where \(j, k, \dinp\) are such that \(0 \le jk < d\)
and \(k > 0\) is a power of \(2\) with \(\dinp \le k\), and
\(\points_{jk-\dinp:jk}\) consists of \(\dinp \le k\) consecutive FFT points
starting at a multiple of \(k\); more precisely,
\(\points_{jk-\dinp:jk} = (\ratio^{d - 1 - \bitrevlen{jk-m+i}{d}})_{i=0,1,\ldots,\dinp-1}
% = (\ratio^{\bitrevlen{d - 1 - jk+m-i}{d}})_{0 \le i < \dinp}
% = (\ratio^{\bitrevlen{(d/k - j)k + m - 1-i}{d}})_{0 \le i < \dinp}
= (\ratio^{\bitrevlen{(d/k - j)k + \bar\imath}{d}})_{\bar\imath =
\dinp-1,\ldots,1,0}\).

One recovers the cost bound for the general case, where \(d\) is not a power of
\(2\), by changing the splitting \(d = \dreca+\drecb\) for recursive calls:
we would now take \(\dreca\) as the largest power of \(2\) strictly less than
to \(d\). This is the other standard way to write the divide and conquer
recursion in such algorithms, from which it is also easily observed that the
overall leading constant in the cost bound remains the same as the leading
constant for the case where \(d\) is a power of \(2\); in the context of the
half-gcd algorithm, this is detailed in \cite[Sec.\,4.4]{vdHoeven2025}.

\begin{proposition}
  \label{prop:evalintbasis:fft_complexity}%
  Let \(\points = (\ratio^{\bitrev{d-1-i}})_{0\le i < d} \in \field^d\) be the
  reversed list of the first \(d\) points in a length-\(N\) list of radix-2 FFT
  points in bit-reversed order as in \cref{lem:radix2_fft}, for \(N\) some
  power of \(2\) such that \(N \ge d\). Consider the variant of
  \algoName{algo:Eval-IntBasis2} where one rather picks \(\dreca \gets
  2^{\lceil \log_2(d) \rceil - 1}\) which is the largest power of \(2\)
  less than \(d\). For any input values \((\aevs,\bevs) \in
  (\field^d)^2\) and any shift \(\ash\), the call
  \(\Call{algo:Eval-IntBasis2}{\points, \aevs, \bevs, \ash}\) costs
  \(
    \timepm{\lceil d/2 \rceil} \log_2(\lceil d/2 \rceil)
    + \bigO{d \log(d)}
  \)
  operations in \(\field\) in the worst case. If \(\ash \in \{-1,0,1\}\) and
  the genericity condition \(\bar\Gamma_{\points,\ash}(\aevs,\bevs) \neq 0\) of
  \cref{cor:interp_generic_mindeg} holds, this reduces to
  \(
    {\textstyle\frac{5}{6}} \timepm{\lceil d/2 \rceil}\log_2(\lceil d/2 \rceil)
    + \bigO{d\log(d)}
  \).
\end{proposition}

\begin{proof}
  As explained above, we focus on \(d\) being a power of \(2\), in which case
  \(\dreca = \drecb = d/2\), and all extrapolations that occur fit within the
  situation where the cost bound of \cref{lem:radix2_fft} is valid. Following
  the analysis done for \cref{prop:evalintbasis:geometric_complexity},
  from \cref{lem:radix2_fft}, we obtain that the first call to
  \textproc{BasisExtrapolate} costs
  \begin{equation}
    \label{eqn:evalintbasis:fftcomp:firstcall}%
    \sum_{jk \in \{00,01,10,11\}}
    {\textstyle\frac{3}{2}} (\lenP_{jk} + d/2) \log_2(\lenP_{jk} + d/2) + \bigO{d}
    \;\;\subseteq\;\;
    {\textstyle\frac{9}{2}} d \log_2(d) + \bigO{d}
  \end{equation}
% this is in \(\timepm{d/2} + \bigO{d}\).
  operations in \(\field\), whereas the second call costs
  \begin{align*}
    & \sum_{jk \in \{00,01,10,11\}}
    {\textstyle\frac{3}{2}} (\lenP_{jk} + \lenP - \lenP_2) \log_2(\lenP_{jk} + \lenP - \lenP_2) + \bigO{d}
    \\
    & \quad \subseteq\;\;
    {\textstyle\frac{3}{2}} (d + 4 \lenP_1) \log_2(d + 4 \lenP_1) + \bigO{d},
    % {\textstyle\frac{9}{2}} d \log_2(d) + \bigO{d},
  \end{align*}
  where we have used \(\lenP - \lenP_2 \le \lenP_1\).
  The worst case bound \(\lenP_1 \le d/2\) implies that the cost of the second
  call is also within \(\frac{9}{2} d \log_2(d) + \bigO{d}\). As noted in
  \cref{sec:interp:pmintbasis:fft}, this is \(\timepm{d/2} + \bigO{d}\) in
  this FFT context \cite[Thm.\,8.18]{GathenGerhard2013}. Thus we arrive at the
  complexity recurrence \(\Cxity{d} = 2\Cxity{d/2} + 2\timepm{d/2} +
  \bigO{d}\), which leads to the first bound in the statement.

  Now, under the genericity assumption, we can use the refined bound \(\lenP_1
  \le d/4 + 1\), showing that the second call to \textproc{BasisExtrapolate}
  costs \(\frac{2}{3}\timepm{d/2} + \bigO{d}\). This gives the same complexity
  recurrence as above with the term \(2\timepm{d/2}\) replaced by \(\frac{5}{3}
  \timepm{d/2}\), leading to the second announced bound.
  \qed
\end{proof}

\subsection{Further optimizations: deduce evaluations in linear cost}
\label{sec:interp:eval_intbasis_opti}

To conclude \cref{sec:interp}, we suggest two ideas that can lead to further
optimizations of \algoName{algo:Eval-IntBasis2}. We focus on describing these
optimizations and the gain they bring on the resulting leading constant in
complexity bounds. However, we leave as a perspective the study of genericity
aspects as well as the formal description of an algorithm incorporating
these optimizations. As outlined below, this algorithm would detect whether the
conditions allowing these optimizations are met and would otherwise switch to
the above version in \cref{algo:Eval-IntBasis2}. As can be inferred from the
explanations below, the detection step has low cost, since it only consists in
checking that some known evaluations are nonzero. Thus, we expect such an
optimized algorithm to never perform slower than \nameref{algo:Eval-IntBasis2}
called on the same input.

\subsubsection{First basis extrapolation: exploiting the known determinant.}

For generic input, the basis \(\firstbasis\) computed by the first recursive
call has determinant exactly \(\modulus_1 = \prod_{0 \le i < \dreca} (x -
\point_i)\); indeed, \(\det(\firstbasis)\) always divides \(\modulus_1\), and 
\cref{app:approx:genericity} shows that \(\deg(\det(\firstbasis)) = \dreca\)
generically. Furthermore, for points in geometric progression, or in arithmetic
progression, or for FFT points, the evaluations of \(\modulus_1\) at the second
half of points \(\points_{\dreca:d}\) can be determined in linear time, that
is, \(\bigO{d}\) operations in \(\field\). The first call to
\textproc{BasisExtrapolate} in \nameref{algo:Eval-IntBasis2} aims to compute
\(\firstbasis(\points_{\dreca:d})\), and can therefore exploit the knowledge of
\(\modulus_1(\points_{\dreca:d}) = \det(\firstbasis)(\points_{\dreca:d})\).
Indeed, we can deduce the evaluations of \(p_{11}\) from those of the three
other entries, using \(p_{11}(\point_i) = (\modulus_1(\point_{i}) +
p_{01}(\point_i) p_{10}(\point_i)) p_{00}(\point_i)^{-1}\), provided
\(p_{00}(\point_i) \neq 0\). Then, the algorithmic approach would be to use
extrapolations to find the needed values for the entries \(p_{00}\) and
\(p_{01}\); if there are zero values in both \(p_{00}(\points_{\dreca:d})\) and
\(p_{01}(\points_{\dreca:d})\), compute all extrapolations as done in
\nameref{algo:Eval-IntBasis2}; else if all values in
\(p_{00}(\points_{\dreca:d})\) are nonzero, extrapolate \(p_{10}\) and use the
above formula to deduce the values of \(p_{11}(\points_{\dreca:d})\) in linear
time; else if all values in \(p_{01}(\points_{\dreca:d})\) are nonzero,
extrapolate \(p_{11}\) and use the known determinant values to deduce those of
\(p_{10}(\points_{\dreca:d})\) in linear time. We expect the property
\(p_{00}(\point_i) \neq 0\) to hold throughout the whole recursion when the
algorithm has been called on generic input.

Assuming the genericity property in
\cref{prop:evalintbasis:geometric_complexity}, and assuming the conditions are
met to apply this optimization, here is the impact on the cost of the first
basis extrapolation. For FFT points, this is a plain 25\% saving for the cost
of the first call to \textproc{BasisExtrapolate}: in the relevant part of the
proof of \cref{prop:evalintbasis:fft_complexity}, the bound \(9/2 \cdot d
\log_2(d) + \bigO{d}\) in \cref{eqn:evalintbasis:fftcomp:firstcall} is reduced
to \(3/4 \cdot 9/2 \cdot d \log_2(d) + \bigO{d}\), which is \(3/4 \cdot
\timepm{d/2} + \bigO{d}\).
For points in arithmetic or geometric progression, this amounts to replacing
\(\cstpolmulrep{4}\) by \(\cstpolmulrep{3}\) in the last paragraph of the proof
of \cref{prop:evalintbasis:geometric_complexity}: the cost reduces from
\(3\cstpolmulrep{4}/4 \cdot \timepm{\lceil d / 2
\rceil} + \bigO{d}\) to \(3\cstpolmulrep{3}/4 \cdot \timepm{\lceil d / 2
\rceil} + \bigO{d}\). The latter is a 25\% saving in the general case with
\(\cstpolmulrep{3}=3\) and \(\cstpolmulrep{4} = 4\),
but a slightly smaller gain when better methods are available, typically
with \(\cstpolmulrep{3} = 7/3\) and \(\cstpolmulrep{4} = 3\).

\subsubsection{Second basis extrapolation: exploiting known relations.}

The basis \(\secondbasis\) computed by the second recursive call at
\cref{algo:Eval-IntBasis2:secondcall} of \nameref{algo:Eval-IntBasis2}
satisfies, by definition of interpolants,
\[
  \left\{
    \begin{array}{l}
  q_{00}(\point_{\dreca+i}) \arev_i + q_{01}(\point_{\dreca+i}) \brev_i = 0 \\
  q_{10}(\point_{\dreca+i}) \arev_i + q_{11}(\point_{\dreca+i}) \brev_i = 0
\end{array}
\right.
  \quad\text{for } 0 \le i < \drecb
  .
\]
This call only returns \(\lenP_2\) evaluations
\(\secondbasis(\points_{d-\lenP_2:d})\), and the call to
\textproc{BasisExtrapolate} then finds further evaluations
\(\secondbasis(\points_{d-\lenP:d-\lenP_2})\). For all those evaluations that
concern a point \(\point_{\dreca+i}\) in \(\points_{\dreca:d-\lenP_2}\), one
may use the above relations to deduce \(q_{01}(\point_{\dreca+i})\) and
\(q_{11}(\point_{\dreca+i})\) from \(q_{00}(\point_{\dreca+i})\) and
\(q_{10}(\point_{\dreca+i})\), provided that \(\brev_i \neq 0\). Starting from
generic input, we expect these values \(\arev_i\) and \(\brev_i\) to be nonzero
throughout the whole recursion.

The approach would be as follows: if there are zero values in both \(\arevs\)
and \(\brevs\), then compute all extrapolations as in
\nameref{algo:Eval-IntBasis2} (or try a determinant-based optimization as
above); else if all values in \(\brevs\) are nonzero, extrapolate \(q_{00}\)
and \(q_{10}\) and use the above formulas to deduce the evaluations of
\(q_{01}\) and \(q_{11}\); else if all values in \(\arevs\) are nonzero,
extrapolate \(q_{01}\) and \(q_{11}\) and use similar formulas to deduce the
evaluations of \(q_{00}\) and \(q_{10}\). While the actual extrapolations, say
of \(q_{00}\) and \(q_{10}\), determine all sought values at
\(\points_{d-\lenP:d-\lenP_2}\), deducing the remaining evaluations of
\(q_{01}\) and \(q_{11}\) is only feasible for points in
\(\points_{\dreca:d-\lenP_2}\), since the used formulas are not valid for
points in \(\points_{0:\dreca}\). Fortunately, one has \(\lenP \le \lceil
(d+1)/2 \rceil \le \drecb + 1\) in the generic case, so that \(d - \lenP \ge d
- (\drecb+1) = \dreca - 1\). This means that the above procedure misses at most
one value for \(q_{01}\) and \(q_{11}\), specifically the one at
\(\point_{\dreca-1}\). This value can be computed in linear time using the
formulas from \cref{lem:arithmetic_progression,lem:geometric_progression}.

Assuming the genericity property in
\cref{prop:evalintbasis:geometric_complexity}, and assuming the conditions are
met to apply this optimization, here is the impact on the cost of the second
basis extrapolation. For FFT points, this is a plain 50\% saving: the bound
\(2/3 \cdot \timepm{d/2} + \bigO{d}\) at the end of the proof of
\cref{prop:evalintbasis:fft_complexity} becomes \(1/3 \cdot \timepm{d/2} +
\bigO{d}\). For points in arithmetic or geometric progression, the bound
\(\cstpolmulrep{4}/2 \cdot \timepm{\lceil d / 2 \rceil} + \bigO{d}\) in the
last paragraph of the proof of \cref{prop:evalintbasis:geometric_complexity}
reduces to \(\cstpolmulrep{2}/2 \cdot \timepm{\lceil d / 2 \rceil} +
\bigO{d}\). The latter is a 50\% saving in the general case with
\(\cstpolmulrep{2}=2\) and \(\cstpolmulrep{4} = 4\), but a slightly smaller
gain when better methods are available, typically with \(\cstpolmulrep{2} =
5/3\) and \(\cstpolmulrep{4} = 3\).

\subsubsection{Cost bound when combining both optimizations.}

Here, we derive cost bounds in the most favorable case where both optimizations
can be applied at all stages of the algorithm. For points in geometric or
arithmetic progression, we get the complexity equation \(\Cxity{d} =
2\Cxity{\lceil d/2 \rceil} + (3\cstpolmulrep{3} + 2\cstpolmulrep{2})/4 \cdot
\timepm{\lceil d/2 \rceil} + \bigO{d}\). This yields the leading constant
\((3\cstpolmulrep{3} + 2\cstpolmulrep{2})/8\), which is always less than or
equal to \(13/8\), and is \(31/24\) if \(\cstpolmulrep{k} = (2k+1)/3\) is
feasible. For FFT points, the equation is \(\Cxity{d} = 2\Cxity{d/2} + (3/4 +
1/3) \timepm{d/2} + \bigO{d}\). This yields the leading constant
\(13/24\).

% Fakesection acknowledgements
%\begin{acknowledgements}
%\end{acknowledgements}

% Fakesection bibliography
% \bibliographystyle{spmpsci}      % mathematics and physical sciences

\appendix

\section{A sharp bound for divide and conquer approaches with two calls in half degree}
\label{app:dnc_factor_half}

The goal of this appendix is to give a sharp upper bound on \(\sum_{i=0}^{k-1} 2^i \timepm{n/2^i}\)
when \(k\) is roughly \(\log_2(n)\). This quantity often arises when analyzing divide and conquer
algorithms for univariate polynomials. This is the case for algorithms studied in this paper,
but also for example for fast multipoint evaluation or interpolation at general points
\cite[Chap.\,10]{GathenGerhard2013}: these algorithms have two recursive calls, each operating on
polynomials whose degree is half of the input degree. We first recall and prove a statement
from \cref{sec:approx:polmul_midprod}.

\begin{lemma}
  \label{lem:nondecr_cst_factor}%
  Assuming \cref{hyp:timepm:nondecr}, \(\sum_{i=0}^{k-1} 2^i \timepm{n/2^i}
  \leq \frac{1}{2}\timepm{n}(k+1)\) holds for any \(n = 2^k \in \ZZp\).
\end{lemma}

\begin{proof}
  The cases \(n = 1\) and \(n=2\) are clear, so we consider \(n>2\). Then \(k\ge
  2\) and \hyptime{hyp:timepm:nondecr} implies, for all \(j \in \{0,\ldots,k-2\}\),
  \[
    \frac{\timepm{n/2^{j+1}}}{\frac{n}{2^{j+1}}\log_2(n/2^{j+1})}
    \leq
    \frac{\timepm{n/2^{j}}}{\frac{n}{2^{j}}\log_2(n/2^{j})},
    \text{ hence }\;
    2 \frac{\timepm{n/2^{j+1}}}{\timepm{n/2^{j}}}
    \leq
    \frac{\log_2(n/2^{j+1})}{\log_2(n/2^{j})}
    .
  \]
  For a given \(i \in \{0, \ldots, k-1\}\), multiplying over all \(j \in
  \{0,\ldots,i-1\}\) the latter inequalities (of positive numbers) yields
  \[
    2^i \frac{\timepm{n/2^{i}}}{\timepm{n}}
    \leq
    \frac{\log_2(n/2^{i})}{\log_2(n)}
    = 1 - \frac{i}{k}.
  \]
  Summing these inequalities concludes the proof:
  \[
    \sum_{i=0}^{k-1} 2^i \timepm{\frac{n}{2^i}}
   \leq \sum_{i=0}^{k-1} \timepm{n} \left(1-\frac{i}{k} \right)
           = \frac{1}{2}\timepm{n}(k+1).
           \qquad\qed
  \]
\end{proof}

Now, we recall and prove \cref{lem:nondecr_cst_factor_general}, which is a
similar, more general statement involving an arbitrary number \(n\).

\begin{lemma}
  Let \(n \in \ZZp\) and \(k = \lfloor \log_2(n) \rfloor\). If
  \hyptime{hyp:timepm:nondecr} holds, then
  \[
    \sum_{i=0}^{k-1} 2^i \timepm{\left\lceil \frac{n}{2^i}\right\rceil}
    \;\;\le\;\; {\textstyle\frac{1}{2}}\timepm{n}(k+5).
  \]
\end{lemma}

\begin{proof}
  The cases \(n \in \{1, 2, 3\}\) are easily verified; in what follows,
  \(n \ge 4\) and \(k \ge 2\).
  % NOTE some details:
  % \(n = 1\): \(0 \le 5\timepm{1}\) \\
  % \(n = 2\): \(\timepm{2} \le \frac{1}{2} \timepm{2} (1+5)\) \\
  % \(n = 3\): \(\timepm{3} \le \frac{1}{2} \timepm{3} (1+5)\) \\
  Following the beginning of the proof above, while incorporating suitable
  ceilings, leads to
  % NOTE some details:
  % for all \(j \in \{0,\ldots,k-2\}\),
  % \[
  %   \frac{\timepm{\lceil n/2^{j+1} \rceil}}{\lceil \frac{n}{2^{j+1}} \rceil\log_2(\lceil n/2^{j+1} \rceil)}
  %   \leq
  %   \frac{\timepm{\lceil n/2^{j} \rceil}}{\lceil \frac{n}{2^{j}} \rceil\log_2(\lceil n/2^{j} \rceil)},
  %   \text{ hence }\;
  %   \frac{\lceil \frac{n}{2^{j}} \rceil \timepm{\lceil n/2^{j+1} \rceil}}{\lceil \frac{n}{2^{j+1}} \rceil\timepm{\lceil n/2^{j} \rceil}}
  %   \leq
  %   \frac{\log_2(\lceil n/2^{j+1} \rceil)}{\log_2(\lceil n/2^{j} \rceil)}
  %   ,
  % \]
  % and multiplying these for \(j \in \{0,\ldots,i-1\}\), we deduce
  \[
    2^i \timepm{\lceil n/2^{i} \rceil}
    \leq
    \frac{\timepm{n}}{n\log_2(n)}
    2^i \left\lceil \frac{n}{2^{i}} \right\rceil\log_2\!\left(\left\lceil \frac{n}{2^{i}} \right\rceil\right)
    \le
    \frac{\timepm{n}}{n\log_2(n)}
    2^i \left( \frac{n}{2^{i}} + 1 \right)\log_2\!\left( \frac{n}{2^{i}} + 1 \right),
  \]
  for all \(i \in \{0, \ldots, k-1\}\). Now, since the real-valued function
  \(\lambda \in [1,+\infty) \mapsto \lambda \log_2(\lambda)\) is convex,
  \[
    \left( \frac{n}{2^{i}} + 1 \right) \log_2\!\left( \frac{n}{2^{i}} + 1\right)
    \le
    \frac{1}{2}
    \left( \frac{n}{2^{i-1}} \log_2\!\left( \frac{n}{2^{i-1}}\right)
    + 2 \log_2(2) \right)
    =
    1 + n (\log_2(n) + 1 - i) / 2^{i} .
  \]
  Combining this with the previous inequality and summing over \(i \in \{0, \ldots, k-1\}\), we get
  \begin{align*}
    \sum_{i = 0}^{k-1} 2^i \timepm{\lceil n/2^{i} \rceil}
    & \leq \frac{\timepm{n}}{n\log_2(n)} \left( \sum_{i = 0}^{k-1} 2^i + n (\log_2(n) + 1 - i)\right) \\
    & = \frac{\timepm{n}}{n\log_2(n)} \left( 2^k - 1 + n \left((\log_2(n) + 1) k - \frac{k(k-1)}{2}\right)\right).
  \end{align*}
  Using \(k-1 \ge \log_2(n) - 2 \ge 0\), we get \((\log_2(n) + 1) k -
  \frac{k(k-1)}{2} \le k (2 + \log_2(n) / 2)\), and on the other hand we have
  \(2^k - 1 < n\). Hence
  \[
    \frac{1}{n\log_2(n)} \left( 2^k - 1 + n \left((\log_2(n) + 1) k - \frac{k(k-1)}{2}\right)\right)
    \le \frac{1}{\log_2(n)} + \frac{k}{2} + 2
    \le \frac{1}{2} (k+5),
  \]
  and the conclusion follows.
  \qed
\end{proof}

\section{Degrees of approximant bases and interpolant bases for generic input}
\label{app:approx:genericity}

For any integers \(d \in \NN\) and \(\ash \in \ZZ\), define
the pair \(\delta_{d,\ash} \in \{0,1,\ldots,d\}^2\) as
\begin{equation}
  \label{eqn:generic_mindeg}%
  \delta_{d,\ash}
  =
  \left\{
    \begin{array}{ll}
      (d,0) & \text{if } \ash \le -d, \\
      (0,d) & \text{if } \ash \ge d, \\
      (\lceil (d-\ash)/2 \rceil, \lfloor (d+\ash)/2 \rfloor) & \text{if } -d \le \ash \le d.
    \end{array}
  \right.
  %% NOTE
  %%  \delta_0 = \max(0,\min(d,\lceil (d-\ash)/2 \rceil))
  %%  \text{ and }
  %%  \delta_1 = \max(0,\min(d,\lfloor (d+\ash)/2 \rfloor))
\end{equation}
The following theorem states that, generically, modules of approximants at
order \(d\) and modules of interpolants at \(d\) distinct points have
\(\ash\)-minimal degree \(\delta_{d,\ash}\).

\begin{theorem}
  \label{thm:generic_mindeg}%
  Consider an integer \(d \in \NN\) and a shift \(\ash \in \ZZ\).
  \begin{itemize}[noitemsep,topsep=2pt]
    \item There exists a polynomial \(\Delta_{d,\ash}\) in \(2d\) variables
      over \(\field\), nonzero and of total degree at most \(d\), such that for
      all polynomials \(\apo =
      \apo_0 + \apo_1 x + \cdots + \apo_{d-1} x^{d-1}\) and \(\bpo = \bpo_0 +
      \bpo_1 x + \cdots + \bpo_{d-1} x^{d-1}\) in \(\pR\), the \(\pR\)-module
      \(\appmod{d}{\apo,\bpo}\) has \(\ash\)-minimal degree \(\delta_{d,\ash}\)
      if and only if
      \(\Delta_{d,\ash}(\apo_0,\ldots,\apo_{d-1},\bpo_0,\ldots,\bpo_{d-1}) \neq 0\).
    \item Let \(\points\in\field^d\) be pairwise distinct. There exists a polynomial
      \(\Gamma_{\points,\ash}\) in \(2d\) variables over \(\field\), nonzero and
      of total degree at most \(d\), such that for all tuples \(\aevs =
      (\aev_0, \aev_1, \ldots, \aev_{d-1})\) and \(\bevs = (\bev_0,
      \bev_1, \ldots, \bev_{d-1})\) in \(\field^d\), the \(\pR\)-module
      \(\intmod{\points}{\aevs,\bevs}\) has \(\ash\)-minimal degree \(\delta_{d,\ash}\)
      if and only if
      \(\Gamma_{\points,\ash}(\aev_0,\ldots,\aev_{d-1},\bev_0,\ldots,\bev_{d-1}) \neq
      0\).
  \end{itemize}
\end{theorem}

We prove this result in two steps: in \cref{sec:approx:genericity:existence},
we show the existence of such polynomials \(\Delta_{d,\ash}\) and
\(\Gamma_{\points,\ash}\) with all the above properties barring nonzeroness;
before that, in \cref{sec:approx:genericity:nonzero} we prove that any such
\(\Delta_{d,\ash}\) or \(\Gamma_{\points,\ash}\) must be nonzero by giving an
explicit point in \(\field^{2d}\)
% \((\apo_0,\ldots,\apo_{d-1},\bpo_0,\ldots,\bpo_{d-1}) \in \field^{2d}\)
where it does not vanish.
While we give all details concerning \(\Delta_{d,\ash}\), the case of
\(\Gamma_{\points,\ash}\) is only sketched for conciseness, as the arguments
are very similar.
Before turning to this proof, we state basic properties of \(\delta_{d,\ash}\)
and a consequence of the above theorem. First, observe that the two entries of
\(\delta_{d,\ash}\) sum to \(d\) in all cases.
%%% NOTE keep for us in case
%% Let \((\delta_0,\delta_1)\) be the pair \(\delta_{d,\ash}\).
%% i.e.\ \(\delta_0 + \delta_1 = d\). Furthermore, when \(-d \le \ash \le d\), one
%% has \(\delta_{d+1,\ash} = (\delta_0,\delta_1+1)\) if \(d-\ash\) is
%% odd, while \(\delta_{d+1,\ash} = (\delta_0+1,\delta_1)\) if
%% \(d-\ash\) is even.
The next lemma highlights a compatibility between \(\delta_{d,\ash}\) and the
general recursive approach outlined in \cref{lem:dnc_relbas}.

\begin{lemma}
  \label{lem:approx_generic_mindeg}
  Consider integers \(\dreca, \drecb \in \NN\) and a shift \(\ash \in
  \ZZ\). Write \((\delta_0,\delta_1) = \delta_{\dreca,\ash}\) and define
  \(\ashr = \ash+\delta_0 - \delta_1\). Then
  \(\delta_{\dreca+\drecb,\ash} = \delta_{\dreca,\ash} +
  \delta_{\drecb,\ashr}\), where the sum is entry-wise.
\end{lemma}

\begin{proof}
  %% NOTE This is a consequence of \cref{thm:generic_mindeg,lem:dnc_relbas}, but
  %% can also be proved concisely, as follows.
  Suppose first \(\ash \le -(\dreca+\drecb)\). By definition,
  \(\delta_{\dreca+\drecb,\ash} = (\dreca+\drecb,0)\) and
  \(\delta_{\dreca,\ash} = (\dreca,0)\). The latter implies \(\ashr
  = \ash+\delta_0-\delta_1 = \ash + \dreca \le -\drecb\), which
  ensures \(\delta_{\drecb,\ashr} = (\drecb,0)\) by definition again;
  this proves the sought identity. The proof is similar for the case
  \(\ash \ge \dreca+\drecb\). Now, in the rest of this proof, assume
  \(- (\dreca+\drecb) \le \ash \le \dreca+\drecb\). Observe that, since
  \(\delta_0+\delta_1 = \dreca\), we have \(\dreca + \drecb - \ash =
  \drecb - \ashr + 2\delta_0\), hence \(\lceil (\dreca + \drecb - \ash) / 2 \rceil
  = \lceil (\drecb - \ashr) / 2 \rceil + \delta_0\);
  similarly, \(\lfloor (\dreca + \drecb + \ash) / 2 \rfloor = \lfloor
  (\drecb + \ashr) / 2 \rfloor + \delta_1\). This shows the identity in
  the case \(-\drecb \le \ashr \le \drecb\). To conclude the proof, we
  show that our current assumptions forbid both \(\ashr < -\drecb\) and
  \(\ashr > \drecb\). Suppose by contradiction \(\ashr < -\drecb\), and
  further, \(\ash \ge -\dreca\). Since \(\ashr = \ash + \delta_0 - \delta_1 \le
  0\), we have \(\ash \le \delta_1 \le \dreca\). Thus \(-\dreca \le \ash \le \dreca\),
  hence
  \[
  \ashr = \ash + \lceil (d-\ash)/2 \rceil - \lfloor (d+\ash)/2 \rfloor
  = \lceil (d+\ash)/2 \rceil - \lfloor (d+\ash)/2 \rfloor
  \in \{0 , 1\}.
  \]
  Yet, having \(\ashr \in \{0 , 1\}\) is absurd since \(\ashr < -\drecb \le 0\).
  Therefore \(\ash < -\dreca\). So
  \((\delta_0,\delta_1) = (\dreca,0)\) and \(\ash = \ashr +
  \delta_1 - \delta_0 = \ashr - \dreca < -(\dreca+\drecb)\), which is
  again absurd. Therefore \(\ashr \ge -\drecb\). The same reasoning
  shows that \(\ashr \le \drecb\).
  \qed
\end{proof}

\begin{corollary}
  \label{cor:approx_generic_mindeg}%
  Consider an integer \(d \in \NN\) and a shift \(\ash \in \ZZ\). There exists
  a polynomial \(\bar\Delta_{d,\ash}\) in \(2d\) variables over \(\field\),
  nonzero and of total degree at most \(d(d+1)/2\), such that for all
  \(\apo = \apo_0 + \apo_1 x + \cdots + \apo_{d-1} x^{d-1}\) and \(\bpo =
  \bpo_0 + \bpo_1 x + \cdots + \bpo_{d-1} x^{d-1}\) in \(\pR\), if
  \(\bar{\Delta}_{d,\ash}(\apo_0,\ldots,\apo_{d-1},\bpo_0,\ldots,\bpo_{d-1})
  \neq 0\) then for all \(0 \le k \le d\) the \(\pR\)-module
  \(\appmod{k}{\apo,\bpo}\) has \(\ash\)-minimal degree
  \(\delta_{k,\ash}\).

  Assume
  \(\bar{\Delta}_{d,\ash}(\apo_0,\ldots,\apo_{d-1},\bpo_0,\ldots,\bpo_{d-1})
  \neq 0\). Let \(0 \le \dreca \le d\), let \(
  [\begin{smallmatrix}
    p_{00} & p_{01} \\
    p_{10} & p_{11}
  \end{smallmatrix}]
  \in \mpR{2}{2} \) be an \(\ash\)-weak Popov basis of
  \(\appmod{\dreca}{\apo,\bpo}\), and let \(\ashr = \ash+\deg(p_{00}) -
  \deg(p_{11})\). Then, for any \(0 \le \drecb \le d-\dreca\),
  the polynomials \(\apor = (x^{-\dreca}(p_{00} \apo + p_{01} \bpo)) \rem
  x^{\drecb}\) and \(\bpor = (x^{-\dreca}(p_{10} \apo + p_{11} \bpo) \rem
  x^{\drecb}\) are such that for all \(0 \le k \le \drecb\) the \(\pR\)-module
  \(\appmod{k}{\apor,\bpor}\) has \(\ashr\)-minimal degree
  \(\delta_{k,\ashr} = \delta_{\dreca+k,\ash} - \delta_{\dreca,\ash}\).
\end{corollary}

\begin{proof}
  The first claim follows from \cref{thm:generic_mindeg} by taking for
  \(\bar\Delta_{d,\ash}\) the product of the polynomials
  \(\Delta_{k,\ash}\) for \(0\le k \le d\). The second claim follows from
  the properties of the general recursive approach. Indeed,
  \cref{lem:dnc_relbas} ensures that the \(\ash\)-minimal degree of
  \(\appmod{\dreca+k}{\apo,\bpo}\) (which is \(\delta_{\dreca+k,\ash}\) by
  assumption) is the entry-wise sum of the \(\ash\)-minimal degree of
  \(\appmod{\dreca}{\apo,\bpo}\) (which is \(\delta_{\dreca,\ash}\) by
  assumption) and of the \(\ashr\)-minimal degree of
  \(\appmod{k}{\apor,\bpor}\). So, the latter is \(\delta_{\dreca+k,\ash}
  - \delta_{\dreca,\ash}\), which is \(\delta_{k,\ashr}\) according
  to \cref{lem:approx_generic_mindeg}.
  \qed
\end{proof}

The following similar result holds for modules of interpolants.

\begin{corollary}
  \label{cor:interp_generic_mindeg}%
  Let \(\points\in\field^d\) be pairwise distinct and let \(\ash \in \ZZ\). There exists
  a polynomial \(\bar\Gamma_{\points,\ash}\) in \(2d\) variables over \(\field\),
  nonzero and of total degree at most \(d(d+1)/2\), such that for all tuples
  \(\aevs = (\aev_0, \aev_1, \ldots, \aev_{d-1})\) and \(\bevs =
  (\bev_0, \bev_1, \ldots, \bev_{d-1})\) in \(\field^d\), if
  \(\bar{\Gamma}_{\points,\ash}(\apo_0,\ldots,\apo_{d-1},\bev_0,\ldots,\bev_{d-1})
  \neq 0\) then for all \(0 \le k \le d\) the \(\pR\)-module
  \(\intmod{\points_{0:k}}{\aevs_{0:k},\bevs_{0:k}}\) has \(\ash\)-minimal degree
  \(\delta_{k,\ash}\).

  Assume
  \(\bar{\Gamma}_{\points,\ash}(\aev_0,\ldots,\aev_{d-1},\bev_0,\ldots,\bev_{d-1})
  \neq 0\). Let \(0 \le \dreca \le d\), let \(
  [\begin{smallmatrix}
    p_{00} & p_{01} \\
    p_{10} & p_{11}
  \end{smallmatrix}]
  \in \mpR{2}{2}\) be an \(\ash\)-weak Popov basis of
  \(\intmod{\points_{0:\dreca}}{\aevs_{0:\dreca},\bevs_{0:\dreca}}\), and let \(\ashr =
  \ash+\deg(p_{00}) - \deg(p_{11})\). For any \(0 \le \drecb \le
  d-\dreca\), the tuples \(\arevs = (p_{00}(\point_{i}) \aev_{i} + p_{01}(\point_{i})
  \bev_{i})_{\dreca \le i < \dreca+\drecb}\)
  and \(\brevs = (p_{10}(\point_{i}) \aev_{i} + p_{11}(\point_{i})
  \bev_{i})_{\dreca \le i < \dreca+\drecb}\)
  are such that, for all \(0 \le k \le \drecb\),
  % the \(\pR\)-module
  \(\intmod{\points_{\dreca:\dreca+k}}{\arevs_{0:k},\brevs_{0:k}}\) has \(\ashr\)-minimal degree
  \(\delta_{k,\ashr} = \delta_{\dreca+k,\ash} - \delta_{\dreca,\ash}\).
\end{corollary}

\subsection{If \texorpdfstring{\(\Delta_{d,\ash}\)}{Delta[d,s]}
or \texorpdfstring{\(\Gamma_{\points,\ash}\)}{Gamma[w,s]}
characterizes the minimal degree
\texorpdfstring{\(\delta_{d,\ash}\)}{delta[d,s]}, it is nonzero}
\label{sec:approx:genericity:nonzero}

By ``characterizing the minimal degree'', we mean more precisely that
\(\Delta_{d,\ash}\) is any polynomial in \(2d\) variables over \(\field\) which
vanishes at the coefficients of \(\apo,\bpo \in \pR_{<d}\) exactly when the
\(\ash\)-minimal degree of \(\appmod{d}{\apo,\bpo}\) differs from
\(\delta_{d,\ash}\). For the moment, we discard the question of the existence
of such a polynomial \(\Delta_{d,\ash}\): we simply prove that any polynomial
with this property must be nonzero. For this, we explicitly describe some
\(\apo,\bpo \in \pR_{<d}\) such that \(\appmod{d}{\apo,\bpo}\) has
\(\ash\)-minimal degree \(\delta_{d,\ash}\). Thus the coefficients of \(\apo\)
and \(\bpo\) provide an evaluation point in \(\field^{2d}\) at which
\(\Delta_{d,\ash}\) does not vanish. We detail this in the next lemma, and then
we sketch the analogous case of \(\Gamma_{\points,\ash}\) for interpolant
modules \(\intmod{\points}{\aevs,\bevs}\) in
\cref{lem:interp_generic_mindeg:nonzero}.

\begin{lemma}
  \label{lem:approx_generic_mindeg:nonzero}%
  Let \(d \in \NN\) and \(\ash \in \ZZ\). Write \((\delta_0,\delta_1)
  = \delta_{d,\ash}\). Define \(\apo,\bpo \in \pR_{<d}\) as
  \[
    (\apo,\bpo) =
    \left\{
    \begin{array}{ll}
      (1,0) & \text{if } \ash \le -d, \\
      (0,1) & \text{if } \ash \ge d, \\
      (1, x^{\delta_0}) & \text{if } {-d} < \ash \le 0, \\
      (x^{\delta_1}, 1) & \text{if } 0 < \ash < d.\\
    \end{array}
    \right.
  \]
  Then, \(\appmod{d}{\apo,\bpo}\) has \(\ash\)-minimal degree
  \(\delta_{d,\ash}\).
\end{lemma}

\begin{proof}
  \emph{Case} \(\ash \le -d\). With our choice \((\apo,\bpo) = (1,0)\), it
  is easily proved that \(\appmod{d}{\apo,\bpo} = x^d\pR \times \pR\). This
  module has shifted Popov basis \([\begin{smallmatrix} x^d & 0 \\ 0 & 1
  \end{smallmatrix}]\) for any shift, including \(\ash\). Thus
  \(\appmod{d}{\apo,\bpo}\) has \(\ash\)-minimal degree \((d,0) =
  \delta_{d,\ash}\).

  \emph{Case} \(\ash \ge d\). This is similar to the previous case, with
    here \(\appmod{d}{\apo,\bpo} = \pR \times x^d\pR\), shifted Popov basis
    \([\begin{smallmatrix} 1 & 0 \\ 0 & x^d
  \end{smallmatrix}]\), and \(\ash\)-minimal degree \((0,d) =
  \delta_{d,\ash}\).

  \emph{Case} \(-d < \ash \le 0\). Here, \((\apo,\bpo) = (1,
  x^{\delta_0})\). It is known that \(\apo_0 \neq 0\) implies that all bases of
  \(\appmod{d}{\apo,\bpo}\) have determinantal degree \(d\): they all share the
  same determinantal degree, and for example one can easily prove that
  \(
  [\begin{smallmatrix}
    x^{d} & 0 \\
    -x^{\delta_0} & 1
  \end{smallmatrix}]
  \)
  is a basis of \(\appmod{d}{\apo,\bpo}\) (see
  \cref{ex:approx_generic:unbalanced_shift} for related considerations
  for more general \(\apo,\bpo\)). It follows that the matrix
  \(
  P =
  [\begin{smallmatrix}
    x^{\delta_0} & -1 \\
    0 & x^{\delta_1}
  \end{smallmatrix}]
  \)
  is a basis of \(\appmod{d}{\apo,\bpo}\), since its rows are in this module
  and \(\det(\det(P)) = \delta_0 + \delta_1 = d\). Thus, to obtain that the
  \(\ash\)-minimal degree of
  \(\appmod{d}{\apo,\bpo}\) is \((\delta_0,\delta_1) =
  \delta_{d,\ash}\), it suffices to prove that \(P\) is in
  \(\ash\)-weak Popov form, that is, the \(\ash\)-pivots of
  \(P\) are on its diagonal. Obviously the \(\ash\)-pivots of the
  second row \([0 \;\; x^{\delta_1}]\) is on the diagonal, so it remains
  to show that \(\deg(x^{\delta_0}) + \ash > \deg(-1)\), or equivalently
  \(\delta_0 + \ash > 0\). This follows from the definition of
  \(\delta_0 = \lceil (d-\ash)/2 \rceil\) and from the assumption
  \(d > -\ash\):
  \[
    \delta_0 + \ash
    \ge \frac{d-\ash}{2} + \ash
    = \frac{d + \ash}{2}
    > 0.
  \]
  As a side remark, \(P\) is in fact the \(\ash\)-Popov basis of
  \(\appmod{d}{\apo,\bpo}\) except in the case where \(\delta_1=0\),
  i.e., when \(d + \ash = 1\); in the latter case, \(\delta_0=d\) and
  the \(\ash\)-Popov form of \(P\) is simply
  \(
    [\begin{smallmatrix}
      x^{d} & 0 \\
      0 & 1
      \end{smallmatrix}]
  \).

  \emph{Case} \(0 < \ash < d\). Here, \((\apo,\bpo) = (x^{\delta_1},
  1)\). This is similar to the above case, but not exactly symmetrical because
  \(\delta_0\) involves a ceiling while \(\delta_1\) involves a floor (this
  comes from the tie breaking choice in the definition of pivots when degrees
  are equal). Following the same arguments as above, one can prove that
  \(P =
    [\begin{smallmatrix}
      x^{\delta_0} & 0 \\
      -1 & x^{\delta_1}
    \end{smallmatrix}]
  \)
  is a basis of \(\appmod{d}{\apo,\bpo}\), and it remains to prove that it is
  in \(\ash\)-weak Popov form. The latter is equivalent to
  \(\deg(x^{\delta_1}) \ge \deg(-1) + \ash\), i.e., \(\delta_1 - \ash \ge
  0\). As above, this follows from the definition of \(\delta_1 = \lfloor
  (d+\ash)/2 \rfloor\) and from the assumption \(\ash < d\):
  \[
    \delta_1 - \ash
    \ge \frac{d+\ash}{2} - \frac{1}{2} - \ash
    = \frac{d - \ash - 1}{2}
    \ge 0.
  \]
  The side remark also extends: here, \(\delta_0 > 0\) and \(P\) is in fact the
  \(\ash\)-Popov basis of \(\appmod{d}{\apo,\bpo}\).
  \qed
\end{proof}

\begin{lemma}
  \label{lem:interp_generic_mindeg:nonzero}%
  Let \(\points \in \field^d\) be pairwise distinct points and let \(\ash \in
  \ZZ\). Write \((\delta_0,\delta_1) = \delta_{d,\ash}\) and let
  \(\modulus_0 = \prod_{0 \le i < \delta_0} (x - \point_i)\)
  and \(\modulus_1 = \prod_{\delta_0 \le i < d} (x - \point_i)\).
  One has \(\deg(\modulus_1) = \delta_1\), since as noted above
  \(\delta_0 + \delta_1 = d\). Define  polynomials \(\apo,\bpo \in \pR_{<d}\)
  and their evaluations \(\aevs,\bevs \in \field^{d}\) as
  \[
    \left\{
      \begin{array}{l}
        \aevs = \evaluate(\points,\apo) \\
        \bevs = \evaluate(\points,\bpo) \\
      \end{array}
    \right.
    \;\;\;\text{where}\;\;\;
    (\apo,\bpo) =
    \left\{
    \begin{array}{ll}
      (1,0) & \text{if } \ash \le -d, \\
      (0,1) & \text{if } \ash \ge d, \\
      (1, \modulus_0) & \text{if } {-d} < \ash \le 0, \\
      (\modulus_1, 1) & \text{if } 0 < \ash < d.\\
    \end{array}
    \right.
  \]
  Then, \(\intmod{\points}{\aevs,\bevs}\) has \(\ash\)-minimal degree
  \(\delta_{d,\ash}\).
\end{lemma}

\begin{proof}
  We only sketch the proof, as it is very close to the one of
  \cref{lem:approx_generic_mindeg:nonzero}. If \(\ash \le -d\) (resp.\
  \(\ash \ge d\)), then our choice of \((\aevs,\bevs)\) ensures that
  \(\intmod{\points}{\aevs,\bevs}\) is \(\modulus\pR \times \pR\) (resp.\
  \(\pR \times \modulus\pR\)),
  where \(\modulus = \prod_{0 \le i < d} (x - \point_i)\).
  If \(-d < \ash \le 0\), one shows that
  \(
    P =
    [\begin{smallmatrix}
      \modulus_0 & -1 \\
      0 & \modulus_1
    \end{smallmatrix}]
  \)
  is an \(\ash\)-weak Popov basis of \(\intmod{\points}{\aevs,\bevs}\);
  and if \(0 < \ash < d\),
  \(P =
    [\begin{smallmatrix}
      \modulus_0 & 0 \\
      -1 & \modulus_1
    \end{smallmatrix}]
  \)
  is a \(\ash\)-weak Popov basis of \(\intmod{\points}{\aevs,\bevs}\).
  \qed
\end{proof}

\subsection{Existence of polynomials
\texorpdfstring{\(\Delta_{d,\ash}\)}{Delta[d,s]}
and \texorpdfstring{\(\Gamma_{\points,\ash}\)}{Gamma[w,s]}
characterizing the minimal degree \texorpdfstring{\(\delta_{d,\ash}\)}{delta[d,s]}}
\label{sec:approx:genericity:existence}

As mentioned above, we give details for \(\Delta_{d,\ash}\) and approximant
modules, and then we briefly sketch how this adapts to interpolant modules.
Fix polynomials \(\apo,\bpo \in \pR_{<d}\). From a linear algebra viewpoint,
approximants in \(\appmod{d}{\apo,\bpo}\) form
\(\field\)-linear combinations that are zero between the polynomials
\(\apo\), \(\bpo\), \(x\apo\), \(x\bpo\), \(x^2 \apo\), \(x^2 \bpo\), etc.,
all truncated modulo \(x^d\).
Thus, restricting to approximants of degree at most \(d\), these can be found
in the left nullspace of either of the matrices
\begin{equation}
  \label{eqn:krylov_popov_hermite}%
  \mathcal{K} =
  \begin{bmatrix}
    \apo_0 & \apo_1 & \cdots         & \apo_{d-1}     \\
    \bpo_0 & \bpo_1 & \cdots         & \bpo_{d-1}     \\
    0      & \apo_0 & \cdots         & \apo_{d-2}     \\
    0      & \bpo_0 & \cdots & \bpo_{d-2}     \\
    0      &  0     & \ddots & \vdots \\
    0      &  0     & 0              & \apo_0         \\
    0      &  0     & 0              & \bpo_0         \\
    0      &  0     & 0              & 0              \\
    0      &  0     & 0              & 0
  \end{bmatrix}
  \quad\text{and}\quad
  \overline{\mathcal{K}} =
  \begin{bmatrix}
    \apo_0 & \apo_1 & \cdots         & \apo_{d-1}     \\
    0      & \apo_0 & \cdots         & \apo_{d-2}     \\
    0      &  0     & \ddots & \vdots \\
    0      &  0     & 0              & \apo_0         \\
    0      &  0     & 0              & 0              \\
    \bpo_0 & \bpo_1 & \cdots         & \bpo_{d-1}     \\
    0      & \bpo_0 & \cdots & \bpo_{d-2}     \\
    0      &  0     & \ddots & \vdots \\
    0      &  0     & 0              & \bpo_0         \\
    0      &  0     & 0              & 0
  \end{bmatrix}
  \quad\text{both in } \mKK{2(d+1)}{d}.
\end{equation}
(Why we define two matrices will be made clear below.)

This point of view has been used in algorithms for finding bases of vector
interpolation modules which generalize the relation modules
\(\relmod{\modulus}{\apo,\bpo}\) of \cref{eqn:relmod}, thus covering in
particular \cref{pbm:app,pbm:int}. This was done notably in contexts of
fraction-free computations \cite{BeckermannLabahn2000} and for cases where
parameters lead to favor linear algebra tools over polynomial arithmetic
\cite{JeannerodNeigerSchostVillard2017}. These references describe a
\emph{(striped) Krylov matrix} \(\mathcal{K}_{d,\ash}(\apo,\bpo)\)
(\cite[Sec.\,2]{BeckermannLabahn2000} and
\cite[Sec.\,7.1]{JeannerodNeigerSchostVillard2017}) whose rank properties yield
the minimal degree of the relation module, and they show how the shifted Popov
basis of this module can be recovered from a well-chosen basis of the nullspace
of this Krylov matrix (\cite[Thm.\,7.3]{BeckermannLabahn2000} and
\cite[Sec.\,7.2 and 7.3]{JeannerodNeigerSchostVillard2017}).
In \cite[Sec.\,4]{BeckermannLabahn2000}, Beckermann and Labahn consider
regularity assumptions based on ranks and determinants of Krylov matrices;
thanks to the above link between ranks and minimal degree, these considerations
are closely related to the genericity properties of
\cref{thm:generic_mindeg}. Yet, we cannot apply the results from
\cite{BeckermannLabahn2000} as such, as they require a strong link between the
shift \(\ash\) and the order \(d\), which may not hold in our case. For
this reason, we will rely on the construction of the Krylov matrix in
\cite[Sec.\,7.1]{JeannerodNeigerSchostVillard2017}, which does not assume this
link, and we will briefly explain how some Krylov matrix determinant yields the
polynomial \(\Delta_{d,\ash}\).

\begin{lemma}
  \label{lem:approx_generic_mindeg_exists}%
  There exists a polynomial \(\Delta_{d,\ash}\) in \(2d\) variables
  \((X_0,\ldots,X_{d-1},Y_0,\ldots,Y_{d-1})\) over \(\field\), either zero or
  of total degree exactly \(d\), such that for any \(\apo,\bpo \in \pR\), the
  \(\pR\)-module \(\appmod{d}{\apo,\bpo}\) has \(\ash\)-minimal degree
  \(\delta_{d,\ash}\) if and only if
  \(\Delta_{d,\ash}(\apo_0,\ldots,\apo_{d-1},\bpo_0,\ldots,\bpo_{d-1})
  \neq 0\).
\end{lemma}

\begin{proof}
  Following \cite[Sec.\,7.1]{JeannerodNeigerSchostVillard2017}, the Krylov
  matrix \(\mathcal{K}_{d,\ash}(\apo,\bpo) \in \mKK{2(d+1)}{d}\) is built
  as a row permutation of the matrix \(\mathcal{K}\) in
  \cref{eqn:krylov_popov_hermite}, with a permutation derived from the shift
  \(\ash\). The construction of this matrix depends on
  \(d,\ash\) but is independent from the particular coefficients of
  \(\apo\) and \(\bpo\): in other words, if one follows the same construction
  for generic polynomials \(X = \sum_{i < d} X_i x^i\) and \(Y = \sum_{i < d}
  Y_i x^i\) to obtain a matrix \(\mathcal{K}_{d,\ash}(X,Y)\), then
  replacing \(X,Y\) by any particular \(\apo,\bpo \in \pR\) will indeed provide
  \(\mathcal{K}_{d,\ash}(\apo,\bpo)\).
  We then let \(D\) be the \(d\times d\) submatrix of
  \(\mathcal{K}_{d,\ash}(X,Y)\) formed by its top \(d\) rows. One
  important feature of this construction is that \(D\) contains \(\delta_0\)
  rows derived from \(X\) (or the underlying \(\apo\)), and \(\delta_1\) rows
  derived from \(Y\) (or the underlying \(\bpo\)), where \(\delta_0\) and
  \(\delta_1\) are such that \(\delta_{d,\ash} = (\delta_0,\delta_1)\). We
  take \(\Delta_{d,\ash}\) as the determinant of \(D\). Since \(D\) has
  each of its entries in \(\{0,X_0,\ldots,X_{d-1},Y_0,\ldots,Y_{d-1}\}\), the
  polynomial \(\Delta_{d,\ash}\) is either zero, or homogeneous of total
  degree \(d\). Furthermore,
  \(\Delta_{d,\ash}(\apo_0,\ldots,\apo_{d-1},\bpo_0,\ldots,\bpo_{d-1})\)
  is nonzero if and only if the top \(d\times d\) submatrix of
  \(\mathcal{K}_{d,\ash}(\apo,\bpo)\) is invertible, i.e.\
  \(\mathcal{K}_{d,\ash}(\apo,\bpo)\) has generic row rank profile
  \((0,1,\ldots,d-1)\). The latter holds if and only if the minimal exponents
  \(\mu_0\) and \(\mu_1\) such that the row corresponding to \(x^{\mu_0} \apo\)
  (resp.\ \(x^{\mu_1} \bpo\)) in \(\mathcal{K}_{d,\ash}(\apo,\bpo)\) is a
  \(\field\)-linear combination of rows above it, are precisely \(\mu_0 =
  \delta_0\) and \(\mu_1 = \delta_1\). The conclusion then follows from the
  fact that these minimal exponents \((\mu_0,\mu_1)\) (also called minimal
  degree in \cite[Def.\,7.8]{JeannerodNeigerSchostVillard2017}) form precisely
  the \(\ash\)-minimal degree of \(\appmod{d}{\apo,\bpo}\), according to
  \cite[Cor.\,7.11 and Lem.\,7.12]{JeannerodNeigerSchostVillard2017}.
  \qed
\end{proof}

We now illustrate the above considerations through some examples. The first one
is the corner case \(d=0\), where one can easily derive
\(\Delta_{0,\ash}\) without considering Krylov matrices.

\begin{example}[order \(d\) is zero]
  \label{ex:approx_generic:zero_order}%
  Here, take an arbitrary shift \(\ash\). One has
  \(\appmod{0}{\apo,\bpo} = \pR^2\), so that its \(\ash\)-Popov basis is
  the \(2\times 2\) identity matrix for any \(\apo,\bpo\). Then the
  \(\ash\)-minimal degree of \(\appmod{0}{\apo,\bpo}\) is \((0,0) =
  \delta_{0,\ash}\), and one may take \(\Delta_{0,\ash} = 1\).
  \qed
\end{example}

The next two examples consider situations representing extremal values for the
shift \(\ash\): when it is highly unbalanced such as \(\ash = d\)
and leads to the Hermite normal form, and when it is uniform \(\ash = 0\)
and leads to the ``nonshifted'' Popov form. Finally, we consider a
case involving a non-extremal shift, illustrating the construction of the
Krylov matrix in that case.

\begin{example}[highly unbalanced shift]
  \label{ex:approx_generic:unbalanced_shift}%
  Assume here that \(\ash \le -d\), with \(d > 0\). Let \((\apo,\bpo) \in \pR^2\). The
  construction of the Krylov matrix leads to
  \(\mathcal{K}_{d,\ash}(\apo,\bpo) = \overline{\mathcal{K}}\), where the
  latter matrix is defined in \cref{eqn:krylov_popov_hermite}. In this case,
  the matrix \(D\) in the proof of \cref{lem:approx_generic_mindeg_exists} is
  the upper triangular \(d \times d\) Toeplitz matrix whose first row is \([X_0
  \;\; X_1 \;\; \cdots \;\; X_{d-2} \;\; X_{d-1}]\); it has determinant
  \(\Delta_{d,\ash} = X_0^d\). Thus, the \(\ash\)-minimal
  degree of the module of approximants is \((d,0) = \delta_{d,\ash}\) if
  and only if \(a_0 \neq 0\).

  If \(\apo_0 \neq 0\), one can derive the \(\ash\)-Popov approximant
  basis explicitly in terms of \(\apo\) and \(\bpo\). Let \(c = (\bpo \apo^{-1}) \rem x^d\)
  be the truncation at order \(d\) of the power series expansion of \(b
  a^{-1}\). Define the matrix \(P = [\begin{smallmatrix} x^d & 0 \\ -c & 1
  \end{smallmatrix}]\). Since \(\ash \le -d\),
  \(P\) is in \(\ash\)-Popov form (and, as a matter of fact, also in
  HNF). Furthermore, one easily verifies that each row of \(P\) is in
  \(\appmod{d}{\apo,\bpo}\), and that any element of this module is a
  \(\pR\)-linear combination of the rows of \(P\). Hence \(P\) is the
  \(\ash\)-Popov basis of \(\appmod{d}{\apo,\bpo}\), which confirms that
  the latter has \(\ash\)-minimal degree \((d,0)\).
  \qed
\end{example}

\begin{example}[uniform shift]
  Assume here that \(\ash=0\) and \(d > 0\).
  Let \((\apo,\bpo) \in \pR_{< d}^2\). The construction
  of the Krylov matrix leads to \(\mathcal{K}_{d,\ash}(\apo,\bpo) =
  \mathcal{K}\), where the latter matrix is defined in
  \cref{eqn:krylov_popov_hermite}. The shape of the matrix \(D\) formed by the
  first \(d\) rows of \(\mathcal{K}_{d,\ash}(X,Y)\) as in the proof of
  \cref{lem:approx_generic_mindeg_exists} differs slightly according to the
  parity of \(d\).  For example for \(d = 6\) and \(d=7\), the matrix \(D\) is
  \[
    \begin{bmatrix}
      X_0 & X_1 & X_2 & X_3 & X_4 & X_5 \\
      Y_0 & Y_1 & Y_2 & Y_3 & Y_4 & Y_5 \\
      0   & X_0 & X_1 & X_2 & X_3 & X_4 \\
      0   & Y_0 & Y_1 & Y_2 & Y_3 & Y_4 \\
      0   &  0  & X_0 & X_1 & X_2 & X_3 \\
      0   &  0  & Y_0 & Y_1 & Y_2 & Y_3
    \end{bmatrix}
    \text{ if } d=6, \text{ and }
    \begin{bmatrix}
      X_0 & X_1 & X_2 & X_3 & X_4 & X_5 & X_6 \\
      Y_0 & Y_1 & Y_2 & Y_3 & Y_4 & Y_5 & Y_6 \\
      0   & X_0 & X_1 & X_2 & X_3 & X_4 & X_5 \\
      0   & Y_0 & Y_1 & Y_2 & Y_3 & Y_4 & Y_5 \\
      0   &  0  & X_0 & X_1 & X_2 & X_3 & X_4 \\
      0   &  0  & Y_0 & Y_1 & Y_2 & Y_3 & Y_4 \\
      0   & 0   &  0  & X_0 & X_1 & X_2 & X_3
    \end{bmatrix}
    \text{ if } d = 7.
  \]
  The polynomial \(\Delta_{d,\ash}\) is the determinant of such a matrix
  \(D\). Here, it does not seem straightforward that this determinant is
  nonzero; if it is nonzero, it is homogeneous of degree \(d\). Recall that in
  this case, \(\delta_{d,\ash}\) is either \((d/2,d/2)\) if \(d\) is even,
  or \((\lceil d/2 \rceil,\lfloor d/2 \rfloor)\) if \(d\) is odd, which can be
  related to the number of shifts of \(X\) and \(Y\), respectively, appearing
  in \(D\), and to the exponents \(\mu_0\) and \(\mu_1\) as defined in the
  proof of \cref{lem:approx_generic_mindeg_exists}.
  \qed
\end{example}

\begin{example}[weakly unbalanced shift]
  Fix \(d=6\) and \(\ash = -2\), and let \((\apo,\bpo) \in \pR_{<
  d}^2\). Following \cite[Sec.\,7.1]{JeannerodNeigerSchostVillard2017}, the
  Krylov matrix is
  \[
    \mathcal{K}_{d,\ash}(\apo,\bpo) =
    \begin{bmatrix}
      \apo_0 & \apo_1 & \apo_2 & \apo_3 & \apo_4 & \apo_5 \\
      0      & \apo_0 & \apo_1 & \apo_2 & \apo_3 & \apo_4 \\
      0      &  0     & \apo_0 & \apo_1 & \apo_2 & \apo_3 \\
      \bpo_0 & \bpo_1 & \bpo_2 & \bpo_3 & \bpo_4 & \bpo_5 \\
      0      &  0     &  0     & \apo_0 & \apo_1 & \apo_2 \\
      0      & \bpo_0 & \bpo_1 & \bpo_2 & \bpo_3 & \bpo_4 \\
      0      &  0     &  0     & 0      & \apo_0 & \apo_1 \\
      0      &  0     & \bpo_0 & \bpo_1 & \bpo_2 & \bpo_3 \\
      0      &  0     &  0     & 0      & 0      & \apo_0 \\
      0      &  0     &  0     & \bpo_0 & \bpo_1 & \bpo_2 \\
      0      &  0     &  0     & 0      & 0      & 0      \\
      0      &  0     &  0     & 0      & \bpo_0 & \bpo_1 \\
      0      &  0     &  0     & 0      & 0      & \bpo_0 \\
      0      &  0     &  0     & 0      & 0      & 0
    \end{bmatrix}
    \in \mKK{2(d+1)}{d}.
  \]
  Its topmost \(6 \times 6\) submatrix contains \(4\) (resp.\ \(2\)) rows with
  coefficients of \(\apo\) (resp.\ \(\bpo\)) and its shifts. The results in
  \cite[Sec.\,7.2]{JeannerodNeigerSchostVillard2017} imply that this submatrix
  is invertible if and only if the \(\ash\)-minimal degree of
  \(\appmod{d}{\apo,\bpo}\) is \((4,2)\). One easily verifies that \((4,2) =
  \delta_{6,-2}\). For \(\Delta_{d,\ash}\), one can take the determinant
  of the matrix \(D\), which is the topmost \(6 \times 6\) submatrix of the
  above matrix with \(\apo_i\) replaced by \(X_i\) and \(\bpo_i\) replaced by
  \(Y_i\). Here as well, it does not seem straightforward that this determinant
  is nonzero; if it is nonzero, it is homogeneous of degree \(d\).
  \qed
\end{example}
% pring = PolynomialRing(QQ, 12, 'a')
% X = list(pring.gens()[:6])
% Y = list(pring.gens()[6:])
% K = matrix([
%       X,
%       [0] + X[:-1],
%       [0,0] + X[:-2],
%       Y,
%       [0,0,0] + X[:-3],
%       [0] + Y[:-1]
%     ])
% K.det()

The same idea shows the existence of \(\Gamma_{\points,\ash}\) for interpolant
modules. Fix pairwise distinct points \(\points\in\field^d\) and tuples
\(\aevs,\bevs \in \field^{d}\). From a linear algebra viewpoint, interpolants
in \(\intmod{\points}{\aevs,\bevs}\) form \(\field\)-linear combinations that
are zero between the vectors \(\aevs\), \(\bevs\), \(\aevs \Omega\), \(\bevs \Omega\),
\(\aevs \Omega^2\), \(\bevs \Omega^2\), etc., where \(\Omega = \diag{\points}
\in \mKK{d}{d}\). Thus, restricting to interpolants of degree at most \(d\),
these can be found in the left nullspace of either of the matrices
\begin{equation*}
  \label{eqn:krylov_popov_hermite_interpolant}%
  % \mathcal{K} =
  \begin{bmatrix}
    \aev_0            & \aev_1            & \cdots         & \aev_{d-1}                    \\
    \bev_0            & \bev_1            & \cdots         & \bev_{d-1}                    \\
    \aev_0 \point_0   & \aev_1 \point_1   & \cdots         & \aev_{d-1} \point_{d-1}       \\
    \bev_0 \point_0   & \bev_1 \point_1   & \cdots         & \bev_{d-1} \point_{d-1}       \\
    \aev_0 \point_0^2 & \aev_1 \point_1^2 & \cdots         & \aev_{d-1} \point_{d-1}^2     \\
    \bev_0 \point_0^2 & \bev_1 \point_1^2 & \cdots         & \bev_{d-1} \point_{d-1}^2     \\
    \vdots            &  \vdots           &                &  \vdots                       \\
    \aev_0 \point_0^d & \aev_1 \point_1^d & \cdots         & \aev_{d-1} \point_{d-1}^d     \\
    \bev_0 \point_0^d & \bev_1 \point_1^d & \cdots         & \bev_{d-1} \point_{d-1}^d     \\
  \end{bmatrix}
  \quad\text{and}\quad
  % \overline{\mathcal{K}} =
  \begin{bmatrix}
    \aev_0            & \aev_1            & \cdots         & \aev_{d-1}                    \\
    \aev_0 \point_0   & \aev_1 \point_1   & \cdots         & \aev_{d-1} \point_{d-1}       \\
    \aev_0 \point_0^2 & \aev_1 \point_1^2 & \cdots         & \aev_{d-1} \point_{d-1}^2     \\
    \vdots            &  \vdots           &                &  \vdots                       \\
    \aev_0 \point_0^d & \aev_1 \point_1^d & \cdots         & \aev_{d-1} \point_{d-1}^d     \\
    \bev_0            & \bev_1            & \cdots         & \bev_{d-1}                    \\
    \bev_0 \point_0   & \bev_1 \point_1   & \cdots         & \bev_{d-1} \point_{d-1}       \\
    \bev_0 \point_0^2 & \bev_1 \point_1^2 & \cdots         & \bev_{d-1} \point_{d-1}^2     \\
    \vdots            &  \vdots           &                &  \vdots                       \\
    \bev_0 \point_0^d & \bev_1 \point_1^d & \cdots         & \bev_{d-1} \point_{d-1}^d     \\
  \end{bmatrix},
\end{equation*}
both in \(\mKK{2(d+1)}{d}\). More generally, considering the suitable
\(s\)-shifted Krylov matrix leads to the following result, which is analogous
to \cref{lem:approx_generic_mindeg_exists} and proved similarly.

\begin{lemma}
  \label{lem:interp_generic_mindeg_exists}%
  There exists a polynomial \(\Gamma_{\points,\ash}\) in \(2d\) variables
  \((X_0,\ldots,X_{d-1},Y_0,\ldots,Y_{d-1})\) over \(\field\), either zero or
  of total degree \(d\), such that for any pairwise distinct
  \(\points\in\field^d\) and any \(\aevs,\bevs \in \field^d\),
  \(\intmod{d}{\aevs,\bevs}\) has \(\ash\)-minimal
  degree \(\delta_{d,\ash}\) if and only if
  \(\Gamma_{\points,\ash}(\aev_0,\ldots,\aev_{d-1},\bev_0,\ldots,\bev_{d-1})
  \neq 0\).
\end{lemma}

% % For one-column wide figures use
% \begin{figure}
% % Use the relevant command to insert your figure file.
% % For example, with the graphicx package use
%   \includegraphics{example.eps}
% % figure caption is below the figure
% \caption{Please write your figure caption here}
% \label{fig:1}       % Give a unique label
% \end{figure}
% %
% % For two-column wide figures use
% \begin{figure*}
% % Use the relevant command to insert your figure file.
% % For example, with the graphicx package use
%   \includegraphics[width=0.75\textwidth]{example.eps}
% % figure caption is below the figure
% \caption{Please write your figure caption here}
% \label{fig:2}       % Give a unique label
% \end{figure*}
% %
% % For tables use
% \begin{table}
% % table caption is above the table
% \caption{Please write your table caption here}
% \label{tab:1}       % Give a unique label
% % For LaTeX tables use
% \begin{tabular}{lll}
% \hline\noalign{\smallskip}
% first & second & third  \\
% \noalign{\smallskip}\hline\noalign{\smallskip}
% number & number & number \\
% number & number & number \\
% \noalign{\smallskip}\hline
% \end{tabular}
% \end{table}

\end{document}